\documentclass[sigconf, nonacm]{acmart}

\newcommand\vldbdoi{XX.XX/XXX.XX}

\newcommand\vldbpagestyle{plain} 

\usepackage[utf8]{inputenc} % allow utf-8 input
\usepackage[T1]{fontenc}    % use 8-bit T1 fonts
\usepackage{hyperref}       % hyperlinks
\usepackage{url}            % simple URL typesetting
\usepackage{booktabs}       % professional-quality tables
\usepackage{amsfonts}       % blackboard math symbols
\usepackage{nicefrac}       % compact symbols for 1/2, etc.
\usepackage{microtype}      % microtypography
\usepackage{xcolor}         % colors

\usepackage[noend]{algpseudocode}
\usepackage[linesnumbered,ruled,vlined]{algorithm2e}
\usepackage{amsmath}
\usepackage{graphicx}
\usepackage{booktabs}
\usepackage{url}
\usepackage{hyperref}
\usepackage{xcolor}
\usepackage{subcaption}
\usepackage{multirow}
\usepackage{mathrsfs}
\usepackage{amsthm}
\usepackage{enumitem}
\usepackage{comment}
\usepackage{float}
\usepackage{wrapfig}

\newtheorem{fact}{Fact}

\newtheorem{observation}{Observation}

\newtheorem{definition}{Definition}
\newtheorem{theorem}{Theorem}
\newtheorem{lemma}{Lemma}
\newtheorem{condition}{Condition}
\newtheorem{corollary}{Corollary}

\newcommand{\authorsforshort}{\emph{et al.}}

\newcommand{\alvin}[1]{\textcolor{purple}{Alvin: #1}}

\newcommand{\junhao}[1]{\textcolor{blue}{Junhao: #1}}

\newcommand{\jz}[1]{\textcolor{orange}{JZ: #1}}

\newcommand{\algoagp}{{\em AGP-Static}}
\newcommand{\algosta}{{\em AGP-Static++}}
\newcommand{\algodyn}{{\em AGP-Dynamic}}
\newcommand{\algodpss}{{\em AGP-DPSS}}
\def\cso{c_{\text{so}}}
\def\pr{\text{Pr}}
\def\E{\text{E}}
\def\Var{\text{Var}}

\begin{document}
\title{Approximate Graph Propagation Revisited: \\Dynamic Parameterized Queries, Tighter Bounds and  Dynamic Updates}

%%
%% The "author" command and its associated commands are used to define the authors and their affiliations.
\author{Zhuowei Zhao}
\affiliation{%
  \institution{The University of Melbourne}
  %\streetaddress{}
  %\city{Parkville}
  %\state{VIC}
  %\country{Australia}
  %\postcode{3010}
  }
\email{zhuoweiz1@student.unimelb.edu.au}

\author{Zhuo Zhang}
\affiliation{%
  \institution{The University of Melbourne}
  %\streetaddress{}
  %\city{Parkville}
  %\state{VIC}
  %\country{Australia}
  %\postcode{3010}
  }
\email{zhuo.zhang@student.unimelb.edu.au}

\author{Hanzhi Wang}
\affiliation{%
  \institution{University of Copenhagen}
  %\streetaddress{}
  %\city{Copenhagen}
  %\state{}
  %\country{Denmark}
  %\postcode{}
  }
\email{hanzhi.hzwang@gmail.com}

\author{Junhao Gan}
\affiliation{%
  \institution{The University of Melbourne}
  %\streetaddress{}
  %\city{Parkville}
  %\state{VIC}
  %\country{Australia}
  %\postcode{3010}
  }
\email{junhao.gan@unimelb.edu.au}

\author{Zhifeng Bao}
\affiliation{%
  \institution{The University of  Queensland}
  %\streetaddress{124 La Trobe St}
  %\city{Melbourne}
  %\state{VIC}
  %\country{Australia}
  %\postcode{3000}
  }
\email{zhifeng.bao@uq.edu.au}

\author{Jianzhong Qi}
\affiliation{%
  \institution{The University of Melbourne}
  %\streetaddress{}
  %\city{Parkville}
  %\state{VIC}
  %\country{Australia}
  %\postcode{3010}
  }
\email{jianzhong.qi@unimelb.edu.au}

%%
%% The abstract is a short summary of the work to be presented in the
%% article.
\begin{abstract}
We revisit \emph{Approximate Graph Propagation} (AGP), a unified framework that captures various graph propagation tasks, such as PageRank, Personalized PageRank, feature propagation in Graph Neural Networks, and graph-based Retrieval-Augmented Generation. 
Our work focuses on the settings of \emph{dynamic graphs} and \emph{dynamic parameterized queries}, where the underlying graphs evolve over time (updated by edge insertions or deletions) and the input query parameters are specified on the fly to fit application needs. 
Our first contribution is an interesting observation that the SOTA solution, \emph{AGP-Static}, can be adapted to support dynamic parameterized queries; however, several challenges remain unresolved.
Firstly, the query time complexity of AGP-Static is based on an assumption of using an optimal algorithm for its subset sampling. Unfortunately, back to that time, such an algorithm did not exist; without such an optimal algorithm, an extra $O(\log^2 n)$ factor is required in the query complexity, where $n$ is the number of vertices in the graphs.
%Secondly, the theoretical analysis of \emph{AGP-Static} is not tight enough such that an unnecessary factor was introduced. 
Secondly, AGP-Static performs poorly on dynamic graphs, taking $O(n\log n)$ time to process each update.
To address these challenges, 
we propose a new algorithm, \emph{AGP-Static++}, which is simpler yet reduces roughly a factor of $O(\log^2 n)$ in the query complexity while preserving the approximation guarantees of AGP-Static.
However, AGP-Static++ still requires $O(n)$ time per update. 
%Besides, with the recent Dynamic Parameterized Subset Sampling (DPSS) technique, an extra $O(\log n)$ can be further removed, leading to overall a factor of $O(\log^3 n)$ improvement in theory.
%More importantly, with a more careful analysis, we prove that \emph{APG-static++} shaves a factor of $O(L)$ in the query time complexity, where $L = O(\log n)$ in most practical applications.
%Moreover, 
%to support dynamic graphs---where vertices and edges can be inserted or deleted---
To better support dynamic graphs, we further propose 
%a strengthened version of \emph{AGP-Static++}, called 
\algodyn, which achieves \emph{$O(1)$ amortized time per update}, significantly improving the aforementioned $O(n)$ per-update bound, while still preserving the query complexity and approximation guarantees.
%Moreover, our \algodyn~algorithm, in contrast, supports queries with parameters specified on the fly, 
%enabling users to efficiently perform various propagation tasks to fit their various real-world application needs, 
%without the need of building a dedicated data structure for every possible query parameter combination. 
%evaluate different parameters
%---corresponding to various kernel functions in real-world applications. 
%Second, with a more careful theoretical analysis, we derive a \emph{tighter bound} on the query time complexity, improving upon 
%the state-of-the-art randomized AGP algorithm 
%SOTA by a factor of $O(\log n)$, 
%where $n$ is the number of the vertices in the graph. 
%Third, we introduce an efficient \emph{randomized initialization} for certain types of queries, which eliminates 
%that maintains the same query complexity and error bounds, while 
%the $O(n)$ initialization cost from the query complexity without affecting the approximation guarantees. 
%to an expected $O\left(\min\left(\frac{\log{\frac{1}{\delta}}}{\delta}, n\right)\right)$, where $\delta$ is an approximation parameter, 
%while maintaining the query complexity and error guarantees.
%in applicable scenarios, 
%where $\delta$ is the approximation parameter. 
Last, our comprehensive experiments validate the theoretical improvements: 
compared to the baselines, our algorithm achieves speedups of up to $177\times$ on %average 
update and $10\times$ on query efficiency.
%, and $11\times$ on initialization cost.

\end{abstract}

\maketitle

%%% do not modify the following VLDB block %%
%%% VLDB block start %%%
\pagestyle{\vldbpagestyle}
% \begingroup\small\noindent\raggedright\textbf{PVLDB Reference Format:}\\
% \vldbauthors. \vldbtitle. PVLDB, \vldbvolume(\vldbissue): \vldbpages, \vldbyear.\\
% \href{https://doi.org/\vldbdoi}{doi:\vldbdoi}
% \endgroup
% \begingroup
% \renewcommand\thefootnote{}\footnote{\noindent
% This work is licensed under the Creative Commons BY-NC-ND 4.0 International License. Visit \url{https://creativecommons.org/licenses/by-nc-nd/4.0/} to view a copy of this license. For any use beyond those covered by this license, obtain permission by emailing \href{mailto:info@vldb.org}{info@vldb.org}. Copyright is held by the owner/author(s). Publication rights licensed to the VLDB Endowment. \\
% \raggedright Proceedings of the VLDB Endowment, Vol. \vldbvolume, No. \vldbissue\ %
% ISSN 2150-8097. \\
% \href{https://doi.org/\vldbdoi}{doi:\vldbdoi} \\
% }\addtocounter{footnote}{-1}\endgroup
% %%% VLDB block end %%%

% %%% do not modify the following VLDB block %%
% %%% VLDB block start %%%
% \ifdefempty{\vldbavailabilityurl}{}{
% \vspace{.3cm}
% \begingroup\small\noindent\raggedright\textbf{PVLDB Artifact Availability:}\\
% The source code, data, and/or other artifacts have been made available at \url{https://github.com/alvinzhaowei/AGP-dynamic}.
% \endgroup
% }
%%% VLDB block end %%%

\section{Introduction}
Node proximity evaluations, 
such as PageRank~\cite{chung2007heat,jung2017bepi} and Personalized PageRank~\cite{Page1999ThePC},
have shown great usefulness in various graph mining and machine learning tasks, including but not limited to  
%are critical in graph mining and machine learning tasks.
%especially in Graph Neural Networks (GNNs). 
%These queries
%such as PageRank~\cite{chung2007heat,jung2017bepi, Page1999ThePC}, 
%play a crucial role in tasks like 
link prediction~\cite{Page1999ThePC}, graph representation learning~\cite{wu2019simplifying}, spam detection~\cite{gyongyi2004combating}, and web search~\cite{Page1999ThePC}. 
Recent studies have integrated node proximity with Graph Neural Networks (GNNs) to enhance scalability, simplifying the original GNN architecture by replacing feature transformation across multiple layers with a feature propagation step~\cite{wu2019simplifying}. 
Depending on concrete applications, 
%these approaches focus on 
different node proximity models are deployed.  
For example, Personalized PageRank~\cite{Page1999ThePC} is used in PPRGo~\cite{bojchevski2020scaling}, while heat kernel PageRank~\cite{foster2001faster} is adopted in \emph{graph diffusion convolution} (GDC)~\cite{gasteiger2019diffusion}. 
%As a result, each model often requires individual study. 
%
Recently, Wang~\authorsforshort~\cite{wang2021approximate}
unify a wide range of node proximity models into the following {\em graph propagation} equation:
%in the following form:
\vspace{-1mm}
\begin{equation}
\label{eq:1}
\boldsymbol{\pi} = \sum^{\infty}_{i = 0}w_i \cdot (\mathbf{D^{-a}}\mathbf{A}\mathbf{D^{-b}})^i\cdot \mathbf{x}\,,
\end{equation}
where $\boldsymbol{\pi}$ is the node proximity evaluation result vector, $\mathbf{A}$ and $\mathbf{D}$ are the adjacency matrix and the diagonal degree matrix of the underlying undirected graph $G = \langle V, E\rangle$,  respectively, while $a$, $b$, $w_i$ and $\mathbf{x}$ are query parameters 
specifying the node proximity model.
As shown in Table~\ref{tab:example}, with different parameter settings,  
Equation~\eqref{eq:1} captures a range of graph propagation applications.
%as shown in Table~\ref{tab:example}. 
%For example, when $a = 1$, $b = 0$, $w_i = \alpha (1-\alpha)^i$, and $\mathbf{x}$ is a one-hot vector, solving Equation~\eqref{eq:1} reduces to computing a {\em single-target Personalized PageRank} defined under $\alpha$-random walks~\cite{Page1999ThePC}. When $a = \frac{1}{2}$, $b = \frac{1}{2}$, $w_i = e^{-t}\cdot {\frac{t^i}{i!}}$ (Poisson distribution), and $\mathbf{x}$ is the graph signal, Equation~\eqref{eq:1} models the GDC~\cite{gasteiger2019diffusion}.
%\junhao{Add one or two more examples here.} \alvin{I have added table 1 for examples.}
% \vspace{-2mm}

%\junhao{I am not sure why we need to introduce the details of AGP here.}

%\junhao{we can just mention AGP is the sota approximate algorithm for the static case, and then motivate dynamic updates and parameters on the fly.}

%\alvin{It is now shortened and updated.}

%By proposing an efficient unified algorithm to solve Equation~\eqref{eq:1} 
%it immediately addresses all aforementioned applications.
%However, 
In most cases, computing the exact solution to Equation~\eqref{eq:1} can be expensive. 
To achieve better efficiency, 
Wang~\authorsforshort~\cite{wang2021approximate} propose 
the notion of $(\delta, c)$-approximation,
%was proposed in~\cite{wang2021approximate}, 
which only aims to provide approximations within a $c$-relative error to entries where $\boldsymbol{\pi}(v) \geq \delta$, while ignoring the errors for entries with small values, i.e., $\boldsymbol{\pi}(v) < \delta$.
%They then design 
They further proposed a state-of-the-art algorithm, called 
%an approximate algorithm called 
{\em Approximate Graph Propagation} (denoted as \algoagp), for computing $(\delta, c)$-approximate results for queries on {\em static graphs with fixed and pre-specified parameters $w_i$'s, $a$ and $b$}. 
%The main idea is to introduce a threshold parameter $\delta$, and only approximate for those nodes $u$ with $\boldsymbol{\pi}[u] > \delta$ while ignoring the others. While \algoagp~represents the state-of-the-art in static scenarios, it does not support dynamic graph updates and on-the-fly parameters $a$ and $b$. 
%As we will discuss in the later sections,
%\algoagp~indeed can support query with parameters given on the fly and support each graph update (either an edge insertion or deletion) in $O(n)$ time, where $n$ is the number of vertices in $G$. 

\vspace{2mm}
\noindent
{\bf SOTA's Limitation 1: the Static Setting.}
Unfortunately, such {\em static} setting of \algoagp~with static graphs and pre-specified query parameters 
means that \algoagp~can only work on a fixed and pre-specified 
node proximity model on a static graph, which
indeed has  
%Such an expensive update performance 
significantly 
limited its applicability in real-world applications, 
especially, when   
the underlying graph $G$ is frequently updated and the support of varying node proximity evaluations is required, as illustrated in the case study below.

\begin{table}
\centering
\caption{Graph Propagation Applications}
\vspace{-2mm}
\setlength{\tabcolsep}{2pt}
\resizebox{0.48\textwidth}{!}{
\begin{tabular}{c|c|c|c|c}
\hline
\textbf{Algorithm}             & \rule{0pt}{12pt} $a$    \rule[-5pt]{0pt}{0pt}         & $b$             & $w_i$ & $\mathbf{x}$                               \\ \toprule
L-hop transition probability              & $0$             & $1$             &  \rule{0pt}{12pt}$w_i = 0 (i \neq L), W_L = 1$    \rule[-5pt]{0pt}{0pt} &   one-hot vector                                \\ \hline
PageRank~\cite{Page1999ThePC}              & $0$             & $1$             &  $\alpha (1-\alpha)^i$     &      \rule{0pt}{12pt} $\left[ \frac{1}{n}, \frac{1}{n}, \dots, \frac{1}{n} \right]^T$ \rule[-5pt]{0pt}{0pt}                            \\ \hline

Personalized PageRank~\cite{Page1999ThePC} & $0$             & $1$             &  \rule{0pt}{12pt}$\alpha (1-\alpha)^i$   \rule[-5pt]{0pt}{0pt}   & one-hot vector \\ \hline
Single-target PPR~\cite{lofgren2013personalized}     & $1$             & $0$             & \rule{0pt}{12pt}  $\alpha (1-\alpha)^i$ \rule[-5pt]{0pt}{0pt}   & one-hot vector                    \\ \hline
Heat kernel PageRank~\cite{chung2007heat}  & $0$             & $1$             &  \rule{0pt}{12pt} $e^{-t}\cdot {\frac{t^i}{i!}}$  \rule[-5pt]{0pt}{0pt}  & one-hot vector                    \\ \hline
% Katz index~\cite{katz1953new}                  & $0$             & $0$             &    \rule{0pt}{12pt} $\beta^i$ \rule[-5pt]{0pt}{0pt}  & one-hot vector                    \\ \hline
Simplifying graph convolutional networks~\cite{wu2019simplifying}                   & \rule{0pt}{12pt} $\frac{1}{2}$ \rule[-5pt]{0pt}{0pt} & $\frac{1}{2}$ &  $w_i = 0 (i \neq L), W_L = 1$     & the graph signal                  \\ \hline
Approx-personalized propagation ~\cite{gasteiger2018predict}                 & \rule{0pt}{12pt} $\frac{1}{2}$ \rule[-5pt]{0pt}{0pt}& $\frac{1}{2}$ &   $\alpha (1-\alpha)^i$    & the graph signal                  \\ \hline
Graph diffusion convolution~\cite{gasteiger2019diffusion}                   & \rule{0pt}{12pt} $\frac{1}{2}$ \rule[-5pt]{0pt}{0pt} & $\frac{1}{2}$ &  $e^{-t}\cdot {\frac{t^i}{i!}}$     & the graph signal                  \\ \bottomrule

\end{tabular}
}
\vspace{-2mm}
\label{tab:example}
\end{table}

\vspace{2mm}
\noindent
\underline{\em Case Study: A Fully Dynamic Scenario.}
In recent years, Large Language Models (LLMs) such as DeepSeek~\cite{liu2024deepseek}, LLaMA~\cite{grattafiori2024llama}, and GPT~\cite{achiam2023gpt} have 
achieved great success in various applications and have 
become a trending topic in AI research~\cite{li2023large, chang2024survey, kasneci2023chatgpt, thirunavukarasu2023large}.
%there has been growing interest in 
Along the process, Retrieval-Augmented Generation (RAG), a technique to enhance the trustworthiness of LLMs through the integration of external knowledge~\cite{fan2024survey, lewis2020retrieval, siriwardhana2023improving}, has attracted increasing interest. 
A notable trend is the use of graph-based RAG~\cite{edge2024local, gutierrez2024hipporag, huang2025ket, wang2024knowledge, wu2024medical}, which leverages graph-structured data, namely, {\em knowledge graphs}, as an external knowledge source and has shown promising results. 
Specifically, a knowledge graph not only stores individual entities but also captures the relationships between them through edges. 
When a user asks a question to an LLM, graph-based RAG first maps the question keywords to vertices in the knowledge graph.
Starting from these vertices as sources, the RAG process retrieves vertices that have high relevance to these sources.

For example, when a question ``What causes Alzheimer’s disease?'' is submitted, the AI system will identify one or more vertices in the knowledge graph that best match the query keywords (e.g., the vertices related to ``Alzheimer's disease'').
It then attempts to retrieve the most relevant vertex candidate set from the knowledge graph, measured under a certain node proximity model, e.g., 
Personalized PageRank (PPR), which computes the probability $\pi[v]$ of a random walk from a given source vertex $s$ to each vertex $v$ in the graph.  
 %and this process can be implemented via graph propagation. 
%For instance, methods such as Personalized PageRank (PPR) perform a random walk starting from the query node $q$ and output the probability distribution of stopping at each vertex.
%The output, denoted as $\pi$, 
%Such a probability distribution $\pi$ 
%indicates the relevance of each vertex to $q$ — 
A
larger $\pi[v]$ value indicates stronger relevance of $v$ to $s$. 
The PPR values can thus be used to select the top-$k$ most relevant vertices (e.g., ``beta-amyloid plaques'', ``tau protein'' or ``neuro degeneration'' for the Alzheimer’s disease example). 
The retrieved vertices (and possibly their attributes or neighbors) are then compiled into a context document and fed into an LLM to generate an answer based on both the question and the graph-derived context.

This kind of graph-based RAG introduces new challenges for graph propagation. 
First, the knowledge graphs can be updated frequently requiring efficient update support.
%, making update efficiency in dynamic scenarios crucial.
%Second, in the RAG framework, it is essential to fine-tune and evaluate the retriever with different parameters. 
Second, different questions may require different types of node proximity evaluations. 
For example, when a user asks, \emph{``Who are the people most affected by Donald Trump?''}, a single-source PPR can be used to model Trump's influence. 
In contrast, for a question like \emph{``Who are the people that most affect Trump?''}, a single-target PPR is needed. 
For question \emph{``What is the importance of Trump?''}, it can be modeled as computing the PageRank of Trump.
Therefore, 
%depending on the questions, different node proximity evaluations are needed, and hence, 
%it requires 
an ideal graph propagation algorithm should be able to support queries (with respect to Equation~\eqref{eq:1}) with parameters $w_i$'s, $a$ and $b$ given on the fly, such that the node proximity model can be specified in queries.
% Graph-based Retrieval-Augmented Generation (RAG) has gained attention as a method to improve the trustworthiness of LLMs by integrating external knowledge from graph-structured data~\cite{edge2024local, gutierrez2024hipporag, huang2025ket, wang2024knowledge, wu2024medical}. Upon a query's arrival, the system retrieves relevant vertices from a knowledge graph, which can be processed as graph propagation based on a node proximity score $\boldsymbol{\pi}$—and feeds the top-scoring vertices into the LLM.
% However, graph-based RAG poses new challenges for graph propagation: (1) frequent graph updates demand efficient dynamic algorithms; (2) retrievers must be tuned and evaluated under varying parameters; (3) different queries may require different proximity measures. More details are provided in Appendix~\ref{ap:sc}. 

%<<<<<<< HEAD
\vspace{2mm}
\noindent
{\bf SOTA's Limitation 2: Looseness on Theoretical Bounds.}
Besides the static setting, 
\algoagp~also suffers from several limitations in their theoretical analysis.
First, in~\cite{wang2021approximate}, the running time analysis of \algoagp~is based on an assumption that 
an optimal algorithm for solving a {\em special parameterized subset sampling problem}
is adopted.
However, such an algorithm did not 
exist back to the time when ~\cite{wang2021approximate} was published until a very recent work~\cite{gan2024optimal}.
Indeed, the actual implementation of \algoagp~adopts a compromising subset sampling algorithm 
%compromising in the theoretical bound, 
which leads to a blow-up factor of $O(\log^2 n)$ in the claimed query time complexity,
where $n$ is the number of vertices in the graph.
Here, a $O(\log n)$ factor in this blow up 
comes from the usage of power-of-two bucketing technique, while the other $O(\log n)$ factor is introduced by the need of generating binomial random variates in their subset sampling algorithm, where according to the best known result~\cite{farach2015exact}, generating each binomial random variate takes $O(\log n)$ expected time\footnote{Whether 
a binomial random variate can be generated in  
$O(1)$ expected time still remains an open problem~\cite{garcia2022binomial,kuhl2017history,kachitvichyanukul1988binomial}.}.

To address the above challenges and limitations, we made the following technical contributions:
%\vspace{-2mm}
\begin{itemize}[leftmargin = *]
\item {\bf A Strengthened Version of \algoagp.} We observe that \algoagp~indeed can support dynamic queries with parameters given on the fly. We strengthen the algorithm to support each graph update (either an edge insertion or deletion) in $O(d_{\max}\log n) = O(n\log n)$ time,
%, improving the naive $O(n \log n)$ per-update time
where $d_{\max}$ and $n$ are the maximum degree and the number of vertices in the current graph $G$.  
%With a more careful analysis, 
We further show a precise expected query time complexity of \algoagp, $O(\log^2 n \cdot \frac{L^2}{\delta}\cdot Z + n)$, where $L$ is the number of iterations which is often equal to $O(\log n)$, and $Z$ is an output-size sensitive term. 
% \vspace{-2mm}
%\item {\bf A More Accurate Query Complexity of \algoagp}. 
\item {\bf Our Solutions for Enhanced Efficiency.} 
We propose \algosta~which adopts a simpler subset sampling algorithm yet is more efficient in practice. With a careful analysis, we prove that 
the expected query time complexity of \algosta~is bounded by  
$O(\log n \cdot \frac{L}{\delta}\cdot Z + n)$, a factor of $O(L \log n) \approx O(\log^2 n)$ improvement over \algoagp~since the first term often dominates the second in the complexity.
Moreover, \algosta~also improves over \algoagp~in terms of update efficiency, though it still  
%\algosta~still 
takes $O(d_{\max}) = O(n)$ time to process each graph update. 

%By applying a recent Dynamic Parameterized Subset Sampling (DPSS) technique~\cite{gan2024optimal}, our \algodpss~further improves the expected query time complexity to $O(\frac{L}{\delta}\cdot Z + n)$ which is proven to be optimal~\cite{wang2021approximate} and leads to in total a $O(L \log^2 n)$, roughly $O(\log^3 n)$, factor improvement over \algoagp. 
%Both \algosta~and \algodpss~require a cost of $O(d_{\max}) = O(n)$ per update. 
%bounded by $O(n)$.
%the same 
%update time complexity 
%as that of \algoagp.

\item {\bf Our Fully-Dynamic Algorithm.} 
To support efficient graph updates, we propose our ultimate algorithm, \algodyn, which 
not only supports each update in $O(1)$ {\em amortized} time, but also achieves 
the same expected query complexity and preserves the same approximation guarantees as \algosta.

\item {\bf Extensive Experimental Study.}
We conduct experiments on nine real-world datasets. 
The experiment results show that our \algodyn~outperforms 
the SOTA \algoagp~ 
%the baseline methods by up to 11 times on initialization efficiency, 
by up to $117$ times on update efficiency with different update patterns and insertion-deletion ratios, and by up to $10$ times on query efficiency with small variance on arbitrary input parameters. 
%In addition, for each dataset and across different combinations of $a$ and $b$, the query time remains within the same order of magnitude, with a standard deviation less than 33\% of the average.
\end{itemize}

\section{Preliminaries}\label{sec:pre}

\subsection{Problem Definition}

%Next, we formally define the problem.
Consider an undirected and unweighted graph $G = \langle V, E\rangle$, where $V$ is a set of $n$ vertices and $E$ is a set of $m$ edges. The \emph{neighborhood} of a vertex $u$, denoted as $N[u] = \{ v \in V \mid (u,v) \in E \}$, is the set of all vertices adjacent to $u$. The \emph{degree} of $u$, denoted by $d_u = |N[u]|$, represents the number of neighbors of $u$.
The adjacency matrix and diagonal degree matrix of $G$ are denoted as
$\mathbf{A}$ and $\mathbf{D}$.
~~~We aim to solve the following {\bf Unified Graph Propagation Equation} (Equation~\eqref{eq:1}):
%\junhao{we then have two Equation (1) in Introduction.} \alvin{I refer to Equation 1 instead.}
\vspace{-.4em}
\begin{equation*}
     \boldsymbol{\pi} = \sum^{\infty}_{i = 0}w_i \cdot (\mathbf{D^{-a}}\mathbf{A}\mathbf{D^{-b}})^i\cdot \mathbf{x}\,,
     %\tag{\ref{eq:1}}
 \end{equation*}
where 
the query parameters $a$, $b$, $w_i$ and $\mathbf{x}$ satisfy the following conditions:
\begin{itemize}[leftmargin = *]
\item $a, b \in [0,1]$ are {\em constants} such that $a + b\geq 1$;
%or $a = b = 0$;
\item the weight parameters $w_i$'s satisfy:
\begin{itemize}
\item $w_i \geq 0$ for all $i$, and $\sum_{i=0}^\infty w_i = 1$ 
(
%this is without loss of generality, because 
this can be achieved by normalizing $\sum_{i=0}^\infty w_i$ to $1$); 
%\item there exit a constant integer $L_0 \geq 0$ and a constant value $\lambda$ such that 
\item $w_i \leq \lambda^i$ for all $i \geq L_0$, where $L_0 \geq 0$ and $\lambda \geq 1$ are some constants;
\item $w_i$'s are input as an {\em oracle} $\mathcal{O}_w$ with which these information can be returned in $O(1)$ time: (i) the value of $w_i$ for all integer $i\geq 0$, and (ii) the value of $L_{\mathcal{O}_w}(\Delta)$ which is the smallest integer $k$ such that $\sum_{i = k}^\infty w_i \leq \Delta$ for the given $0 \leq \Delta \leq 1$.
%and (ii) the smallest integer $L \geq 0$ such tat $\sum_{i = L+1}^{\infty} w_i \leq \delta$ 
%can be returned in $O(1)$ time;
%, and (ii) an .   
\end{itemize}
\item $\mathbf{x}$ is an $n$-dimensional vector such that $\|\mathbf{x}\|_1 = 1$ and $\mathbf{x}_i \geq 0, \forall~i = 1,2,3,\cdots, n $.  
\end{itemize}
It can be verified that all the node proximity models in Table~\ref{tab:example} satisfy the above conditions, optionally after normalization.
In what follows, the node proximity model parameterized by
$a$, $b$, $\mathcal{O}_w$ and $\mathbf{x}$ is called a {\em query} and denoted by
 $q(a, b, \mathcal{O}_w, \mathbf{x})$.

\vspace{2mm}
\noindent{\bf $(\delta, c)$-Approximation}. Given a threshold $\delta >0$ and a constant $c > 0$, 
%an 
%an approximate graph propagation is to return 
an estimation of $\boldsymbol{\pi}$, denoted by $\boldsymbol{\hat{\pi}}$, 
is an $(\delta,c)$-approximation of $\boldsymbol{\pi}$ if it satisfies:
%that satisfies the  following with at least a constant probability:
$$\text{for all } v \in V \text{  with  }|\boldsymbol{\pi}(v)| > \delta
,\;\;\;
|\boldsymbol{\pi}(v) - \boldsymbol{\hat{\pi}}(v)| \leq c \cdot \boldsymbol{\pi}(v) 
\,.$$
% \forall v \in V \wedge |\boldsymbol{\pi}(v)| > \delta \,.$$
%where $c$ is a constant.
%
We have the following fact:
\begin{fact}\label{fact:delta}%{\emph{[*]}} 
   % An $(\delta, c)$-Approximation $\boldsymbol{\hat{\pi}}$ of $\boldsymbol{\pi}$ can be computed by 
For every query $q(a, b, \mathcal{O}_w, \mathbf{x})$, 
%let $L = L_{\mathcal{O}_w}(c\cdot \delta) \in O(\log_{\lambda} \frac{1}{c\cdot \delta}) = O(\log \frac{1}{\delta})$,
$\mathbf{\hat{\pi}} = \sum^{L}_{i = 0}w_i \cdot (\mathbf{D^{-a}}\mathbf{A}\mathbf{D^{-b}})^i\cdot~\mathbf{x}$
is an $(\delta, c)$-approximation of $\boldsymbol{\pi}$, where
%$L = \log_{\lambda} \frac{1}{\delta} \in O(\log \frac{1}{\delta})$.
$L = L_{\mathcal{O}_w}(c\cdot \delta) \in O(\log_{\lambda} \frac{1}{c\cdot \delta}) = O(\log \frac{1}{\delta})$.
\end{fact} 

\begin{proof}
By the fact that $\mathbf{AD}^{-1}$ has an eigenvalue of $1$, we have:
$\sum^{\infty}_{i = L + 1}w_i \cdot (\mathbf{D^{-a}}\mathbf{A}\mathbf{D^{-b}})^i\cdot \mathbf{x}
 = 
\sum^{\infty}_{i = L + 1} w_i \cdot (\mathbf{A}\mathbf{D}^{-(a+b)})^{i} \cdot \mathbf{x}
\leq 
\sum^{\infty}_{i = L + 1} w_i \cdot (\mathbf{A}\mathbf{D}^{-1})^{i} \cdot \mathbf{x}
\leq 
\sum^{\infty}_{i = L + 1} w_i \cdot \mathbf{x}
\leq \sum^{\infty}_{i = L + 1} w_i \cdot \|\mathbf{x}\|_1 \leq c \cdot \delta$.
\end{proof}

%\begin{definition}
%[Dynamic Approximation Graph Propagation]
\noindent{\bf Problem Definition.}
Given an undirected graph $G = \langle V, E \rangle$,
%that can be updated by edge insertions or deletions, 
the problem of {\emph{Dynamic Approximate Graph Propagation}} ({DAGP}) is to design an algorithm which:
\begin{itemize}[leftmargin = *]
\item for every query $q(a, b, \mathcal{O}_w, \mathbf{x})$, approximation parameters $\delta> 0$ and $c > 0$, 
with at least a constant probability, 
returns an $(\delta, c)$-approximation $\boldsymbol{\hat{\pi}}$, and 
\item supports each graph update (i.e., either an edge insertion or deletion) efficiently. 
\end{itemize}

\subsection{The State-Of-The-Art Algorithm}
\label{sec:sota}

\begin{comment}
It can be verified that for every $q(a, b, \mathcal{O}_w, \mathbf{x})$, Equation~\eqref{eq:1} converges. 
As a result, Equation~\eqref{eq:1} can be solved by the standard Power Method~\cite{trefethen2022numerical}.
However, the number of iterations to obtain a precise $\boldsymbol{\pi}$ can be large.
%Thanks to the $(\delta, c)$-approximation, 
Wang~\authorsforshort~\cite{wang2021approximate}
proved the following fact: 
\begin{fact}[\cite{wang2021approximate}]\label{fact:delta}
    An $(\delta, c)$-Approximation $\boldsymbol{\hat{\pi}}$ of $\boldsymbol{\pi}$ can be computed by $\mathbf{\hat{\pi}} = \sum^{L}_{i = 0}w_i \cdot (\mathbf{D^{-a}}\mathbf{A}\mathbf{D^{-b}})^i\cdot \mathbf{x}$ with $L = \log_{\lambda} \frac{1}{\delta} \in O(\log \frac{1}{\delta})$.
\end{fact} 
\end{comment}

According to Fact~\ref{fact:delta},
one can compute an $(\delta, c)$-approximation $\boldsymbol{\hat{\pi}}$ by running the Power Method~\cite{trefethen2022numerical} for $L$ iterations, which takes $O(m \cdot L)$ time, where $m$ is the number of edges in $G$.
To improve the query efficiency,
Wang~\authorsforshort~\cite{wang2021approximate} proposed the SOTA algorithm, \algoagp, which simply pushes the ``probability mass'' from some vertices $u$ to a random subset of their neighbors in each iteration. 
The pseudocode of \algoagp\ is in Algorithm~\ref{alg:framework}.

\vspace{-2mm}
\begin{small}
\begin{algorithm}
\caption{\emph{\algoagp}}\label{alg:framework}
\DontPrintSemicolon
\SetKwComment{Comment}{/* }{ */}
% \SetKwBlock{preprocess}{Preprocess Procedure:}{end}
% \SetKwBlock{query}{Query Procedure:}{end}

% \preprocess{
%     \KwIn{$G = \langle V, E \rangle$}
%     \KwOut{Sorted neighborhood and Data structure $B$}
%     \For{each $u \in V$}
%         {
%             Sort $N[u]$ into a sorted linked list by the degree of each neighbor in descending order\;
%             Group neighbors into power-of-2 buckets by their degree into $B[u]$\;
%         }
%     \Return $B$ and sorted neighborhood
% }

% \query{
\KwIn{$q(a, b, \mathcal{O}_w, \mathbf{x})$, parameters $\delta$ and $c$ for $(\delta,c)$-approximation
%, number of steps $L$ computed based on $\mathcal{O}_w$, data structure $B$, and sorted neighborhood} 
}
\KwOut{$\boldsymbol{\hat{\pi}}$, an $(\delta, c)$-approximation of $\boldsymbol{\pi}$}

$L \leftarrow L_{\mathcal{O}_w}(c \cdot \delta)$;
$\varepsilon \leftarrow \frac{c^2 \delta}{2 L^2}$;\;
$\mathbf{\hat{r}}^{(0)} \leftarrow \mathbf{x}$;\; 
\For{$i = 0$ to $L -1 $}
{
    \For{each $u \in V$ with $\mathbf{\hat{r}}^{(i)}(u) \neq 0$}
    {   
        Scan the sorted neighborhood\;
        \For{each $v \in N[u]$ and $d_v \leq \left(\frac{1}{\varepsilon} \cdot \frac{Y_{i+1}}{Y_i} \cdot \frac{\mathbf{\hat{r}^{(i)}(u)}}{d_u^b} \right)^{\frac{1}{a}}$}
        {
            $\mathbf{\hat{r}}^{(i+1)}(v) \leftarrow \mathbf{\hat{r}}^{(i+1)}(v) + \frac{Y_{i+1}}{Y_i}\cdot \frac{\mathbf{\hat{r}}^{(i)}(u)}{d_v^a\cdot d_u^b};$ \;
        
        }
        
        Sample each remaining neighbor $v \in N[u]$ with probability $p_{u,v} = \frac{1}{\varepsilon} \cdot \frac{Y_{i+1}}{Y_i} \cdot \frac{\mathbf{\hat{r}^{(i)}(u)}}{d_v^ad_u^b}$\;
        \For{each $v$ sampled}{
            $\mathbf{\hat{r}}^{(i+1)}(v) \leftarrow \mathbf{\hat{r}}^{(i+1)}(v) + \varepsilon;$ \;
        }
        $\mathbf{\hat{q}}^{(i)}(u) \leftarrow \mathbf{\hat{q}}^{(i)}(u) + \frac{w_i}{Y_i}\cdot \mathbf{\hat{r}}^{(i)}(u);$ 
    }
    $\boldsymbol{\hat{\pi}} \leftarrow \boldsymbol{\hat{\pi}} + \mathbf{\hat{q}}^{(i)};$\;
}
$\mathbf{\hat{q}}^{(L)} = \frac{w_L}{Y_L} \cdot \mathbf{\hat{r}}^{(L)};$ and
$\boldsymbol{\hat{\pi}} \leftarrow \boldsymbol{\hat{\pi}} + \mathbf{\hat{q}}^{(L)};$\;
\Return $\boldsymbol{\hat{\pi}}$
% }

\end{algorithm}
\end{small}

\begin{comment}
%To achieve this, 
Specifically, \algoagp\ pre-processes the underlying graph $G$ as follows.
For each vertex $u \in V$, 
\begin{itemize}[leftmargin = *]
%\item for each $u \in V$,
%\begin{itemize}
\item sort all the neighbors $v \in N[u]$ by $d_v$, the degree of $v$;
\item group $v \in N[u]$ into a sorted list of power-of-two buckets by $d_v$ 
such that all the neighbors $v$ with $2^{i-1} \leq  d_v < 2^i$ are put into the bucket of index $i$, denoted by $\mathcal{B}_u(i)$;
denote the resulting sorted list of {\em non-empty} buckets by $\mathcal{B}_u$;
% and the bucket of index $i$ in $\mathcal{B}_u$ by $\mathcal{B}_u(i)$; 
%\end{itemize} 
\end{itemize}
\end{comment}

%To obtain an $(\delta, c)$-approximation $\boldsymbol{\hat{\pi}}$, 
Given a query $q(a, b, \mathcal{O}_w, \mathbf{x})$,
\algoagp\ introduces the notions of {\em residue} and {\em reserve vector} at the $i^{\text{th}}$ iteration (where $i = 0 , 1, \ldots, L$) to
facilitate iterative computation:
\begin{itemize}
    \item {\em residue} $\mathbf{r^{(i)}}= Y_i\cdot (\mathbf{D^{-a}}\mathbf{A}\mathbf{D^{-b}})^i\cdot \mathbf{x}$, where $Y_i = \sum^\infty_{k=i}w_k$; 
    \item {\em reserve vector} $\mathbf{q^{(i)}} = \frac{w_i}{Y_i}\cdot \mathbf{r^{(i)}}= w_i\cdot (\mathbf{D^{-a}}\mathbf{A}\mathbf{D}^{-b})^i\cdot \mathbf{x}$.
\end{itemize}
%{\em Residue} can be computed recursively step by step as 
%According to the above definition,
Moreover, the residue $\mathbf{r}^{(i+1)}$ in the next step (i.e., iteration) can be computed from the previous step by: 
%can be computed step by step as
$$\mathbf{r}^{(i+1)}= \frac{Y_{i+1}}{Y_i}\cdot (\mathbf{D^{-a}}\mathbf{A}\mathbf{D}^{-b})\cdot \mathbf{r}^{(i)}\,,$$ 
and so does {reserve vector} $\mathbf{q}^{(i)}$. % 
%By summing up the reserve vectors $\mathbf{q}^{(i)}$, 
An $(\delta,c)$-approximation $\boldsymbol{\hat{\pi}}$ can thus be obtained by 
%summing up the reserve vectors $\mathbf{q}^{(i)}$:
%\vspace{-4mm}
%\begin{small}
%\begin{align*}
$\boldsymbol{\hat{\pi}} = \sum^{L}_{i = 0}\mathbf{q}^{(i)} = \sum^{L}_{i = 0}w_i \cdot (\mathbf{D^{-a}}\mathbf{A}\mathbf{D^{-b}})^i\cdot \mathbf{x}$.
%\end{align*}
%\end{small}
%\vspace{-2mm}

Before answering any query, 
\algoagp\ needs to construct two data structures on the neighborhood $N[u]$ for 
each vertex $u \in V$:
(i) a sorted list of $N[u]$ by their degrees, and
(ii) a specific data structure (which is introduced below) to support efficient {\em subset sampling} from $N[u]$.
%$u$'s neighbors.
%a sorted list of non-empty {\em power-of-two} buckets, denoted by $B[u]$, such that  
To answer a query, \algoagp\  performs $L$ iterations; 
in the $i^{\text{th}}$ iteration, for each vertex $v$ in the neighborhood $N[u]$:
%in the $i^{\text{th}}$ step:
\begin{itemize}[leftmargin = *]
    \item by scanning the sorted neighbor list: if $d_v$ is small (Line 5)
    , propagate the exact amount of probability mass with respect to $\mathbf{\hat{r}}^{(i)}(u)$ to $\mathbf{\hat{r}}^{(i+1)}(v)$ (Line 6);
    \item by the subset sampling structure:  sample each remaining neighbor $v$ with a certain probability (Line 7)
    , and if $v$ is sampled, a probability mass of $\varepsilon = O(\frac{\delta}{L^2})$ is propagated to $\mathbf{\hat{r}}^{(i+1)}(v)$ 
    (Lines 9 - 10).
\end{itemize}

% \input{algorithms/framework}
% The correctness of \algoagp\ is given by the following lemma:
% \begin{lemmaapp}[\cite{wang2021approximate}]
%    For every query $q(a, b, \mathcal{O}_w, \mathbf{x})$, \algoagp\ returns an $(\delta, c)$-Approximation $\boldsymbol{\hat{\pi}}(v)$ with constant probability. 
% \end{lemmaapp}

%To support an efficient subset sampling, \algoagp\ pre-sorts the adjacency list by degrees and stores it in a {\bf sorted linked list}, hence the algorithm can access the small-degree vertices by a scan that stops right after $d_v > \left(\frac{L^2}{\delta} \cdot \frac{Y_{i+1}}{Y_i} \cdot \frac{\mathbf{\hat{r}(u)}}{d_u^b} \right)^{\frac{1}{a}}$. 
%\algoagp\ further constructs a data structure---denoted as $B$---for sampling the remaining neighbors, which is described next.

% \vspace{1mm}
\noindent
\textbf{The Subset Sampling Data Structure in \algoagp.} 
To sample the neighbors efficiently,
for each $u \in V$, \algoagp\ constructs a data structure $B[u]$, which consists of $\lfloor \log_2 n \rfloor + 1$ {\em power-of-two} buckets. 
Each bucket is denoted as $B[u]_j$, where $j$ is the bucket index. 
Each neighbor $v \in N[u]$ is assigned to the bucket $B[u]_j$ such that $2^{j} \leq d_v < 2^{j+1}$.
%, where $d_v$ represents the degree of $v$. 
The structure $B[u]$ is implemented as a sorted list of all the non-empty buckets and the vertices (i.e., $u$'s neighbors) therein. 
Therefore, it takes $O(d_u)$ space for each vertex and $O(n + m)$ for the entire graph to store this structure.

To decide whether a neigbhor $v \in N[u]$ is sampled,
\algoagp\ performs a subset sampling on $N[u]$. 
For each bucket $B[u]_j \in B[u]$,
\vspace{-2mm}
\begin{itemize}[leftmargin = *]
\item generate a {\em binomial} random variate $X_j \sim \text{Bin}(|B[u]_j|, p^*)$, where $p^* = \frac{1}{\varepsilon} \cdot \frac{Y_{i+1}}{Y_i} \cdot \frac{\mathbf{\hat{r}}^{(i)}(u)}{(2^{j})^a d_u^b}$, and $X_j$ is the target number of distinct vertices in $B[u]_j$ to be sampled with probability $p^*$;
\item uniformly at random sample $X_j$ {\em distinct} vertices in $B[u]_j$;
\item each sampled neighbor $v$ is accepted with probability 
$\frac{p_{u,v}}{p^*}$, where
%$=(\frac{2^{j}}{d_v})^a \geq \frac{1}{2^a}$, where
$p_{u, v} = \frac{1}{\varepsilon} \cdot \frac{Y_{i+1}}{Y_i} \cdot \frac{\mathbf{\hat{r}}^{(i)}(u)}{d_v^ad_u^b}$. 
\end{itemize}

\noindent\textbf{Analysis.}
According to the SOTA result~\cite{farach2015exact}, generating a binomial random variate with probability $p^*$ takes $O(\log n)$ expected time.
As there can be $O(\log n)$ buckets in $B[u]$,
and by the fact that each neighbor taken from the bucket will be accepted by probability at least $\frac{p_{u,v}}{p^*} =(\frac{2^{j}}{d_v})^a \geq \frac{1}{2^a} \in \Omega(1)$,   
the neighbor subset sampling for each vertex $u \in V$ takes $O(\log^2 n + \mu_u)$ expected time, where $\mu_u$ is the expected size of the sample set.

\vspace{2mm}
\begin{observation}
While \algoagp~was originally proposed~\cite{wang2021approximate} for queries with parameters $a$ and $b$ given in advance and fixed,
it can support queries with  
$a$ and $b$ given on the fly.
\end{observation}

\section{Our \algosta\ Algorithm}\label{sec:staticpp}

\begin{comment}
\alvin{A basic solution to the AGP problem on dynamic graphs is to run a static query algorithm on the snapshot of the dynamic graph at the time when a query is incurred.
An efficient static query algorithm is essential. 
For this reason, we first propose our AGP-Static++, which not only rectifies the theoretical flaws of AGP-Static but also achieves an improved query time complexity, as summarized in Algorithm~\ref{alg:imp_agp}.}
\end{comment}
In a nutshell, our \algosta~follows the algorithm framework of \algoagp.
%(Algorithm~\ref{alg:framework}).
We also group the neighbors of each vertex $u$ into a sorted list of power-of-two buckets, $B[u]$, by their degrees. 
However, there are also substantial differences that are the key to our practical and theoretical improvements:
%Before presenting our algorithm, we first highlight the key differences between \algosta\ and \algoagp:
\vspace{-1mm}
\begin{itemize}[leftmargin = *]
    \item \underline{\em Eliminating the Sorted Lists:} 
We propose a new subset sampling algorithm that removes the need of a sorted list in each neighborhood.
    
    \item \underline{\em Improved Parameter Setting:} In \algosta,  the parameter $\varepsilon$ is set to $O(\frac{\delta}{L})$, reducing a factor of $O(L)$
from that in \algoagp. 
With our tighter analysis (Section~\ref{sec:var_ana}), we  
show that \algosta\ 
achieves the same approximation guarantees as \algoagp~does.
%still returns an $(\delta, c)$-approximation for every query with at least constant probability.
%—while still ensuring a $(\delta, c)$-approximation, due to 
%our tighter bound analysis (Section~\ref{sec:var_ana}). 
This refinement improves the query time complexity by a factor of $O(L)$ compared to \algoagp.
    
    \item \underline{\em Optimized Subset Sampling Technique:} We employ a different subset sampling method that achieves a complexity of $O(\log n + \mu)$, where $\mu$ denotes the expected output size. This removes a $O(\log n)$ factor from the query complexity of \algoagp.
    
    %\item \underline{\em Adoption of Optimal Subset Sampling Algorithm:} 
%By incorporating the recent technique of Dynamic Parameterized Subset Sampling (DPSS) to our algorithm, denoted as \algodpss,  
%another $O(\log n)$ factor is removed from the query complexity.
%To attain optimal subset sampling complexity, we incorporate \emph{DPSS} into our algorithm, resulting in yet another $O(\log n)$ improvement in query complexity.

%
%\item \underline{\em Improved Initialization:} 
%When $x$ admits 
%a compact representation 
%(Condition~\ref{con:x} in Section~\ref{sec:rand}), 
%we show that the $O(n)$ initialization time can 
%be reduced to $O(\min\{\frac{1}{\varepsilon}, n\})$
%in expectation while maintaining the $(\delta, c)$-approximation.

%we proposed a randomized initialization method that avoids $O(n)$ initialization time while maintaining the $(\delta,c)$-approximation. 
%\alvin{We defer the details to our technical report~\cite{sourceCode} due to space limit.}
\end{itemize}

\subsection{The Subset Sampling Algo's in \algosta}\label{sec:ss}

%We now detail the subset sampling method used in \algosta. 
 \algosta\  consolidates the propagation procedures within the neighborhood of vertex $u$ into a single subset sampling process (Lines 3–6). This eliminates the need of the sorted list, 
which not only simplifies the algorithm in the static case 
but is also essential for handling graph updates efficiently, 
as maintaining the sorted list can be costly (e.g., $O(\log n)$ per update). 
Our proposed subset sampling method operates in $O(\mu + \log n)$ expected time, where $\mu$ is the expected output size. 
Specifically, %our \algosta\ 
it works as follows.

\begin{small}
\begin{algorithm}
\caption{\algosta}\label{alg:imp_agp}
\DontPrintSemicolon
\SetKwComment{Comment}{/* }{ */}
\SetKwInOut{Optional}{Optional Input}
\SetKwBlock{preprocess}{Preprocess Procedure:}{end}
\SetKwBlock{query}{Query Procedure:}{end}

% \preprocess{
%     \KwIn{$G = \langle V, E \rangle$}
%     \KwOut{data structure $B$}
%     \For{each $u \in V$}
%         {
%             Group neighbors into power-of-2 buckets by their degree into $B[u]$\;
%         }
%     \Return $B$
% }

% \query{
%    \KwIn{$q(a,b, \mathcal{O}_w, \mathbf{x})$, error parameter of $(\delta,c)$-approximation, number of steps $L$ computed based on $\mathcal{O}_w$, and data structure $B$}
    % \Optional{$P, S$ \text{as special form of $\mathbf{x}$}}
    %\KwOut{$\boldsymbol{\hat{\pi}}$}
 %$\mathbf{\hat{r}}^{(0)}\leftarrow\mathbf{x};
 %\varepsilon = \frac{\delta}{L};$\;
    % \lIf{$\mathbf{x}$ satisfies certain conditions}{$\mathbf{\hat{r}}^{(0)}\leftarrow RandomizedInitialize(P, S, \varepsilon);$}
    % \lElse {$\mathbf{\hat{r}}^{(0)}\leftarrow\mathbf{x};$}
   
\KwIn{$q(a, b, \mathcal{O}_w, \mathbf{x})$, parameters $\delta$ and $c$ for $(\delta,c)$-approximation
%, number of steps $L$ computed based on $\mathcal{O}_w$, data structure $B$, and sorted neighborhood} 
}
\KwOut{$\boldsymbol{\hat{\pi}}$, an $(\delta, c)$-approximation of $\boldsymbol{\pi}$}

$L \leftarrow L_{\mathcal{O}_w}(c \cdot \delta)$;
$\varepsilon \leftarrow \frac{c^2 \delta}{2(L + 1)}$;
$\mathbf{\hat{r}}^{(0)} \leftarrow \mathbf{x}$;\; 

\For{$i = 0$ to $L -1 $}
    {
        \For{each $u \in V$ with $\mathbf{\hat{r}}^{(i)}(u) \neq 0$}
        {   
            obtain a subset $T$ sampled from neighbors in $N[u]$ such that each neighbor is sampled independently with probability $p_{u,v} = \min \left\{ 1, \frac{1}{\varepsilon} \cdot \frac{Y_{i+1}}{Y_i} \cdot \frac{\mathbf{\hat{r}}^{(i)}(u)}{d_v^ad_u^b}\right\}$;\;
            \For{each $v \in T$}{
                $\mathbf{\hat{r}}^{(i+1)}(v) \leftarrow \mathbf{\hat{r}}^{(i+1)}(v) + \max\left\{ \varepsilon, \frac{Y_{i+1}}{Y_i} \cdot \frac{\mathbf{\hat{r}}^{(i)}(u)}{d_v^a d_u^b} \right\}$ \;
            }
    
            $\mathbf{\hat{q}}^{(i)}(u) \leftarrow \mathbf{\hat{q}}^{(i)}(u) + \frac{w_i}{Y_i}\cdot \mathbf{\hat{r}}^{(i)}(u);$
        }
        $\boldsymbol{\hat{\pi}} \leftarrow \boldsymbol{\hat{\pi}} + \mathbf{\hat{q}}^{(i)};$\;
    }
    $\mathbf{\hat{q}}^{(L)} = \frac{w_L}{Y_L} \cdot \mathbf{\hat{r}}^{(L)};$ and
    $\boldsymbol{\hat{\pi}} \leftarrow \boldsymbol{\hat{\pi}} + \mathbf{\hat{q}}^{(L)};$\;
    \Return $\boldsymbol{\hat{\pi}}$
% }
\end{algorithm}
\end{small}

In the $i^{\text{th}}$ iteration, for each neighborhood $N[u]$ with $\mathbf{\hat{r}}^{(i)}(u) \neq 0$:
\vspace{-2mm}
\begin{itemize}[leftmargin=*]
    \item Compute a shifting factor $s^{(i)}_u = \frac{1}{\varepsilon} \cdot \frac{Y_{i+1}}{Y_i} \cdot \frac{\mathbf{\hat{r}}^{(i)}(u)}{d_u^b}$. 
   
\item For each bucket $B[u]_j$ in $B[u]$, 
        let  $p^* = \min\left\{1, \frac{s_u^{(i)}}{(2^{j})^a}\right\}$. 
If $p^* = 1$, all the vertices in $B[u]_j$ are taken in $T$; 
otherwise, 
sample each neighbor independently with probability $p^*$: 
%\vspace{-2mm}
        \begin{itemize}[leftmargin = *]
        \item initialize $y \leftarrow 0$;
            \item generate $j$ from the bounded geometric distribution \\$Geo(p^*, |B[u]_j|)$, and set $y \leftarrow y + j$;  
            \item if $y > |B[u]_j|$,  stop; otherwise, take the $y^{\text{th}}$ vertex in $B[u]_j$ and repeat by generating another $j$ from the second step. 
        \end{itemize}
\end{itemize}

% \vspace{-4mm}
\begin{itemize}[leftmargin=*]
    \item For each $v \in T$, accept $v$ with probability $p_{\text{ac}} = \frac{p_{u,v}}{p^*}$ 
    and propagate a probability mass of 
    $\max\left\{ \varepsilon, \frac{Y_{i+1}}{Y_i} \cdot \frac{\mathbf{\hat{r}}^{(i)}(u)}{d_v^a d_u^b} \right\}$ to $\mathbf{r}^{(i+1)}(v)$.
\end{itemize}

The correctness of the above subset sampling method with geometric distribution comes from the fact that 
$j \sim Geo(p^*, |B[u]_j|)$
follows the same distribution of the first index of success in a Bernoulli trial with success probability $p^*$, which, in turn, is equivalent to independently flipping a coin of with success probability of $p^*$ for each neighbor in the bucket.

\noindent
%\underline
{\bf Analysis.} 
It is known that a random variate from a bounded geometric distribution $j \sim Geo(p^*, N)$ can be generated in $O(1)$ expected time in the Word RAM model~\cite{bringmann2013exact}. 
The probability $p^*$  used to sample in each bucket is at most $2^a \in O(1)$ times the required probability $p_{u,v}$ -- recall that $p^* = s_u\cdot \frac{1}{(2^{j})^a}$, $p_{u,v} = s_u\cdot\frac{1}{d_v^a}$ and $2^{j} \leq d_v < 2\cdot 2^{j}$. Given any vertex $u$, there are at most $O(\log n)$ buckets in $B[u]$,
%the number of buckets is bounded by $O(\log n)$ in $B[u]$, 
hence the total time complexity for the subset sampling within each neighborhood is bounded by $\sum_k O(\mu_k + 1) = O(\mu + \log n)$, where $\mu_k$ is the expected output size for the $k^{\text{th}}$ bucket in $B[u]$, and $\mu$ is the expected sample size from $N[u]$.

\section{Theoretical Analysis}\label{sec:ana}
% \alvin{Throughout this section, we defer the proofs of all the theorems and lemmas with [*] to Section~\ref{sec:ana} for a smoother readability.}

\subsection{Query Running Time Analysis}
\label{sec:query_ana}

We now analyze the complexity of \algoagp\ and  \algosta.
%~and \algodpss~algorithms.
Without loss of generality, we denote the subset sampling cost within the neighborhood of a vertex $u$ by $O(\mu_u + \cso)$, where $\mu_u$ is the expected sample size and $\cso$ is the {\em non-chargeable} overhead. 
In particular, $\cso = O(\log^2 n)$ for \algoagp, 
while
$\cso = O(\log n)$ for our \algosta. 
%and $\cso = O(1)$ for our \algosta\ and \algodpss, respectively. 
%, accounting for the overhead cost of subset sampling. While some overhead can be charged into the output cost, others cannot. Here, we focus on the non-chargeable overhead cost, denoted as $C_o$, such that the subset sampling cost can be expressed as $O(\mu + C_o)$.
%
% We prove the following theorem whose proof can be found in Appendix~\ref{ap:t1}.
\begin{theorem}
\label{theorem:query}%{\emph{[*]}}
    The expected query time complexity of  AGP algorithms (e.g., \algoagp\ and \algosta) is bounded by
    $$O\left(\cso \cdot \frac{1}{\varepsilon}\sum_{i=1}^L \|Y_i(\mathbf{D}^{-a}\mathbf{A}\mathbf{D}^{-b})^i\cdot \mathbf{x}\|_1 + \E[C_{\text{init}}]\right)\,,$$
where $E[C_{\text{init}}]$ is the expected cost of initialization, i.e., reading the query input $q(a, b, \mathcal{O}_w, \mathbf{x})$.
%\|Y_i(D^{-a}AD^{-b})^i\cdot \mathbf{x}\|_1 + E[C_{init}]\right)$$
\end{theorem}

\begin{proof}

Besides the initialization cost ($C_{\text{init}}$), 
the query running time mainly consists of these two costs: 
%comes from the accumulation of:
\begin{itemize}[leftmargin = *]
    \item the propagation cost on each edge at each iteration, denoted as $c^{(i)}(u,v)$, meaning the cost to propagate the probability mass from $u$ to $v$ at the $i^{\text{th}}$ iteration;
    \item the subset sampling non-chargeable overhead cost at each neighborhood at each iteration, denoted as $c^{(i)}_{\text{so}}(u)$, meaning the non-chargeable subset sampling overhead cost for $N[u]$ at the $i^{\text{th}}$ iteration. 
Here, we only need to consider the non-chargeable overhead as the output cost and the chargeable overhead can be charged to $c^{(i)}(u,v)$ for those neighbors $v$ sampled in $N[u]$ at step $i$.  
\end{itemize}  

We first analyze 
the expected cost of $c^{(i)}(u,v)$. 
According to Line 4 in Algorithm~\ref{alg:imp_agp},
each neighbor is sampled independently by probability $\min\{1, p_{u,v}\}$, and each propagation cost is $O(1)$. 
Thus,
%we have
the expected cost of $c^{(i)}(u,v)$
conditioned on 
$\mathbf{\hat{r}}^{(i-1)}(u)$,
i.e.,   
$\E[c^{(i)}(u,v)|\mathbf{\hat{r}}^{(i-1)}(u)] = \min \{1, p_{u,v}\} \leq p_{u,v}$, where $ p_{u,v} =\frac{1}{\varepsilon} \cdot \frac{Y_i}{Y_{i-1}}\cdot \frac{\mathbf{\hat{r}}^{(i-1)}(u)}{d_v^a d_u^b}$. 
%Referring to Algorithm~\ref{alg:imp_agp}, this is because 
%each neighbor will be sampled independently by probability $\min\{1, p_{u,v}\}$, and each propagation cost is $O(1)$. 
%Specifically,
%according to Line 4, 
%in Algorithm~\ref{alg:imp_agp}, 
%$E[c^{(i)}(u,v)|\mathbf{\hat{r}}^{(i-1)}(u)] = 1$  when $p_{u,v} = \frac{1}{\varepsilon} \cdot \frac{\mathbf{\hat{r}}^{(i-1)}(u)}{d_v^a d_u^b} \geq 1$, and otherwise, i.e., 
%$p_{u,v} \leq 1$, 
%$E[c^{(i)}(u,v)|\mathbf{\hat{r}}^{(i-1)}(u)] = p_{u,v} = \frac{1}{\varepsilon} \cdot \frac{\mathbf{\hat{r}}^{(i-1)}(u)}{d_v^a d_u^b}$.
%based on the sampled probability.

%Moreover, 
It is known that the following fact holds: 
%has been proven in~\cite{wang2021approximate}:
\begin{fact}
[\cite{wang2021approximate}]
\label{fact:unbias}
$\mathbf{\hat{r}}$ is an unbiased estimation of $\mathbf{r}$, i.e., $\E[\mathbf{\hat{r}}^{(i-1)}(u)] = \mathbf{r}^{(i-1)}(u)$. 
\end{fact}
\noindent
By Fact~\ref{fact:unbias}, we then have:
%the following lemma:
%\begin{lemma}
%$\E[c^{(i)}(u,v)]\leq  \frac{1}{\varepsilon} \cdot \frac{Y_i}{Y_{i-1}}\cdot \frac{\mathbf{r}^{(i-1)}(u)}{d_v^a d_u^b}$.
%\end{lemma}
%\begin{proof} 
\begin{small}
\begin{align*}
    \E[c^{(i)}(u,v)] &= \E[\E[c^{(i)}(u,v)|\mathbf{\hat{r}}^{(i-1)}(u)]]\\
    &= \sum_j \pr[\mathbf{\hat{r}}^{(i-1)}(u)_j]\cdot \E[c^{(i)}(u,v)|\mathbf{\hat{r}}^{(i-1)}(u)_j]\\
    &\leq \sum_j \pr[\mathbf{\hat{r}}^{(i-1)}(u)_j]\cdot \frac{1}{\varepsilon} \cdot \frac{Y_i}{Y_{i-1}}\cdot \frac{\mathbf{\hat{r}}^{(i-1)}(u)_j}{d_v^a d_u^b}\\ &=\frac{1}{\varepsilon} \cdot \frac{Y_i}{Y_{i-1}}\cdot \frac{\sum_j \pr[\mathbf{\hat{r}}^{(i-1)}(u)_j]\cdot \mathbf{\hat{r}}^{(i-1)}(u)_j}{d_v^a d_u^b}\\
    &\leq \frac{1}{\varepsilon} \cdot \frac{Y_i}{Y_{i-1}}\cdot \frac{\E[\mathbf{\hat{r}}^{(i-1)}(u)]}{d_v^a d_u^b} =\frac{1}{\varepsilon} \cdot \frac{Y_i}{Y_{i-1}}\cdot \frac{\mathbf{r}^{(i-1)}(u)}{d_v^a d_u^b}\,.
\end{align*}
%\end{proof}
\end{small}

Next, we analyze $\E[c^{(i)}_{\text{so}}(u)]$. 
To facilitate our analysis, we define $I^{(i)}(u)$ as the indicator 
of whether $\mathbf{\hat{r}}^{(i)}(u) > 0$ is true or not.
Specifically, $I^{(i)}(u) = 1$ when $\mathbf{\hat{r}}^{(i)}(u) > 0$, and 
$I^{(i)}(u) = 0$ otherwise.
%equals to 0 otherwise.
%as we only need to process the vertex $u$ at $i^{\text{th}}$ when $\hat{\mathbf{r}}^{(i)}(u) > 0$. 
%Hence we have $E[c^{(i)}_{\text{so}}(u)] = \sum_{v\in N[u]}E[I^{(i)}(v)\cdot C_o]$, as we consider the subset sampling overhead cost only. The output cost is already considered in $c^{(i)}(u,v)$. 
When $\hat{\mathbf{r}}^{(i)}(u) > 0$, only if 
at the $(i-1)^{\text{th}}$ step, at least one of $u$'s neighbors has propagated to $u$.
In other words, $u$ must have been sampled from the neighborhood $N[v]$ of at least one of its neighbors $v \in N[u]$. 
%For any $v \in N[u]$, the probability to propagate from $v$ to $u$ is  
Since $u$ is sampled independently from each of its neighbors $v$ with probability
$\min\{1, p_{v,u}\}$,
%either $\frac{1}{\varepsilon} \cdot \frac{Y_i}{Y_{i-1}}\cdot \frac{\mathbf{\hat{r}}^{(i-1)}(u)}{d_v^a d_u^b}$ if it is sampled or 1 when $1 \leq \frac{1}{\varepsilon} \cdot \frac{Y_i}{Y_{i-1}}\cdot \frac{\mathbf{\hat{r}}^{(i-1)}(u)}{d_v^a d_u^b}$. 
%As all $u$'s neighbors are sampled independently, 
by Union Bound, when $i \geq 1$, $\pr[I^{(i)}(u) = 1] \leq \sum_{v\in N[u]}\frac{1}{\varepsilon} \cdot \frac{Y_i}{Y_{i-1}}\cdot\frac{\mathbf{\hat{r}}^{(i-1)}(v)}{d_u^a d_v^b}$. 
Therefore, we have 
$\E\left[I^{(i)}(u)\right] = \E\left[\E\left[I^{(i)}(u)\big|\mathbf{\hat{r}}^{(i-1)}\right]\right] \leq \sum_{v \in N[u]} \frac{1}{\varepsilon}\cdot\frac{Y_i}{Y_{i-1}}\cdot\frac{\E[\mathbf{ \hat{r}}^{(i-1)}(v)]}{d_u^a d_v^b}= \frac{1}{\varepsilon}\cdot \mathbf{r}^{(i)}(u)$; recall that $\mathbf{\hat{r}}$ is an unbiased estimation. 

Putting these together, the expected overall cost can be computed as follows: 

\begin{small}
\begin{align*}
\E_{\text{cost}} &= \sum_{i=1}^L\sum_{v\in V}\sum_{u \in N[v]} \E[c^{(i)}(u,v)] + \sum_{i=1}^{L-1}\sum_{v \in V} \E[c_{so}^{(i)}(v)] + \E[C_{\text{init}}]\\
 &\leq \sum_{i=1}^L\sum_{v\in V}\sum_{u \in N[v]} \frac{1}{\varepsilon} \cdot \frac{Y_i}{Y_{i-1}}\cdot \frac{\mathbf{r}^{(i-1)}(u)}{d_v^a d_u^b} \\
& \;\;\;\;\;+ \sum_{i=1}^{L-1}\sum_{v \in V} \E[I^{(i)}(v)]\cdot \cso + \E[C_{\text{init}}]\\
% & \leq \sum_{i=1}^L\sum_{v\in V}\sum_{u \in N[v]} \frac{1}{\varepsilon} \cdot \frac{Y_i}{Y_{i-1}}\cdot \frac{\mathbf{r}^{(i-1)}(u)}{d_v^a d_u^b} + \sum_{i=1}^{L-1}\sum_{u\in V}\frac{1}{\varepsilon}\cdot \mathbf{r}^{i}(v) \cdot C_o + E[C_{\text{init}}]\\
& \leq \frac{1}{\varepsilon}\sum_{v\in V}\sum_{i=1}^L\mathbf{r}^{(i)}(v) + \cso \cdot \frac{1}{\varepsilon}\sum_{v \in V}\sum_{i=1}^{L-1}\mathbf{r}^{(i)}(v) + \E[C_{\text{init}}]\\
& = O\left(\cso \cdot \frac{1}{\varepsilon}\sum_{i=1}^L \|Y_i(\mathbf{D}^{-a}\mathbf{A}\mathbf{D}^{-b})^i\cdot \mathbf{x}\|_1 + \E[C_{\text{init}}]\right)\,,
\end{align*}
\end{small}
where the last bound follows the definition of
%By the definition of 
$\mathbf{r}^{(i)}(u)$.
%we have $\E_{cost} \in O\left(\cso \cdot \frac{1}{\varepsilon}\sum_{i=1}^L \|Y_i(\mathbf{D}^{-a}\mathbf{A}\mathbf{D}^{-b})^i\cdot \mathbf{x}\|_1 + \E[C_{\text{init}}]\right)$.
\end{proof}

By substituting the values of $\cso$ and $\varepsilon$
of each algorithm into Theorem~\ref{theorem:query}, the concrete expected query time complexity 
follows.
%It is worth mentioning that,
%as we prove a tighter variance next,
%the value of $\frac{1}{\varepsilon}$ in our algorithms is a factor of $O(L)$ smaller than that in \algoagp.
%In particular, we have the following corollaries:
\begin{corollary}
The expected query time complexity of \algoagp\ is bounded by
    $$O\left(\log^2 n \cdot \frac{L^2}{\delta}\sum_{i=1}^L \|Y_i(\mathbf{D}^{-a}\mathbf{A}\mathbf{D}^{-b})^i\cdot \mathbf{x}\|_1 + \E[C_{\text{init}}]\right)\,.$$
\end{corollary}

\begin{corollary}
The expected query time complexity of \algosta\ is bounded by
$$O\left(\log n \cdot \frac{L}{\delta}\sum_{i=1}^L \|Y_i(\mathbf{D}^{-a}\mathbf{A}\mathbf{D}^{-b})^i\cdot \mathbf{x}\|_1 + \E[C_{\text{init}}]\right)\,.$$
\end{corollary}

\noindent\textbf{Incorporating DPSS.}
\label{ap:dpss}
The subset sampling cost in \algosta~can be further improved to $O(1 + \mu)$ in expectation by 
%Next, we show how to incorporate the 
%The 
the {\em Dynamic Parameterized Subset Sampling} (DPSS) technique~\cite{gan2024optimal} which 
was proposed to solve the problem defined as follows.
%to our algorithm, which aims to solve the problem below.
Consider a set $S$ containing $n$ elements allowing element insertions or deletions.
Each element $e \in S$ is associated with a {\em non-negative integer} weight $w(e)$.
Given a pair of non-negative rational parameters $(\alpha, \beta)$,
the objective is to return a subset of $S$ such that 
each element $e \in S$ is independently sampled with probability $\min\left\{ \frac{w(e)}{\alpha \sum_{e\in S} w(e) + \beta}, 1 \right\}$.
Such a subset sample can be obtained with DPSS in $O(1 + \mu)$ expected time.
To apply DPSS to the subset sampling for \algosta, for each vertex $u$, the element set $S$ is $N[u]$, 
where each vertex $v \in N[u]$ is an element with weight $w(v) = \lceil \frac{n_{\max}}{d^a_v} \rceil$, with $n_{\max}$ being the maximum possible $n$. 
A subset sample $T$ from $N[u]$, where each $v \in N[u]$ is sampled independently  with probability $p_{u,v}$,  can be obtained 
by as follows:
\begin{itemize}[leftmargin =*]
\item perform a DPSS query of 
 $\alpha = 0$ and $\beta = \frac{n_{\max}}{s_u}$ to obtain a candidate set $S' \subseteq S$;\vspace{-1mm}
\item for each $v \in S'$, accept $v$ with probability $\frac{p_{u,v}}{w(v)/\beta}\in \Omega(1)$.
\end{itemize}
Denote the abobe implementation with DPSS of \algosta~by \algodpss.
%Therefore, 
While the subset sampling time for each $N[u]$ 
is bounded by $O(1 + \mu_u)$ in expectation, where $\mu_u$ is the expected subset sample size in $N[u]$, 
%However, 
%due to 
%our setting for the weights $w(v)$, 
%with this technique, \algosta~
%Unfortunately, 
\algodpss~cannot support the parameter $a$ given on the fly 
as $a$ is used to build the data structure, i.e., the weight of each element.
Substituting $\cso = O(1)$ and $\varepsilon = O(\frac{\delta}{L})$ to Theorem~\ref{theorem:query}, we have the following corollary: 
\begin{corollary}
The expected query time complexity of \algodpss\ is bounded by
$$O\left(\frac{L}{\delta}\sum_{i=1}^L \|Y_i(\mathbf{D}^{-a}\mathbf{A}\mathbf{D}^{-b})^i\cdot \mathbf{x}\|_1 + \E[C_{\text{init}}]\right)\,.$$
\end{corollary}

\subsection{Tighter Variance and Error Bound}
\label{sec:var_ana}

We now analyze the variance of $\boldsymbol{\hat{\pi}}$ returned by our \algosta.
%to derive an error bound. 
%
% And we have $E[X^i(u,v)|\mathbf{\hat{r}}^{(i-1)}] = \frac{Y_{i+1}}{Y_i}\cdot \frac{\mathbf{\hat{r}}^{(i)}(u)}{d_v^a\cdot d_u^b}$.
%
We define $\left\{\mathbf{\hat{r}}^{(l)}\right\}$ as a union of $\mathbf{\hat{r}}^{(i)}$ for $i = 0 \cdots l$, and $p^{(i)}(u,v) = \mathbf{e}^{T}_v \cdot (\mathbf{D}^{-a}\mathbf{A}\mathbf{D}^{-b})^{i}\cdot \mathbf{e}_u$, which is the $i^{\text{th}}$ normalized transition probability from vertex $u$ to $v$, where $\mathbf{e}_v$ is a one-hot vector with $\mathbf{e}_v(v) = 1$.  
%Recall that $p^{(i)}(u,v)$ denotes the $i^{\text{th}}$ normalized transition probability from vertex $u$ to vertex $v$. 
In particular, $p^{(0)}(u,v) = 1$ if and only if $u = v$; otherwise, $p^{(0)}(u,v) = 0$. 
Furthermore, $p^{(1)}(u,v) = \frac{1}{d_v^a \cdot d_u^b}$ if $(u,v) \in E$; otherwise, $p^{(1)}(u,v) = 0$.
We first prove Lemma~\ref{lemma:epe} and Lemma~\ref{lemma:var} regarding the expectation and the variance of $\mathbf{\hat{r}}^{(i)}$.

\begin{lemma}%{\emph{[*]}}
\begin{small}
\label{lemma:epe}
    %  $\E\left[\mathbf{\hat{r}}^{(l)}(v)\big|\left\{\mathbf{\hat{r}}^{(l-1)}\right\}\right] = 
    %\mathbf{\hat{r}}^{(l)}(v)$.
     $\E\left[\mathbf{\hat{r}}^{(l)}(v)\big|\left\{\mathbf{\hat{r}}^{(l-1)}\right\}\right] 
    %= \E\left[\sum_{u \in N[v]}X^{(l)}(u,v)\big|\left\{\mathbf{\hat{r}}^{(l-1)}\right\}\right] 
    %= \sum_{u \in N[v]}\frac{Y_{l}}{Y_{l-1}}\cdot \frac{\mathbf{\hat{r}}^{(l-1)}(u)}{d_v^a\cdot d_u^b} 
    = \sum_{u \in N[v]} \frac{Y_{l}}{Y_{l-1}}\cdot p^{(1)}(u,v) \cdot \mathbf{\hat{r}}^{(l-1)}(u)$.
    % 
    %$\E\left[\mathbf{\hat{r}}^{(l)}(v)\big|\left\{\mathbf{\hat{r}}^{(l-1)}\right\}\right] = \E\left[\sum_{u \in N[v]}X^{(l)}(u,v)\big|\left\{\mathbf{\hat{r}}^{(l-1)}\right\}\right] = \sum_{u \in N[v]}\frac{Y_{l}}{Y_{l-1}}\cdot \frac{\mathbf{\hat{r}}^{(l-1)}(u)}{d_v^a\cdot d_u^b} = \sum_{u \in N[v]} \frac{Y_{l}}{Y_{l-1}}\cdot p^{(1)}(u,v) \cdot \mathbf{\hat{r}}^{(l-1)}(u)$.
\end{small}
\end{lemma}
    
\begin{proof}
     As according to Fact~\ref{fact:unbias}, $\boldsymbol{\hat{\pi}}(v) = \sum_{i=0}^L\mathbf{\hat{q}}^{(i)}(v) = \sum_{i=0}^L\frac{w_i}{Y_i}\cdot \mathbf{\hat{r}}^{(i)}(v)$ is an unbiased estimation, we have $\E[\boldsymbol{\hat{\pi}}(v)] = \sum_{i=0}^L\mathbf{q}^{(i)}(v)$.

First, we use $X^{(i)}(u,v)$ to denote the value propagated from $u$ to $v$ at step $i$, which is formally defined as below:
\begin{small}
\begin{align*}
X^{(i)}(u,v)=
\begin{cases}
\frac{Y_{i}}{Y_{i-1}}\cdot \frac{\mathbf{\hat{r}}^{(i-1)}(u)}{d_v^a\cdot d_u^b} \;\;\;\;\;\;\;\;\;\;\;\;\;\;\;\;\;\;\;\;\;\;\;\;\;\;\;\text{when} \; \frac{Y_{i}}{Y_{i-1}}\cdot \frac{\mathbf{\hat{r}}^{(i-1)}(u)}{d_v^a\cdot d_u^b} \geq \varepsilon \\
\begin{cases}
\varepsilon  & \text{w.p.} \; \frac{1}{\varepsilon}\cdot \frac{Y_{i}}{Y_{i-1}}\cdot \frac{\mathbf{\hat{r}}^{(i-1)}(u)}{d_v^a\cdot d_u^b}\\
0 & \text{w.p.} \; 1-\frac{1}{\varepsilon}\cdot\frac{Y_{i}}{Y_{i-1}}\cdot \frac{\mathbf{\hat{r}}^{(i-1)}(u)}{d_v^a\cdot d_u^b}
\end{cases} \;\text{otherwise}
\end{cases}
\end{align*}
\end{small}

\noindent
By the definition of $X^{(l)}(u,v)$, we have:
\\$E\left[X^{(l)}(u,v)\big|\left\{\mathbf{\hat{r}}^{(l-1)}\right\}\right] = \frac{Y_{l}}{Y_{l-1}}\cdot \frac{\mathbf{\hat{r}}^{(l-1)}(u)}{d_v^a\cdot d_u^b}$. This is because, when $\frac{Y_{i}}{Y_{i-1}}\cdot \frac{\mathbf{\hat{r}}^{(i-1)}(u)}{d_v^a\cdot d_u^b} \geq \varepsilon$, $\;\;X^{(l)}(u,v) = \frac{Y_{l}}{Y_{l-1}}\cdot \frac{\mathbf{\hat{r}}^{(l-1)}(u)}{d_v^a\cdot d_u^b}$;
otherwise \\$E[X^{(l)}(u,v)] = \varepsilon \cdot \frac{1}{\varepsilon}\frac{Y_{i}}{Y_{i-1}}\cdot \frac{\mathbf{\hat{r}}^{(i-1)}(u)}{d_v^a\cdot d_u^b} + 0 \cdot (1- \frac{1}{\varepsilon}\frac{Y_{i}}{Y_{i-1}}\cdot \frac{\mathbf{\hat{r}}^{(i-1)}(u)}{d_v^a\cdot d_u^b}) = \frac{Y_{i}}{Y_{i-1}}\cdot \frac{\mathbf{\hat{r}}^{(i-1)}(u)}{d_v^a\cdot d_u^b}$. As $X^{(l)}(u,v)$ denotes the value propagated from $u$ to $v$ at step $l$, hence $\mathbf{\hat{r}}^{(l)}(v) = \sum_{u \in N[v]}X^{(l)}(u,v)$. Therefore, we have $E\left[\mathbf{\hat{r}}^{(l)}(v)\big|\left\{\mathbf{\hat{r}}^{(l-1)}\right\}\right] = E\left[\sum_{u \in N[v]}X^{(l)}(u,v)\big|\left\{\mathbf{\hat{r}}^{(l-1)}\right\}\right] = \sum_{u \in N[v]}\frac{Y_{l}}{Y_{l-1}}\cdot \frac{\mathbf{\hat{r}}^{(l-1)}(u)}{d_v^a\cdot d_u^b}$.
\end{proof}

\begin{lemma}%{\emph{[*]}}
\label{lemma:var}
    % $\Var\left[\mathbf{\hat{r}}^{(l)}(v)|\left\{\mathbf{\hat{r}}^{(l-1)}\right\}\right] 
    %\leq
    %\varepsilon \cdot \mathbf{\hat{r}}^{(l)}(v).$
    $\Var\left[\mathbf{\hat{r}}^{(l)}(v)|\left\{\mathbf{\hat{r}}^{(l-1)}\right\}\right] 
    %= \Var\left[\sum_{u \in N[v]}X^{(l)}(u,v)|\left\{\mathbf{\hat{r}}^{(l-1)}\right\}\right] 
    %\leq \sum_{u \in N[v]}\varepsilon \cdot \frac{Y_{l}}{Y_{l-1}}\cdot \frac{\mathbf{\hat{r}}^{(l-1)}(u)}{d_v^a\cdot d_u^b} 
    \leq \sum_{u \in N[v]}\varepsilon \cdot \frac{Y_{l}}{Y_{l-1}}\cdot p^{(1)}(u,v) \cdot \mathbf{\hat{r}}^{(l-1)}(u)$.
    
    %$\Var\left[\mathbf{\hat{r}}^{(l)}(v)|\left\{\mathbf{\hat{r}}^{(l-1)}\right\}\right] = \Var\left[\sum_{u \in N[v]}X^{(l)}(u,v)|\left\{\mathbf{\hat{r}}^{(l-1)}\right\}\right] \leq \sum_{u \in N[v]}\varepsilon \cdot \frac{Y_{l}}{Y_{l-1}}\cdot \frac{\mathbf{\hat{r}}^{(l-1)}(u)}{d_v^a\cdot d_u^b} = \sum_{u \in N[v]}\varepsilon \cdot \frac{Y_{l}}{Y_{l-1}}\cdot p^{(1)}(u,v) \cdot \mathbf{\hat{r}}^{(l-1)}(u)$.
\end{lemma}

\begin{proof}
     Observe that when $\frac{Y_{l}}{Y_{l-1}}\cdot \frac{\mathbf{\hat{r}}^{(l-1)}(u)}{d_v^a\cdot d_u^b} \geq \varepsilon$, $X^{(l)}(u,v)$ is deterministic, 
    hence this case has no variance. 
    %we can only consider the other case. 
    It thus suffices to focus on the case that  
    $\frac{Y_{l}}{Y_{l-1}}\cdot \frac{\mathbf{\hat{r}}^{(l-1)}(u)}{d_v^a\cdot d_u^b} < \varepsilon$, where $X^{(l)}(u,v) \leq \varepsilon$.
    We have 
    $\Var\left[X^{(l)}(u,v)\big|\left\{\mathbf{\hat{r}}^{(l-1)}\right\}\right] = 
    \E\left[\left(X^{(l)}(u,v)\right)^2\big|\left\{\mathbf{\hat{r}}^{(l-1)}\right\}\right] - \left(\E\left[X^{(l)}(u,v)\big|\left\{\mathbf{\hat{r}}^{(l-1)}\right\}\right]\right)^2 
    \leq 
    \E\left[\left(X^{(l)}(u,v)\right)^2\big|\left\{\mathbf{\hat{r}}^{(l-1)}\right\}\right] 
    \leq \varepsilon^2 \cdot \frac{1}{\varepsilon} \cdot \frac{Y_{l}}{Y_{l-1}}\cdot \frac{\mathbf{\hat{r}}^{(l-1)}(u)}{d_v^a\cdot d_u^b} = \varepsilon \cdot \frac{Y_{l}}{Y_{l-1}}\cdot \frac{\mathbf{\hat{r}}^{(l-1)}(u)}{d_v^a\cdot d_u^b}$.
    Again, by the fact that \\$\mathbf{\hat{r}}^{(l)}(v) = \sum_{u \in N[v]}X^{(l)}(u,v)$, we then have $\Var\left[\mathbf{\hat{r}}^{(l)}(v)|\left\{\mathbf{\hat{r}}^{(l-1)}\right\}\right] = \Var\left[\sum_{u \in N[v]}X^{(l)}(u,v)|\left\{\mathbf{\hat{r}}^{(l-1)}\right\}\right] \leq \sum_{u \in N[v]}\varepsilon \cdot \frac{Y_{l}}{Y_{l-1}}\cdot \frac{\mathbf{\hat{r}}^{(l-1)}(u)}{d_v^a\cdot d_u^b}
    =\sum_{u \in N[v]}\varepsilon \cdot \frac{Y_{l}}{Y_{l-1}}\cdot p^{(1)}(u,v) \cdot \mathbf{\hat{r}}^{(l-1)}(u)$.
\end{proof} 
    
\begin{lemma}%{\emph{[*]}}
\label{lemma:var_ind} 
For $k = 0...L-1$, we have:
\begin{small}
    \begin{align*}
        &\;\;\;\;\Var\left[\sum_{j=0}^k\frac{w_{L-k+j}}{Y_{L-k}} \sum_{u \in V} p^{(j)}(u,v) \cdot \mathbf{\hat{r}}^{(L-k)}(u) + \sum_{i=0}^{L-k-1}\frac{w_i}{Y_i}\mathbf{\hat{r}}^{(i)}(v)\right]\\
        &=\varepsilon\cdot\sum_{j=0}^{k}\mathbf{q}^{(L-k+j)}(v) + Var\Bigg[\sum_{j=0}^{k+1}\frac{w_{L-(k+1)+j}}{Y_{L-(k+1)}} \sum_{u \in V} p^{(j)}(u,v) \cdot \mathbf{\hat{r}}^{(L-(k+1))}(u)\\
        &\;\;\;\;\;\;\;\;\;\;\;\;\;\;\;\;\;\;\;\;\;\;\;\;\;\;\;\;\;\;\;\;\;\;\;\;\;\;\;\;\;\;\;\;\; + \sum_{i=0}^{L-(k+1)-1}\frac{w_i}{Y_i}\mathbf{\hat{r}}^{(i)}(v)\Bigg]
\,.
    \end{align*}
\end{small}
\end{lemma}

Next, we prove Lemma~\ref{lemma:var_ind} by \textbf{mathematical induction}.
We first prove the \textbf{base case} where $k = 0$, which is Lemma~\ref{lemma:base}. 
    \begin{lemma}%{\emph{[*]}}
    \label{lemma:base}
        \begin{small}
            \begin{align*}       &~~~~\underbrace{\Var\left[\sum_{j=0}^0\frac{w_{L+j}}{Y_{L}} \sum_{u \in V} p^{(j)}(u,v) \cdot \mathbf{\hat{r}}^{(L)}(u) + \sum_{i=0}^{L-1}\frac{w_i}{Y_i}\cdot \mathbf{\hat{r}}^{(i)}(v)\right]}_{(2)} \\
                &= \varepsilon \cdot \mathbf{q}^{(L)}(v) + \Var\Bigg[\sum_{j=0}^1\frac{w_{L-1+j}}{Y_{L-1}} \sum_{u \in V} p^{(j)}(u,v) \cdot \mathbf{\hat{r}}^{(L-1)}(u) \\
                &\;\;\;\;\;\;\;\;\;\;\;\;\;\;\;\;\;\;\;\;\;\;\;\;\;\;\;\;\;\;\;\;\;\;\; + \sum_{i=0}^{L-2}\frac{w_i}{Y_i}\mathbf{\hat{r}}^{(i)}(v)\Bigg]
            \end{align*}
        \end{small}
    \end{lemma}
    \begin{proof}
        By the law of Total Variance on Term (2), we have:
        \begin{small}
        \begin{align*}
            \text{Term (2)} &= \Var\left[\sum_{i=0}^L\frac{w_i}{Y_i}\cdot \mathbf{\hat{r}}^{(i)}(v)\right] \\
            &=\underbrace{\E\left[\Var\left[\sum_{i=0}^L\frac{w_i}{Y_i}\cdot \mathbf{\hat{r}}^{(i)}(v)\Bigg|\left\{\mathbf{\hat{r}}^{(L-1)}\right\}\right]\right]}_{(3)} \\
            &\;\;\;\;\;+ 
        \underbrace{\Var\left[\E\left[\sum_{i=0}^L\frac{w_i}{Y_i}\cdot 
        \mathbf{\hat{r}}^{(i)}(v)\Bigg|\left\{\mathbf{\hat{r}}^{(L-1)}\right\} \right]\right]}_{(4)}\;.
        \end{align*}
        \end{small}
        
        We first look into Term (3).
        As $\left\{\mathbf{\hat{r}}^{(L-1)}\right\}$ is given, the variance only comes from $\mathbf{\hat{r}}^{(L)}$. By Lemma~\ref{lemma:var} and the fact that $\mathbf{\hat{r}}$ is an unbiased estimation of $\mathbf{r}$: 
        \begin{small}
        \begin{align*}
            \text{Term (3)} &= \E\left[\Var\left[\frac{w_L}{Y_L}\cdot\mathbf{\hat{r}}^{(L)}(u) \middle| \left\{\mathbf{\hat{r}}^{(L-1)}\right\} \right]\right]\\
            &\leq
        \E\left[(\frac{w_L}{Y_{L}})^2 \cdot \varepsilon \cdot \mathbf{\hat{r}}^{(L)}(v) \right] 
         \leq 
        \varepsilon \cdot \frac{w_L}{Y_L} \cdot \mathbf{r}^{(L)}(v) = \varepsilon \cdot \mathbf{q}^{(L)}(v)\,, 
        %\E\left[\frac{w_L}{Y_{L}} \sum_{u \in N[v]} \varepsilon \cdot \frac{ Y_{L}}{Y_{L-1}} \cdot p^{(1)}(u,v) \cdot \mathbf{\hat{r}}^{(L-1)}(u) \right] \\ 
        %\frac{\mathbf{\hat{r}}^{(L-1)}(u)}{d_v^a\cdot d_u^b}\right] \\
            %&=\varepsilon \cdot \frac{w_L}{Y_{L-1}}   \sum_{u\in N[v]} p^{(1)}(u,v) \cdot\mathbf{r}^{(L-1)}(u)\,. 
        \end{align*}
        \end{small}
        
        where the last inequality follows from $w_L \leq \sum_{j= L}^\infty w_j = Y_L$.
        
        %Given $\mathbf{q^{(i)}} = \frac{w_i}{Y_i}\cdot \mathbf{r^{(i)}}$,  we have Term (3) $= \varepsilon \cdot \mathbf{q}^{(L)}(v)$. 
        Next, we analyze Term (4), again as $\left\{\mathbf{\hat{r}}^{(L-1)}\right\}$ is given, and by Lemma~\ref{lemma:epe}:
        \begin{small}
            \begin{align*}
                \text{Term (4)} &=\Var\left[\E\left[\frac{w_L}{Y_L}\cdot \mathbf{\hat{r}}^{(L)}(v) + \sum_{i=0}^{L-1}\frac{w_i}{Y_i}\cdot \mathbf{\hat{r}}^{(i)}(v)\Bigg|\left\{\mathbf{\hat{r}}^{(L-1)}\right\} \right]\right]\\ 
                &=\Var\left[\E\left[\frac{w_L}{Y_L}\cdot \mathbf{\hat{r}}^{(L)}(v)\Bigg|\left\{\mathbf{\hat{r}}^{(L-1)}\right\}\right] + \sum_{i=0}^{L-1}\frac{w_i}{Y_i}\mathbf{\hat{r}}^{(i)}(v)\right]\\
        % &=\Var\left[\frac{w_L}{Y_{L}} \cdot \mathbf{\hat{r}}^{(L)}(v) + \sum_{i=0}^{L-1}\frac{w_i}{Y_i}\mathbf{\hat{r}}^{(i)}(v)\right]\\
        &=\Var\left[\frac{w_L}{Y_{L}} \sum_{u \in N[v]}\frac{Y_{L}}{Y_{L-1}}\cdot p^{(1)}(u,v)\cdot \mathbf{\hat{r}}^{(L-1)}(u) + \sum_{i=0}^{L-1}\frac{w_i}{Y_i}\mathbf{\hat{r}}^{(i)}(v)\right]\\
                &=\Var\Bigg[\sum_{u \in N[v]}\frac{w_{L}}{Y_{L-1}}\cdot p^{(1)}(u,v)\cdot \mathbf{\hat{r}}^{(L-1)}(u) + \frac{w_{L-1}}{Y_{L-1}}\mathbf{\hat{r}}^{(L-1)}(v)\\
                &\;\;\;\;\;\;\; + \sum_{i=0}^{L-2}\frac{w_i}{Y_i}\mathbf{\hat{r}}^{(i)}(v)\Bigg]\,.
            \end{align*}
        \end{small}
        
        %Recall that $p^{(i)}(u,v)$ denotes the $i^{\text{th}}$ normalized transition probability from vertex $u$ to vertex $v$. In particular, $p^{(0)}(u,v) = 1$ if and only if $u = v$; otherwise, $p^{(0)}(u,v) = 0$. Similarly, $p^{(1)}(u,v) = \frac{1}{d_v^a \cdot d_u^b}$ if $(u,v) \in E$; otherwise, $p^{(1)}(u,v) = 0$. With that we can rewrite Term (3) as 
        By the definition of $p^{(j)}(u,v)$, Term (4) can be  rewritten as:\\
        $ \Var\left[\sum_{j=0}^1\frac{w_{L-1+j}}{Y_{L-1}} \sum_{u \in V} p^{(j)}(u,v) \cdot \mathbf{\hat{r}}^{(L-1)}(u) + \sum_{i=0}^{L-2}\frac{w_i}{Y_i}\mathbf{\hat{r}}^{(i)}(v)\right]$.
        
        Summing Term (3) and Term (4) up, Lemma~\ref{lemma:base} follows.
    \end{proof} 
   
\noindent 
    We then prove the \textbf{inductive case} as stated in Lemma~\ref{lemma:mathind}.
    
    \begin{lemma}%{\emph{[*]}}
    \label{lemma:mathind}
        If Lemma~\ref{lemma:var_ind} holds for $k = l$, then it is also true for $k = l + 1$ for $0 \leq l < L$.
    \end{lemma}
    
    \begin{proof}
        When $k = l$, by Lemma~\ref{lemma:var_ind} we have
            \begin{small}
            \begin{align*}
                &\;\;\;\;\;\Var\left[\sum_{j=0}^l\frac{w_{L-l+j}}{Y_{L-l}} \sum_{u \in V} p^{(j)}(u,v) \cdot \mathbf{\hat{r}}^{(L-l)}(u) + \sum_{i=0}^{L-l-1}\frac{w_i}{Y_i}\mathbf{\hat{r}}^{(i)}(v)\right]\\
                &=\varepsilon\cdot\sum_{j=0}^{l}\mathbf{q}^{(L-l+j)}(v) \\
                &\;\;\;\;\;+ \underbrace{%
                \begin{array}{l}
            \displaystyle\Var\Bigg[\sum_{j=0}^{l+1}\frac{w_{L-(l+1)+j}}{Y_{L-(l+1)}} \sum_{u \in V} p^{(j)}(u,v) \cdot \mathbf{\hat{r}}^{(L-(l+1))}(u) \\
                \displaystyle \;\;\;\;\;\;\;\;+ \sum_{i=0}^{L-(l+1)-1}\frac{w_i}{Y_i}\mathbf{\hat{r}}^{(i)}(v)\Bigg]\end{array}}_{(5)}\,.
            \end{align*} 
        \end{small}

        Again, we apply the law of Total Variance on Term (5):
        
        \begin{small}
        \begin{align*}
                &\text{Term (5)}\\
                 = &\underbrace{%
                 \begin{array}{l}
                 \displaystyle \E\Bigg[\Var\Bigg[ \sum_{j=0}^{l+1}\frac{w_{L-(l+1)+j}}{Y_{L-(l+1)}}\cdot \sum_{u \in V} p^{(j)}(u,v) \cdot \mathbf{\hat{r}}^{(L-(l+1))}(u)\\
                 \displaystyle \;\;\;\;\;\;\;\;\;\;\;\;\;+ \sum_{i=0}^{L-(l+1)-1}\frac{w_i}{Y_i}\mathbf{\hat{r}}^{(i)}(v)\Bigg|\left\{\mathbf{\hat{r}}^{(L-(l+1)-1)}\right\}\Bigg]\Bigg]
                 \end{array}}_{(6)} \\
                & + \underbrace{
                \begin{array}{l}
                 \displaystyle \Var\Bigg[\E\Bigg[ \sum_{j=0}^{l+1}\frac{w_{L-(l+1)+j}}{Y_{L-(l+1)}}\cdot \sum_{u \in V} p^{(j)}(u,v) \cdot \mathbf{\hat{r}}^{(L-(l+1))}(u)\\
                 \displaystyle \;\;\;\;\;\;\;\;\;\;\;\;\;\; + \sum_{i=0}^{L-(l+1)-1}\frac{w_i}{Y_i}\mathbf{\hat{r}}^{(i)}(v)\Bigg|\left\{\mathbf{\hat{r}}^{(L-(l+1)-1)}\right\}\Bigg]\Bigg]\end{array}}_{(7)}\;.
            \end{align*}
        \end{small}
        We first analyze  Term (6). As $\left\{\mathbf{\hat{r}}^{(L-(l+1)-1)}\right\}$ is given, hence
        \begin{small}
            \begin{align*}
                &\;\;\;\;\;\text{Term (6)} \\
                &= \E\left[\Var\left[ \sum_{j=0}^{l+1}\frac{w_{L-(l+1)+j}}{Y_{L-(l+1)}}\cdot \sum_{u \in V} p^{(j)}(u,v) \cdot \mathbf{\hat{r}}^{(L-(l+1))}(u)\Bigg|\left\{\mathbf{\hat{r}}^{(L-(l+1)-1)}\right\}\right]\right]\,. 
            \end{align*}
        \end{small}
        
        By Lemma~\ref{lemma:var},  $\Var\left[\mathbf{\hat{r}}^{(L-(l+1))}(u)\big|\left\{\mathbf{\hat{r}}^{(L-(l+1)-1)}\right\}\right]  \leq \sum_{w \in N[u]}\varepsilon \cdot \frac{Y_{L-(l+1)}}{Y_{L-(l+1)-1}}\cdot p^{(1)}(w,u) \cdot \mathbf{\hat{r}}^{(L-(l+1)-1)}(w)$, and by the definition of variance, we can further have:
        \begin{small}
            \begin{align*}
                \text{Term (6)} &\leq \E\Bigg[ \sum_{u \in V} \left(\sum_{j=0}^{l+1}\frac{w_{L-(l+1)+j}}{Y_{L-(l+1)}}\cdot  p^{(j)}(u,v)\right)^2 \cdot \sum_{w\in N[u]} \varepsilon \cdot \frac{Y_{L-(l+1)}}{Y_{L-(l+1)-1}} \\
                &\;\;\;\;\;\;\;\;\;\;\;\;\;\cdot p^{(1)}(w,u) \cdot \mathbf{\hat{r}}^{(L-(l+1)-1)}(w)\Bigg]\,.\\
            \end{align*}
        \end{small}
        
By the fact that $p^{(j)}(u,v) \leq 1$ and 
        $\sum_{j=0}^{l+1}\frac{w_{L-(l+1)+j}}{Y_{L-(l+1)}} 
        %= \frac{\sum_{j=0}^{l+1} w_{L-(l+1)+j}}{\sum_{i=L-(l+1)}^\infty w_i} 
        \leq 
        \frac{Y_{L - (l + 1)}}{Y_{L-(l+1)}}  \leq 1$,
        %hence $\sum_{j=0}^{l+1}\frac{w_{L-(l+1)+j}}{Y_{L-(l+1)}}\cdot p^{(j)}(u,v) \leq 1$. With this, we have:
        Term (6) can be bounded as:
        
        \begin{small}
        \begin{align*}
            \text{Term (6)} &\leq \E\Bigg[\varepsilon\cdot \sum_{j=0}^{l+1}\frac{w_{L-(l+1)+j}}{Y_{L-(l+1)}}\cdot \sum_{u \in V} p^{(j)}(u,v) \cdot \sum_{w\in N[u]} \frac{Y_{L-(l+1)}}{Y_{L-(l+1)-1}} \\
            &\;\;\;\;\;\;\;\;\;\;\;\;\;\cdot p^{(1)}(w,u) \cdot \mathbf{\hat{r}}^{(L-(l+1)-1)}(w)\Bigg]\\
            &=\E\left[\varepsilon\cdot \sum_{j=0}^{l+1}\frac{w_{L-(l+1)+j}}{Y_{L-(l+1)}}\cdot \sum_{u \in V} p^{(j)}(u,v) \cdot  \mathbf{\hat{r}}^{(L-(l+1))}(u)\right]\,.
        \end{align*}
        \end{small}
        
        According to  the definition of
        $\mathbf{\hat{r}}^{(i)}(v)$, we have:
\\
$\sum_{u \in V} p^{(j)}(u,v) \cdot \mathbf{\hat{r}}^{(L-(l+1))}(u) = \mathbf{\hat{r}}^{(L-(l+1)+j)}(v) \cdot \frac{Y_{L-(l+1)}}{Y_{L-(l+1)+j}}$. Intuitively, it processes $j$ more steps of propagation to $v$. 
        As a result, Term (6) $\leq \varepsilon \cdot \sum_{j=0}^{l+1}\mathbf{q}^{(L-(l+1)+j)}(v)$. 
        
        We now consider Term (7). Again, as $\left\{\mathbf{\hat{r}}^{(L-(l+1)-1)}\right\}$ is given, and by Lemma~\ref{lemma:epe}, 
        \begin{small}
        \begin{align*}
        &\;\;\;\;\text{Term (7)} \\
        &= \Var\Bigg[\E\Bigg[ \sum_{j=0}^{l+1}\frac{w_{L-(l+1)+j}}{Y_{L-(l+1)}} \sum_{u \in V} p^{(j)}(u,v) \cdot \mathbf{\hat{r}}^{(L-(l+1))}(u)
        \big|\left\{\mathbf{\hat{r}}^{(L-(l+1)-1)}\right\}
        \Bigg]\\
        &\;\;\;\;\;\;\;\;\;\;\;\;\;\;\;\;+ \sum_{i=0}^{L-(l+1)-1}\frac{w_i}{Y_i}\mathbf{\hat{r}}^{(i)}(v)\Bigg]\\
        &=\Var\Bigg[\sum_{j=0}^{l+1}\frac{w_{L-(l+1)+j}}{Y_{L-(l+1)}} \sum_{u \in V} p^{(j)}(u,v) \sum_{w\in N[u]} \frac{Y_{L-(l+1)}}{Y_{L-(l+1)-1}} \cdot p^{(1)}(w,u)\\
        &\;\;\;\;\;\;\;\;\;\;\;\;\;\;\;\cdot \mathbf{\hat{r}}^{(L-(l+1)-1)}(w)
        + \sum_{i=0}^{L-(l+1)-1}\frac{w_i}{Y_i}\mathbf{\hat{r}}^{(i)}(v)\Bigg]\,.
        \end{align*}
        \end{small}
        
        \noindent With $\sum_{w\in N[u]} \frac{Y_{L-(l+1)}}{Y_{L-(l+1)-1}} \cdot p^{(1)}(w,u) \cdot \mathbf{\hat{r}}^{(L-(l+1)-1)}(w) = \mathbf{\hat{r}}^{(L-(l+1))}(u)$, we can rewrite $\sum_{j=0}^{l+1}\frac{w_{L-(l+1)+j}}{Y_{L-(l+1)}} \sum_{u \in V} p^{(j)}(u,v) \sum_{w\in N[u]} \frac{Y_{L-(l+1)}}{Y_{L-(l+1)-1}} \cdot p^{(1)}(w,u) \cdot \mathbf{\hat{r}}^{(L-(l+1)-1)}(w)$ as $\sum_{j=1}^{l+2}\frac{w_{L-(l+2)+j}}{Y_{L-(l+2)}} \sum_{u \in V} p^{(j)}(u,v) \cdot \mathbf{\hat{r}}^{(L-(l+2))}(u)$. We can then split $\frac{w_{L-(l+1)-1}}{Y_{L-(l+1)-1}}\mathbf{\hat{r}}^{(L-(l+1) -1)}(v)$\\$ = \sum_{j=0}^{0}\frac{w_{L-(l+2)+j}}{Y_{L-(l+2)}} \sum_{u \in V} p^{(j)}(u,v) \cdot \mathbf{\hat{r}}^{(L-(l+2))}(u)$ from the right term, $\sum_{i=0}^{L-(l+1)-1}\frac{w_i}{Y_i}\mathbf{\hat{r}}^{(i)}(v)$, and adds it to the left term to get:
        \begin{small}
        \begin{align*}
         \text{Term (7)} = \Var\Bigg[&\sum_{j=0}^{l+2}\frac{w_{L-(l+2)+j}}{Y_{L-(l+2)}} \sum_{u \in V} p^{(j)}(u,v) \cdot \mathbf{\hat{r}}^{(L-(l+2))}(u)\\
         &+ \sum_{i=0}^{L-(l+2)-1}\frac{w_i}{Y_i}\mathbf{\hat{r}}^{(i)}(v)\Bigg]\,.
        \end{align*}
        \end{small}
        
        %\vspace{-2mm} 
\noindent 
        Therefore, we have $\text{Term (5)} = \text{Term (6)} + \text{Term (7)}$ \\ $ \leq \varepsilon \cdot \sum_{j=0}^{l+1}\mathbf{q}^{(L-(l+1)+j)}(v) + \Var\Bigg[\sum_{j=0}^{l+2}\frac{w_{L-(l+2)+j}}{Y_{L-(l+2)}} \sum_{u \in V} p^{(j)}(u,v) \cdot \mathbf{\hat{r}}^{(L-(l+2))}(u)  + \sum_{i=0}^{L-(l+2)-1}\frac{w_i}{Y_i}\mathbf{\hat{r}}^{(i)}(v)\Bigg]$. 

This completes the proof of Lemma~\ref{lemma:mathind}.
\end{proof} 
   
By  
Lemma~\ref{lemma:base} (the base case when $k = 0$) and Lemma~\ref{lemma:mathind} (the inductive case), the mathmetical induction proof for Lemma~\ref{lemma:var_ind} establishes.
%we then prove Lemma~\ref{lemma:var_ind}.
%
Finally, by Lemma~\ref{lemma:var_ind}, we can bound the variance:
\begin{small}
\begin{align*}
    \Var[\boldsymbol{\hat{\pi}}(v)] &=\Var\left[\sum_{i=0}^L\frac{w_i}{Y_i}\cdot \mathbf{\hat{r}}^{(i)}(v)\right]\\
    &=\Var\left[\sum_{j=0}^0\frac{w_{L+j}}{Y_{L}} \sum_{u \in V} p^{(j)}(u,v) \mathbf{\hat{r}}^{(L)}(u) + \sum_{i=0}^{L-1}\frac{w_i}{Y_i} \mathbf{\hat{r}}^{(i)}(v)\right]\\
    &=\varepsilon  \sum_{j=0}^0\mathbf{q}^{(L-1+j)}(v)+  \cdots +\varepsilon \sum_{j=0}^{L-1}\mathbf{q}^{(1+j)}(v)\\
    &\;\;\;\;\; +\Var\left[\sum_{j=0}^Lw_j  \sum_{u \in V} p^{(j)}(u,v) \mathbf{\hat{r}}^{(0)}(u)\right] \\
    &= \varepsilon  \sum_{k=0}^{L}\sum_{j=0}^{k}\mathbf{q}^{(L-k+j)}(v) + \Var\left[\sum_{j=0}^Lw_j  \sum_{u \in V} p^{(j)}(u,v) \mathbf{\hat{r}}^{(0)}(u)\right]\;.
\end{align*}
\end{small}

 % $$\Var[\boldsymbol{\hat{\pi}}(v)]  
 %    =\varepsilon \cdot \mathbf{q}^{(L)}(v) + \varepsilon \cdot \sum_{j=0}^1\mathbf{q}^{(L-1+j)}(v) + \varepsilon \cdot \sum_{j=0}^2\mathbf{q}^{(L-2+j)}(v) + \cdots +\varepsilon \cdot \sum_{j=0}^{L-1}\mathbf{q}^{(1+j)}(v) +\sum_{w \in V}\left(\sum_{j=0}^Lw_j \cdot p^{(j)}(w,v)\right)^2\cdot Var[\mathbf{\hat{r}}^{(0)}(w)] 
 %    = \varepsilon \cdot \sum_{k=0}^{L}\sum_{j=0}^{k}\mathbf{q}^{(L-k+j)}(v) + \Var\left[\sum_{j=0}^Lw_j \cdot \sum_{u \in V} p^{(j)}(u,v)\cdot \mathbf{\hat{r}}^{(0)}(u)\right]$$ 

Since $\varepsilon \cdot \sum_{k=0}^{L}\sum_{j=0}^{k}\mathbf{q}^{(L-k+j)}(v) \leq \varepsilon \cdot L\sum_{j=0}^{L}\mathbf{q}^{(j)}(v) \leq \varepsilon \cdot L \cdot \boldsymbol{\pi}(v)$. And by the definition of variance, we have:
\begin{small}
\begin{align*}
    \Var\left[\sum_{j=0}^Lw_j \cdot \sum_{u \in V} p^{(j)}(u,v)\cdot \mathbf{\hat{r}}^{(0)}(u)\right] 
    = \sum_{u \in V}\left(\sum_{j=0}^Lw_j \cdot p^{(j)}(u,v)\right)^2\cdot \Var[\mathbf{\hat{r}}^{(0)}(u)]\,.
\end{align*} 
\end{small}

\noindent
Therefore, 
\begin{small}
\begin{align*}
\Var[\boldsymbol{\hat{\pi}}(v)] \leq \varepsilon \cdot L \cdot \boldsymbol{\pi}(v)  +\sum_{u \in V}\left(\sum_{j=0}^Lw_j \cdot p^{(j)}(u,v)\right)^2\cdot~\Var[\mathbf{\hat{r}}^{(0)}(u)]\,.
\end{align*}
\end{small}
%=======
%Since $\varepsilon \cdot \sum_{k=0}^{L}\sum_{j=0}^{k}\mathbf{q}^{(L-k+j)}(v) \leq \varepsilon \cdot L\sum_{j=0}^{L}\mathbf{q}^{(j)}(v) = \varepsilon \cdot L \cdot \boldsymbol{\pi}_L(v)$, and by the definition of variance, we have $\Var\left[\sum_{j=0}^Lw_j \cdot \sum_{u \in V} p^{(j)}(u,v)\cdot \mathbf{\hat{r}}^{(0)}(u)\right] = \sum_{u \in V}\left(\sum_{j=0}^Lw_j \cdot p^{(j)}(u,v)\right)^2\cdot \Var[\mathbf{\hat{r}}^{(0)}(u)]$, we then have $\Var[\boldsymbol{\hat{\pi}}(v)] \leq \varepsilon \cdot L \cdot \boldsymbol{\pi}_L(v)  +\sum_{u \in V}\left(\sum_{j=0}^Lw_j \cdot p^{(j)}(u,v)\right)^2\cdot \Var[\mathbf{\hat{r}}^{(0)}(u)]$.
%>>>>>>> e8e826b94bac0bba1a068dde4f248794d5c9e1f6

\noindent
Here,
$\Var[\mathbf{\hat{r}}^{(0)}(u)]$ is related to the initialization method. 
When 
%a deterministic initialization for 
$\mathbf{\hat{r}}^{(0)}$ is initialized deterministically (e.g., Line 1 in Algorithm~\ref{alg:imp_agp}), 
$\Var[\boldsymbol{\hat{r}}^{(0)}(u)] = 0$, and thus,
$\Var[\boldsymbol{\hat{\pi}}(v)] \leq \varepsilon \cdot L \cdot \boldsymbol{\pi}(v)$.

Moreover, 
%in our technical report~\cite{sourceCode},
in Section~\ref{sec:rand},
we propose a randomized initialization algorithm for some special cases when the input $\mathbf{x}$ satisfies certain conditions.
We prove that, with our randomized initialization, 
$\Var[\boldsymbol{\hat{\pi}}(v)] \leq \varepsilon \cdot (L + 1) \cdot \boldsymbol{\pi}(v)$ holds.

\begin{comment}
For the case when 
$\mathbf{\hat{r}}^{(0)}$ is initialized by our randomized initialization algorithm
(see in our technical report~\cite{sourceCode}),
%our randomized initialization is applied, 
we prove that, 
%in our technical report~\cite{sourceCode},
$\Var[\boldsymbol{\hat{\pi}}(v)] \leq \varepsilon \cdot (L + 1) \cdot \boldsymbol{\pi}(v)$ holds.
\end{comment}

Therefore, in either initialization way, by Chebyshev's Inequality~\cite{ross2020first},
when $\varepsilon = \frac{c^2 \delta}{2(L + 1) } = O(\frac{\delta}{L})$, we have:
\begin{small}
\begin{align*}
    \pr[|\boldsymbol{\pi}(v) - \boldsymbol{\hat{\pi}}(v)| > c\cdot \boldsymbol{\pi}(v)] 
    &\leq \frac{\Var[\boldsymbol{\hat{\pi}}(v)]}{c^2\cdot \boldsymbol{\pi}(v)^2} \leq \frac{\varepsilon  (L+1)}{c^2 \cdot \boldsymbol{\pi} (v)} 
    \leq \frac{\varepsilon (L+1)}{c^2\delta} =  O(1)\,.
\end{align*}
\end{small} 
%Refer to Appendix \alvin{TBD}, if we apply our randomized initialization, it only adds a constant to make the whole variance bound to $ \varepsilon \cdot (L + 1) \cdot \boldsymbol{\pi}_L(v)$, and the inequality still holds. 

\noindent
Together with Fact~\ref{fact:delta} and Algorithm~\ref{alg:imp_agp}, 
the correctness of our \algosta\ follows and is summarized in the theorem below:

\vspace{2mm}
\begin{theorem}
For any query $q(a, b, \mathcal{O}_w, \mathbf{x})$, 
our \algosta\ can return an $(\delta, c)$-approximation
$\boldsymbol{\hat{\pi}}$ 
with at least a constant probability. 
%satisfying 
%with $\varepsilon = O(\frac{\delta}{L})$.
\end{theorem}

%and hence, the correctness of our \algosta\ follows. 
%and achieves the query time complexity in Corollary~\ref{col:our_query}.
%where $\cso = O(\log n)$ for \algosta\ and $\cso = O(1)$ for \algodpss.

% Together with Fact~\ref{fact:delta} and Theorem~\ref{theorem:query}, we prove the theorem:
% \begin{theorem} 
%     For every query $q(a, b, \mathcal{O}_w, \mathbf{x})$, 
% with as at least contant probability, 
% \algosta\  and \algodpss\ can return an  $(\delta, c)$-approximation $\mathbf{\hat{\pi}}$ of $\mathbf{\pi}$ 
%     %for all $v \in V \wedge |\boldsymbol{\pi}(v)| > \delta$ that satisfies
%     %$|\boldsymbol{\pi}(v) - \boldsymbol{\hat{\pi}}(v)| \leq c \cdot \boldsymbol{\pi}(v) $
%     %with at least a constant success rate in expected 
%     in 
% $$O(\cso \cdot \frac{L}{\delta} \cdot \sum_{i=1}^L \|Y_i(\mathbf{D}^{-a}\mathbf{A}\mathbf{D}^{-b})^i\cdot \mathbf{x}\|_1 + \E[C_{\text{init}}])$$ expected time,
% where $\cso = O(\log n)$ for \algosta\ and $\cso = O(1)$ for \algodpss.
% %, where $c$ is a constant. 
% %When $\mathbf{x}$ satisifes Condition~\ref{cond}, $\E[C_{\text{init}}] = \min\{\frac{1}{\varepsilon}, n\}$, otherwise $C_{\text{init}} = n$.
% \end{theorem}

\section{AGP on Dynamic Graphs}
\label{sec:dyna_agp}

 While our \algosta~can achieve good query time complexity, it does not work well with graph updates.
%\alvin{With our improved static algorithm, we now introduce a dynamic algorithm to efficiently handle dynamic graphs. Recall that the static query algorithms require pre-processing on the graph. 
%When the graph is updated, the previously pre-processed information may be stale and hence, the static query algorithm cannot work properly.}
%Here each update is either an insertion or a deletion of an edge in the graph $G$.
%
To see this, suppose there is an update (either an insertion or a deletion) of an edge $(u,w)$;
due to this  
update, the degree $d_u$ (resp., $d_w$) changes, and hence, the sampling probabilities of $u$ within the neighborhoods $N[v]$ of all $u$'s neighbors $v$ change accordingly.
As a result, while the subset sampling structures in \algosta~
%and \algodpss~
can be updated efficiently in $O(1)$ time,
the overall cost for each graph update incurs $O(d_u)$ (resp., $O(d_w)$) time, which can be as large as $O(n)$ in the worst case.
This limits the applications of \algosta\ to scenarios where the underlying graph $G$ is updated frequently.

To address this challenge, we introduce a strengthened version of \algosta, called \algodyn, which can support each edge insertion or deletion in $O(1)$ amortized time (substantially improving the aforementioned $O(n)$ update time bound). 
The core idea of \algodyn~is an observation that
it is indeed {\em  unnecessary} to update the sampling probability for $u$ within the neighborhoods of $u$'s neighbors for {\em every} graph update related to $u$, as long as 
our algorithm can still have a {\em good overestimate} $p^*$ of $p_{v,u}$
such that $p_{\text{ac}} = \frac{p_{v,u}}{p^*} \in \Omega(1)$.
If this is the case, 
one can still first sample $u$ with probability $p^*$ and set the accept probability $p_{\text{ac}}$ with respect to $u$'s current degree.
%to make the overall probability $p_{v,u}$ 
According to our theoretical analysis, 
the subset sampling complexity within each neighborhood would remain the same, and so as the overall expected query time bound.

\vspace{2mm}
\noindent
{\bf Performing Updates.}
To achieve this, \algodyn~maintains a {\em reference degree}, denoted by $\tilde{d}_u$, for each vertex $u \in V$, 
which records the degree value of $u$ when the {\em last} update for $u$ within the neighborhoods of $u$'s neighbors, and initially, $\tilde{d}_u = d_u$.
The detailed implementation is shown in Algorithm~\ref{alg:update}.

\begin{small}
        \begin{algorithm}
        \caption{\algodyn}\label{alg:update}
        \DontPrintSemicolon
        \SetKwComment{Comment}{/* }{ */}
        \KwIn{an edge update $(u, w)$}
        \lIf{$(u,w)$ is a deletion}{
        remove $u$ from $B[w]$ and
        $d_u \leftarrow d_u - 1;$}   
        \lElse{
        insert $u$ to $B[w]$ and
            $d_u \leftarrow d_u + 1;$}
        
        \If{$d_u < \frac{1}{2}\cdot \tilde{d}_u$ or $d_u > 2\cdot \tilde{d}_u$}
        {
            \For{each $v \in N[u]$}
            {
                update $B[v]$ based on  current $d_u$;
            }
            $\tilde{d}_u \leftarrow d_u$;
        }
        handle $w$ symmetrically
        
        \end{algorithm}
        \end{small}

\begin{theorem}
\label{theorem:update}%{\em {[*}]}
    Algorithm~\ref{alg:update} can process each graph update in $O(1)$ amortized time. 
\end{theorem}

\begin{proof}
It is known that, the sorted linked list $B[u]$ can be maintained using $O(|B[u]|)$ space, and supports each insertion or deletion of a bucket into or from $B[u]$
%, and retrieval of a pointer to the largest bucket $B_i \in B[u]$ with an index $i \leq j$ for any given integer $0 \leq j \leq \lceil \log_2 n \rceil$ 
in $O(1)$ time~\cite{gan2024optimal}. 
With that, removing or inserting an element in $B[u]$ (Lines 1 and 2) takes $O(1)$. 
Therefore, updating $u$ in $B[v]$ within the neighborhoods of $u$'s neighbors (Lines 4-5) takes $O(d_u)$ time.
However, according to the condition in Line 3, such an update cost only incurs when there are $\Omega(d_u)$ updates related to $u$ occurred, because, otherwise, $d_u$ would not have been changed more than a factor of $2$ from $\tilde{d}_u$.
As a result, this $O(d_u)$ update cost can be charged to these $\Omega(d_u)$ updates, and thus, the amortized update cost becomes $O(1)$.   
\end{proof}

\vspace{2mm}
\noindent
{\bf Performing Queries.}
Observe that 
subset sampling structure within the neighborhood $N[v]$ for each $v \in V$ is built on the reference degrees $\tilde{d}_u$ of $v$'s neighbors $u \in N[v]$. 
As a result, the previous overestimate probability $p^*$ of $p_{v,u}$ used in \algosta, 
which is computed based on the bucket index storing $u$ in $B[v]$, 
is no longer guaranteed to be feasible. 
This is because the current degree $d_u$ can now be smaller than $\tilde{d}_u$, making the actual sampling probability $p_{v,u} > p^*$.
Fortunately, 
one crucial observation here is that, according to Line 3 in Algorithm~\ref{alg:update}, $d_u \geq \frac{1}{2}\tilde{d}_u$ always holds.
Therefore,
$p_{v, u} \leq 2^a \cdot p^*$ holds, and thus, ${p^{*}}' = 2^a \cdot p^*$
can be used as a good overestimate of $p_{v,u}$ for the subset sampling within $N[v]$ for $u$.
Hence, the correctness of the query result of \algodyn~is guaranteed.
Moreover, as discussed earlier, 
the expected query time complexity of \algodyn~is the same as that of \algosta.
%using ${p^*}'$ can
%would not change the query time complexity of \algosta.

%\begin{comment}
\section{Randomized Initialization}
\label{sec:rand}

In \algoagp, $\mathbf{\hat{r}}^{(0)}$ is %naively 
initialized to $\mathbf{x}$ taking $O(n)$ time, which is not guaranteed to be dominated by $\cso \cdot \frac{1}{\varepsilon} \cdot \sum_{i=1}^{L} \big\| Y_i (D^{-a} A D^{-b})^i \mathbf{x} \big\|_1 $. In certain applications, $\mathbf{x}$ admits a {\em compact representation}, making a full read unnecessary. A notable example is the  PageRank~\cite{Page1999ThePC}. We exploit this property and propose a sub-linear time randomized algorithm for initialization, when  
 $\mathbf{x}$ satisfies Condition~\ref{con:x}.

\begin{condition}\label{con:x}
    $\mathbf{x}$ satisfies following conditions and can be represented as $P$ and $S$:
    \begin{itemize}[leftmargin=*]
        \item $\mathbf{x}$ can be partitioned into $O(\frac{1}{\varepsilon})$ index groups denoted as $P$, where $P_i \subseteq \{1,2,\cdots,n\}$ indicates the $i^{th}$ group;
        \item for each $P_i$, the indexes are continuative;
        \item for any two indexes $j, k \in P_i$, we have $\mathbf{x}_j = \mathbf{x}_k$ stored at $S_i$.
\end{itemize}
\end{condition}

As shown in Algorithm~\ref{alg:initil}, our randomized initialization algorithm works as follows. When 
$\frac{\mathbf{x}(v)}{\varepsilon} \geq 1$, the exact value of $\mathbf{x}(v)$ is assigned to 
$\mathbf{\hat{r}}^{(0)}(v)$. Otherwise, subset sampling with uniform probability is performed on 
partition $P_i$, where each index is sampled with probability 
$\frac{S_i}{\varepsilon}$. Once an index $j$ is selected, 
$\varepsilon$ is assigned to $\mathbf{x}_j$.
Trivially, this algorithm assures that $E[\mathbf{\hat{r}}^{(0)}(v)] = \mathbf{x}(v)$ for any $v \in V$, while it introduces additional variance to $Var[\boldsymbol{\hat{\pi}}]$ for the whole algorithm. However, as will be shown next, this does not affect the overall error bound and query complexity bound of \algosta.

\begin{small}
\begin{algorithm}
\caption{\emph{Randomized Initialization}}\label{alg:initil}
\DontPrintSemicolon
\SetKwComment{Comment}{/* }{ */}

\KwIn{$P$, $S$, and a parameter $\varepsilon$}
\KwOut{$\mathbf{\hat{r}}^{(0)}$}
\For{$i = 0$ to $|P|$}{
    \If{$\frac{S_i}{\varepsilon} \geq 1$}{
        \For{$j \in P_i$}{
            $\mathbf{\hat{r}}^{(0)}_j \leftarrow S_i;$ \;
        }
    }
    \Else{
        $p = \frac{S_i}{\varepsilon};$\;
        $K = SubsetSample(p, |P|);$\;
        \For{$j \in K$}{
            $\mathbf{\hat{r}}^{(0)}_j \leftarrow \varepsilon;$ \;
        }
    }
}
\Return $\mathbf{\hat{r}}^{(0)}$

\end{algorithm}
\end{small}

%\vspace{1mm}
\noindent
\textbf{Variance Analysis.}
When Algorithm~\ref{alg:initil} is applied, the variance comes from only when $\frac{\mathbf{x}(u)}{\varepsilon} < 1$, as it is a deterministic term when $\frac{\mathbf{x}(u)}{\varepsilon} \geq 1$.  Therefore, $\Var[\mathbf{\hat{r}}^{(0)}(u)] \leq E[(\mathbf{\hat{r}}^{(0)})^2] \leq \varepsilon \cdot \mathbf{{r}}^{(0)}(u)$, where $E[\mathbf{\hat{r}}^{(0)}(u)] = \mathbf{x}(u)$. 
Moreover, since $\Sigma_{j=0}^L w_j \leq 1$ and $p^{(j)}(u,v) \leq 1$, we have 
\begin{small}
\begin{align*}
    &\;\;\;\;\;\;\sum_{u \in V}\left(\sum_{j=0}^Lw_j \cdot p^{(j)}(u,v)\right)^2\cdot \Var[\mathbf{\hat{r}}^{(0)}(u)]\\
    &\leq \varepsilon \cdot\sum_{u \in V}\left(\sum_{j=0}^Lw_j \cdot p^{(j)}(u,v)\right)\cdot  \mathbf{r}^{(0)}(u)
    = \varepsilon\cdot \boldsymbol{\pi}(v)
\end{align*}
\end{small}

\vspace{1mm}
\noindent
\textbf{Complexity Analysis.} There are at most $O(\frac{1}{\varepsilon})$ number of values in $\mathbf{x}$ where $\frac{\mathbf{x}(v)}{\varepsilon} \geq 1$, otherwise their sum already exceeds 1 -- recall that $\|\mathbf{x}\|_1 = 1$.  When running the subset sampling for each partition, the expected time complexity is $O(\mu_i + 1)$, where $\mu_i$ is the expected output size, and we have $\mu_i = |P_i| \cdot \frac{S_i}{\varepsilon}$. Hence, the total expected time complexity for all buckets is bounded by $O(\sum_{i}\mu_i + \frac{1}{\varepsilon}) = O(\frac{1}{\varepsilon}\cdot (\sum_i|P_i|\cdot S_i + 1)) = O(\frac{1}{\varepsilon})$. 
%again as we have $\|\mathbf{x}\|_1 = 1$. 
Therefore, the complexity of Algorithm~\ref{alg:initil} is bounded by $O(\frac{1}{\varepsilon})$. 

\section{Experiments}\label{sec:exp}

\begin{table}
\centering
\caption{Dataset Statistics}
\vspace{-4mm}
\scalebox{0.8}{%
\begin{tabular}{l|r|r|r|c}
\toprule
\multicolumn{1}{c|}{\textbf{Datasets}}         & $\mathbf{n\ (\times10^6)}$     & $\mathbf{m\ (\times10^6)}$   & \multicolumn{1}{c|}{$\boldsymbol{\Bar{d}}$} & \textbf{Domain} \\\hline
\texttt{soc-Slashdot0811} (\textbf{Sl}) & 0.08 & 0.47 &  12.13 & Social network\\\hline
\texttt{web-NotreDame}  (\textbf{ND})  & 0.33  & 1.09 & 6.69 & Website hyperlink\\\hline
\texttt{web-Google}    (\textbf{Go})   & 0.88  & 4.32 & 9.86 & Website hyperlink\\\hline
\texttt{wiki-Topcats}   (\textbf{To})  & 1.79 & 25.44 & 28.38 & Website hyperlink \\\hline
\texttt{soc-Pokec}    (\textbf{Po})    & 1.63 & 22.30 & 27.36 & Social network \\\hline
\texttt{as-Skitter}   (\textbf{Sk})    & 1.70 & 11.10 & 13.06 & Traceroute graph\\\hline
\texttt{wiki-Talk}   (\textbf{Ta})     & 2.39 & 4.66 & 3.90& Interaction graph\\\hline
\texttt{soc-Orkut}    (\textbf{Or})        & 3.07 & 117.19  & 76.22 & Social network\\\hline
\texttt{soc-LiveJournal1} (\textbf{LJ}) & 4.85 & 42.85 & 17.69 & Social Network\\\bottomrule
\end{tabular}}
\label{tab:dataset}
\vspace{-3mm}
\end{table}

\noindent\textbf{Datasets.}
We use nine real-world datasets~\cite{snapnets}, 
% (\texttt{SI}, \texttt{ND}, \texttt{Go}, \texttt{To}, \texttt{Po}, \texttt{Sk}, \texttt{Ta}, \texttt{Or} and \texttt{LJ})
\texttt{soc-Slashdot0811} (\textbf{Sl}), 
\texttt{web-NotreDame}  (\textbf{ND}), 
\texttt{web-Google} (\textbf{Go}),
\texttt{wiki-Topcats}   (\textbf{To}),  
\texttt{soc-Pokec}    (\textbf{Po}),   
\texttt{as-Skitter}   (\textbf{Sk}),    
\texttt{wiki-Talk}   (\textbf{Ta}),     
\texttt{soc-Orkut}    (\textbf{Or}), 
and       
\texttt{soc-LiveJournal1} (\textbf{LJ}) 
 with up to 4.8 million vertices and 117 million edges, whose statistics can be found in Table~\ref{tab:dataset}. We focus on query, update, and initialization efficiency, which are detailed next. Our source code can be found online~\cite{sourceCode}.

% \noindent\textbf{Experiment Environment.}
% All experiments are conducted on a Ubuntu virtual server with a 2.0 GHz CPU and 64 GB of memory. All source codes are in C++ and compiled with -O3 turned on.

\begin{figure}
\centering
\begin{subfigure}{0.8\linewidth}
    \includegraphics[width=1\textwidth]{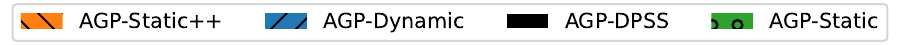}
\end{subfigure}\\
\begin{subfigure}{0.8\linewidth}
    \includegraphics[width=1\textwidth]{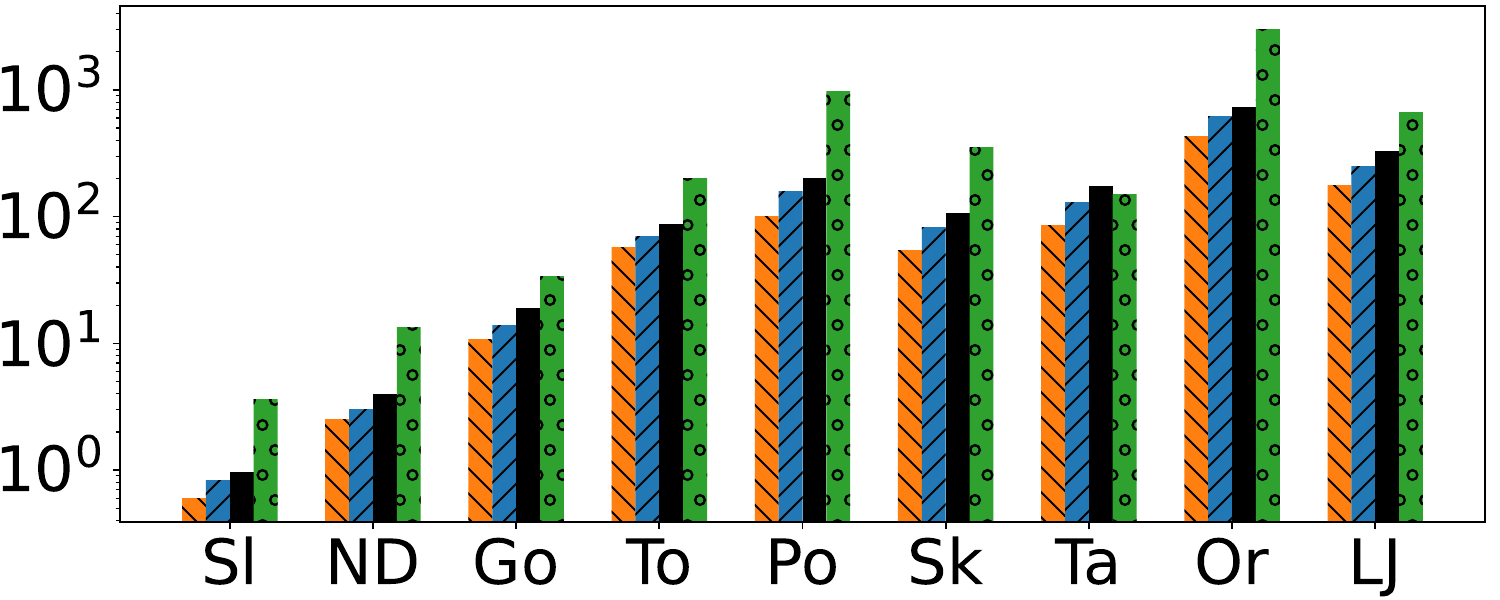}
\end{subfigure}
\vspace{-4mm}
\caption{Average query processing time (second)}
    \label{fig:query}
    \vspace{-4mm}
\end{figure}

\subsection{Query Efficiency}
\noindent\textbf{Baseline.}  We evaluate our algorithms, \algosta\ and \algodyn, against \algoagp~\cite{wang2021approximate} and \algodpss\ (Section~\ref{sec:query_ana}). 
%The latter, \algodpss, is an instantiation of our framework that leverages a state-of-the-art subset sampling algorithm for the Dynamic Parameterized Subset Sampling (DPSS) problem~\cite{gan2024optimal}. 
%By incorporating DPSS, according to Theorem~\ref {theorem:query}, we can further reduce the query complexity by an additional $O(\log n)$ factor. However, \algodpss\ does not support scenarios where the parameter $a$ is provided on the fly. Further details on how DPSS is integrated into our framework can be found in Appendix~\ref{ap:dpss}.

\begin{figure*}
    \centering
    \begin{subfigure}{0.5\linewidth}
        \includegraphics[width=\textwidth]{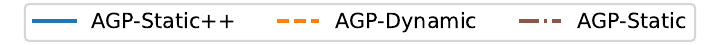}
    \end{subfigure}\\
        \begin{subfigure}{0.17\linewidth}
            \includegraphics[width=\textwidth]{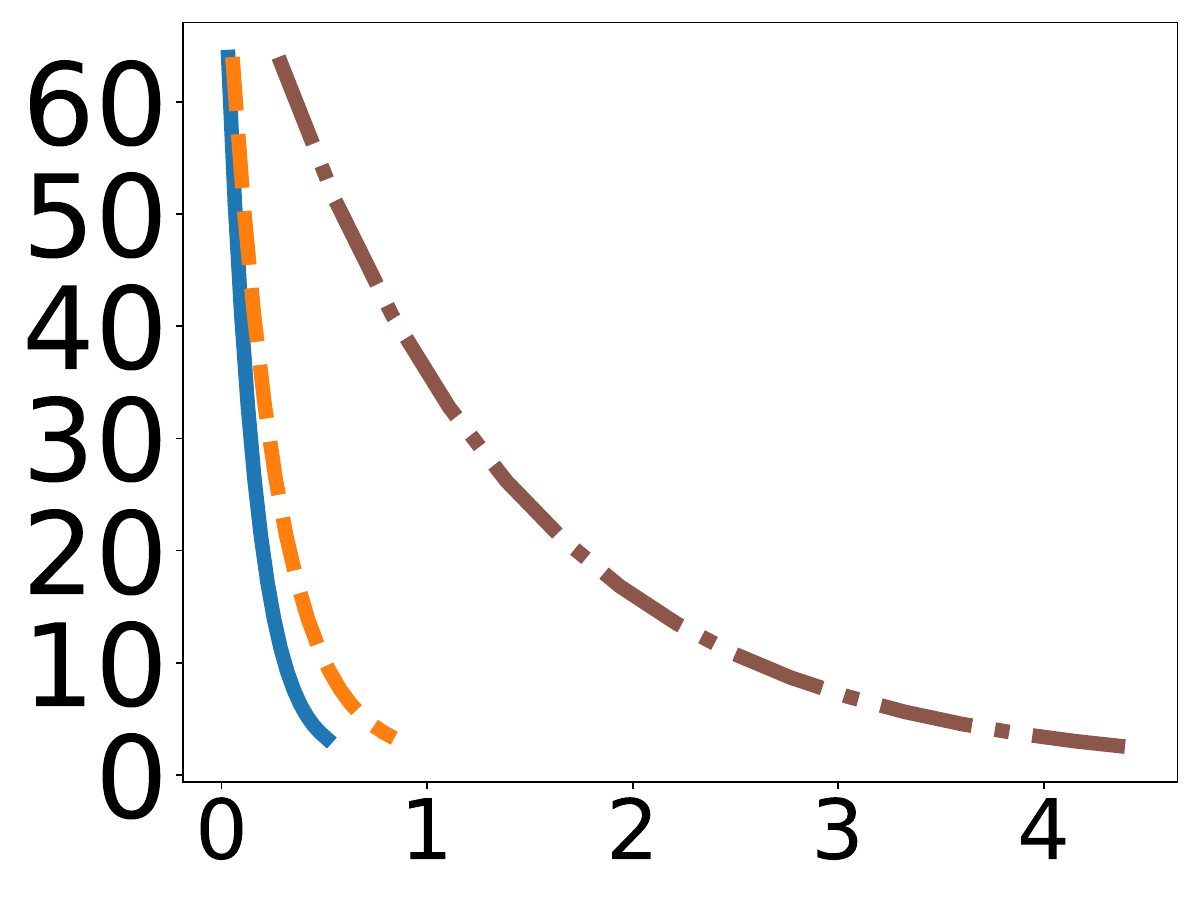}
            \vspace{-5mm}
            \caption{\textbf{Sl}}
        \end{subfigure}
            \begin{subfigure}{0.17\linewidth}
            \includegraphics[width=\textwidth]{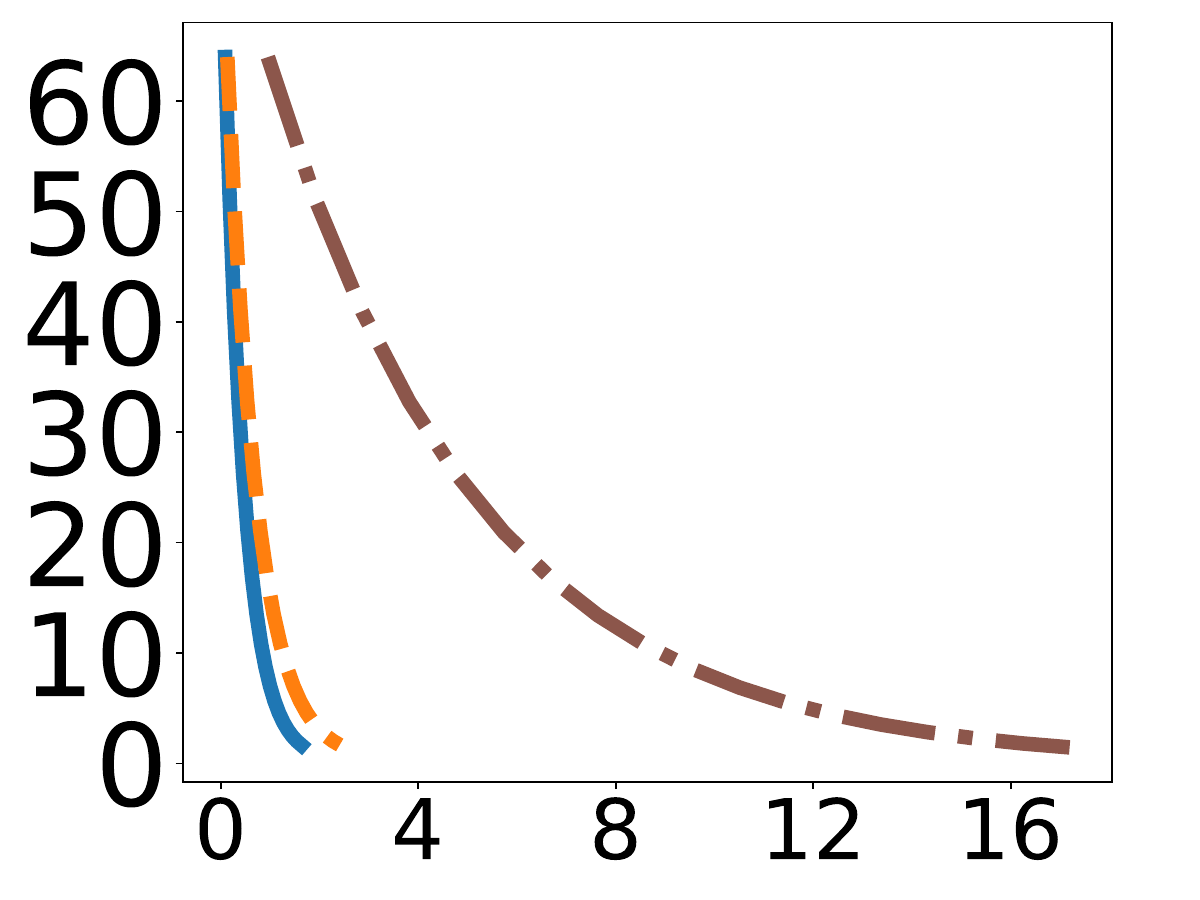}
            \vspace{-5mm}
                \caption{\textbf{ND}}
        \end{subfigure}
            \begin{subfigure}{0.17\linewidth}
                \includegraphics[width=\textwidth]{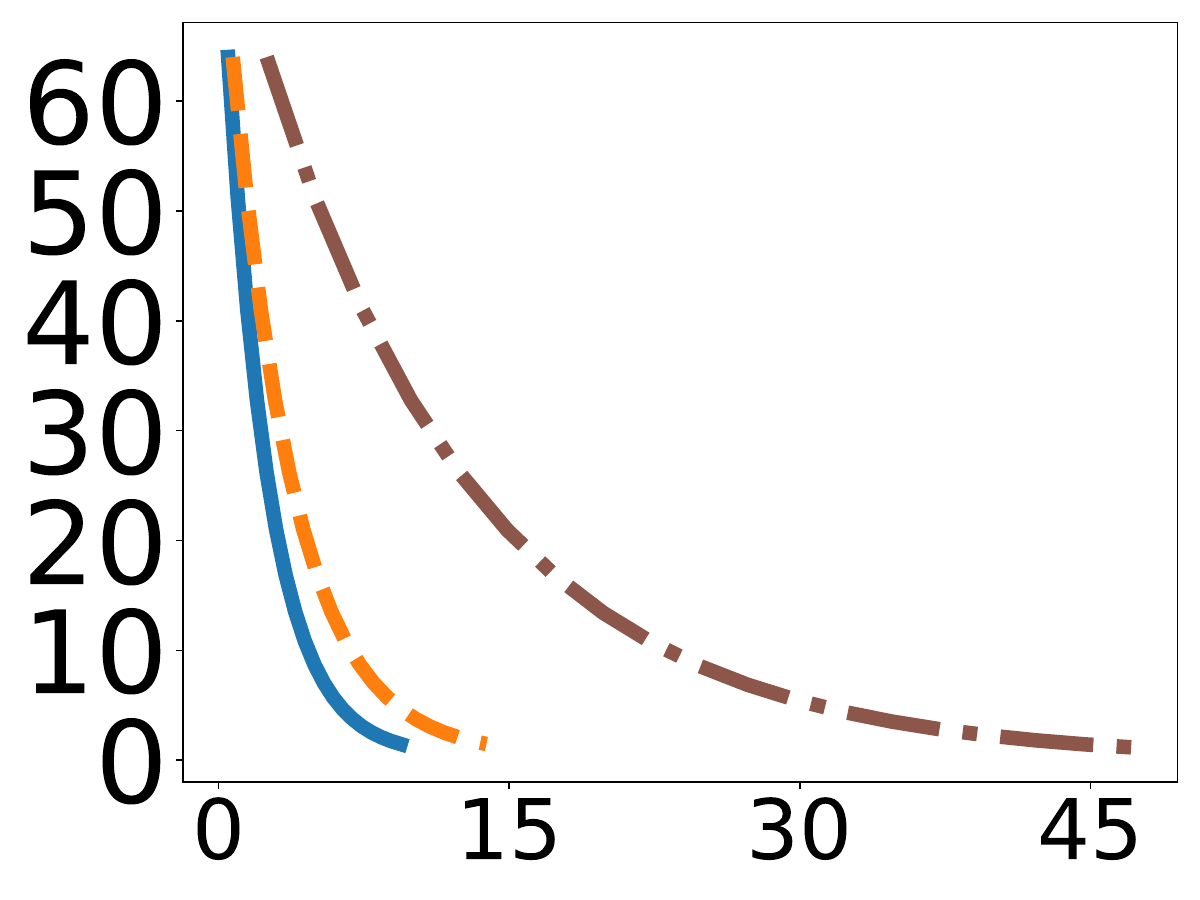}
                \vspace{-5mm}
                \caption{\textbf{Go}}
        \end{subfigure}
            \begin{subfigure}{0.17\linewidth}
                \includegraphics[width=\textwidth]{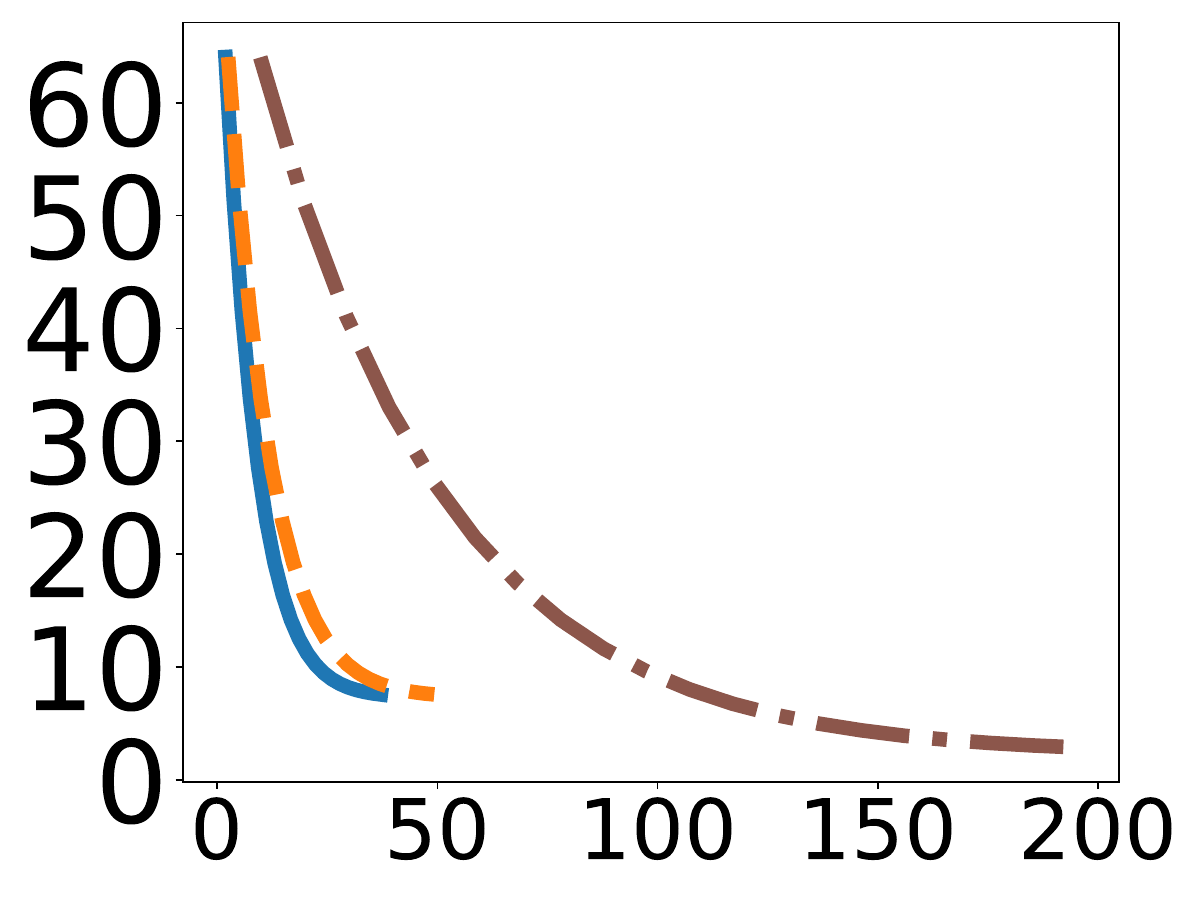}
                \vspace{-5mm}
                \caption{\textbf{To}}
        \end{subfigure}
            \begin{subfigure}{0.17\linewidth}
                \includegraphics[width=\textwidth]{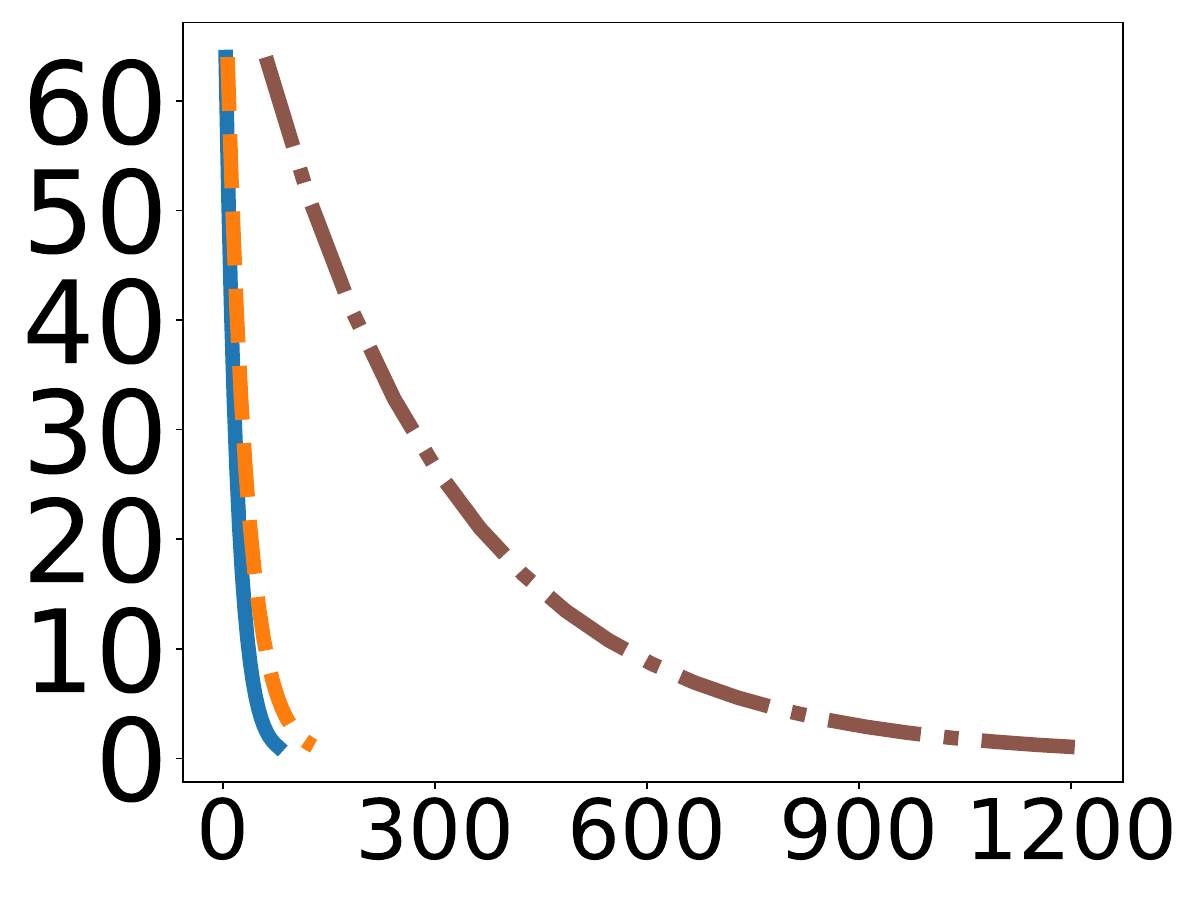}
                \vspace{-5mm}
                \caption{\textbf{Po}}
        \end{subfigure}
        \begin{subfigure}{0.17\linewidth}
            \includegraphics[width=\textwidth]{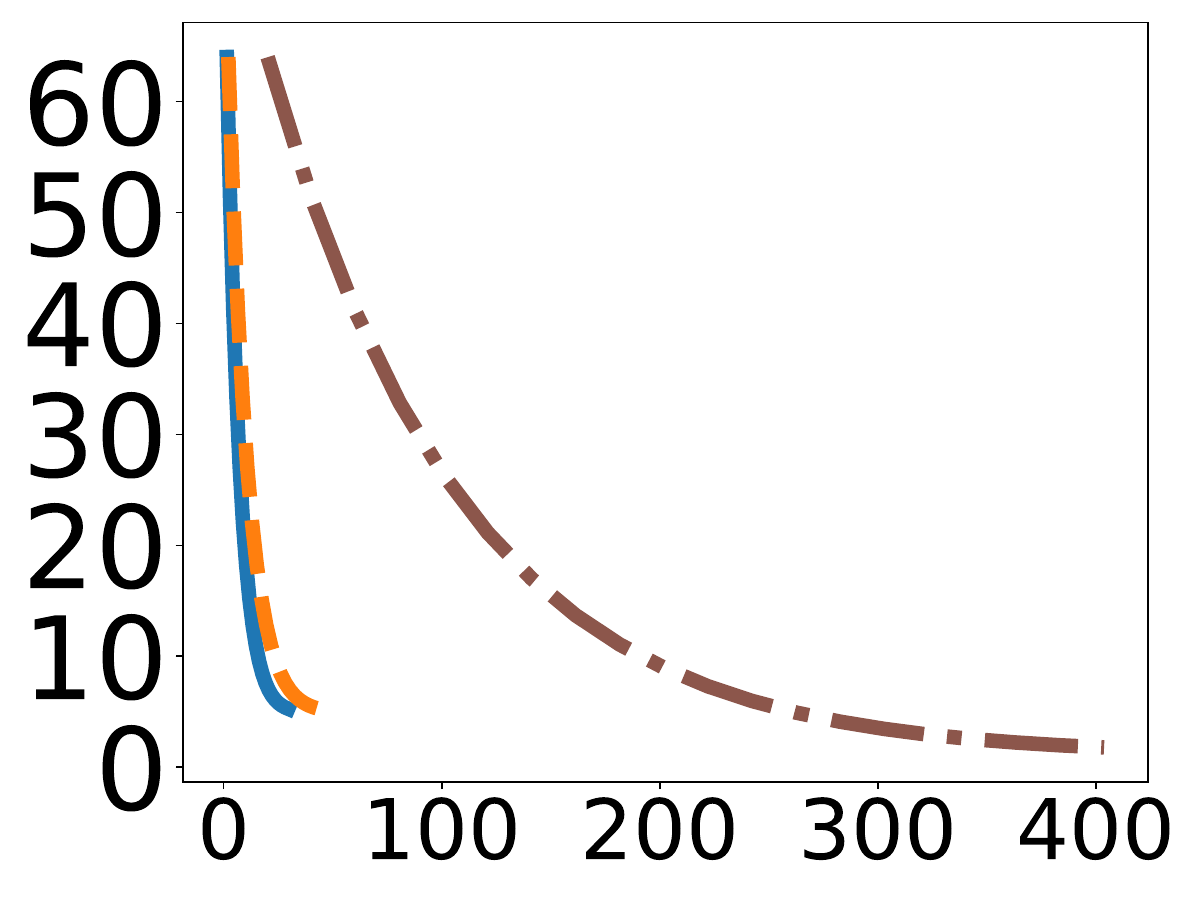}
            \vspace{-5mm}
            \caption{\textbf{Sk}}
        \end{subfigure}
        \begin{subfigure}{0.17\linewidth}
            \includegraphics[width=\textwidth]{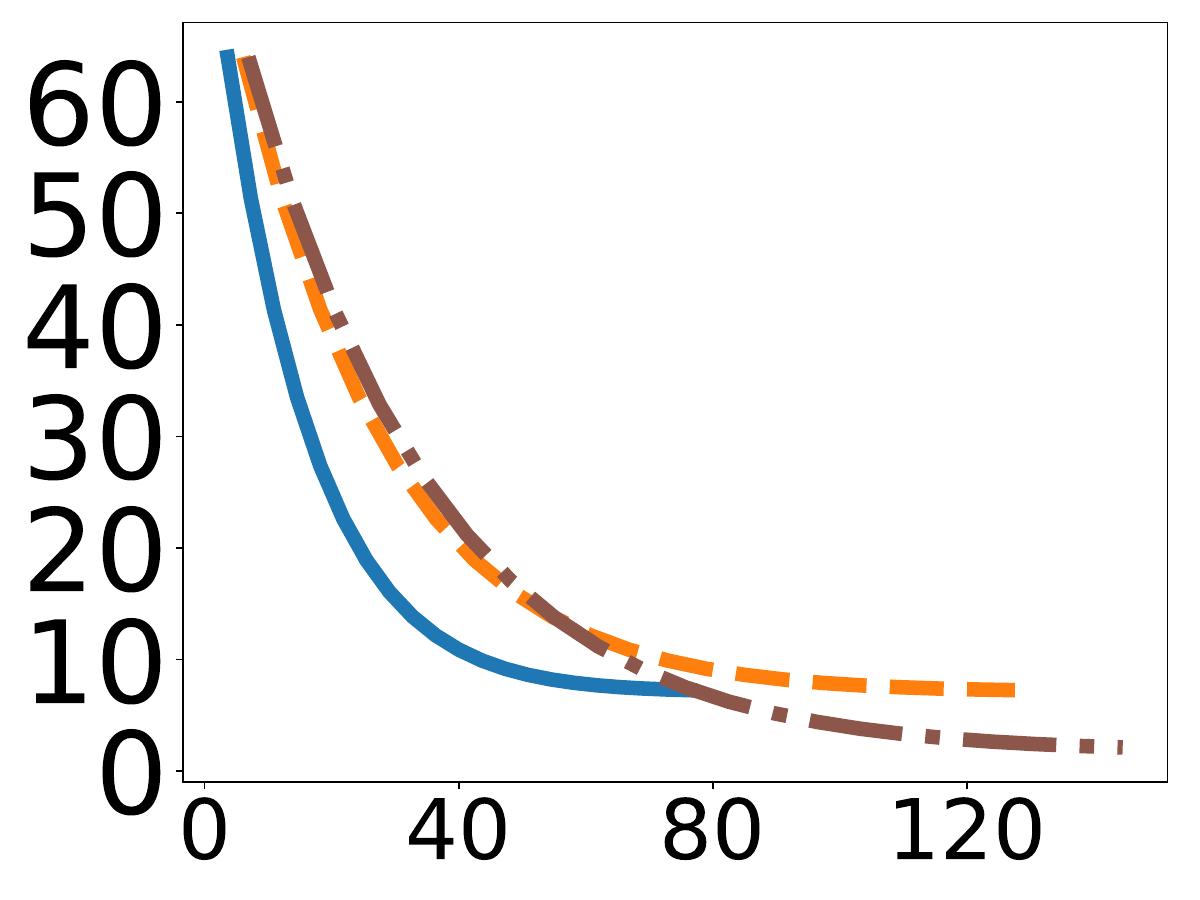}
            \vspace{-5mm}
            \caption{\textbf{Ta}}
        \end{subfigure}
        \begin{subfigure}{0.17\linewidth}
            \includegraphics[width=\textwidth]{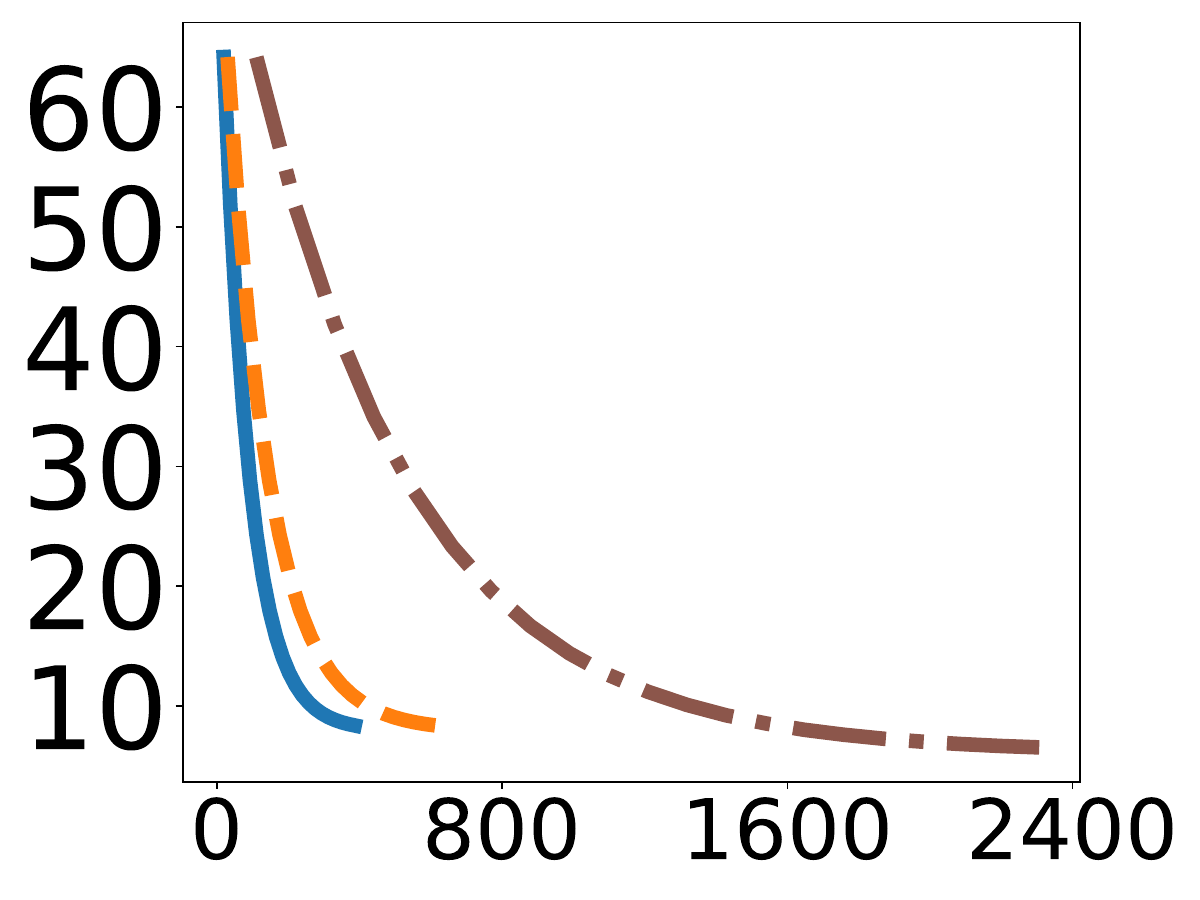}
            \vspace{-5mm}
            \caption{\textbf{Or}}
        \end{subfigure}
        \begin{subfigure}{0.17\linewidth}
            \includegraphics[width=\textwidth]{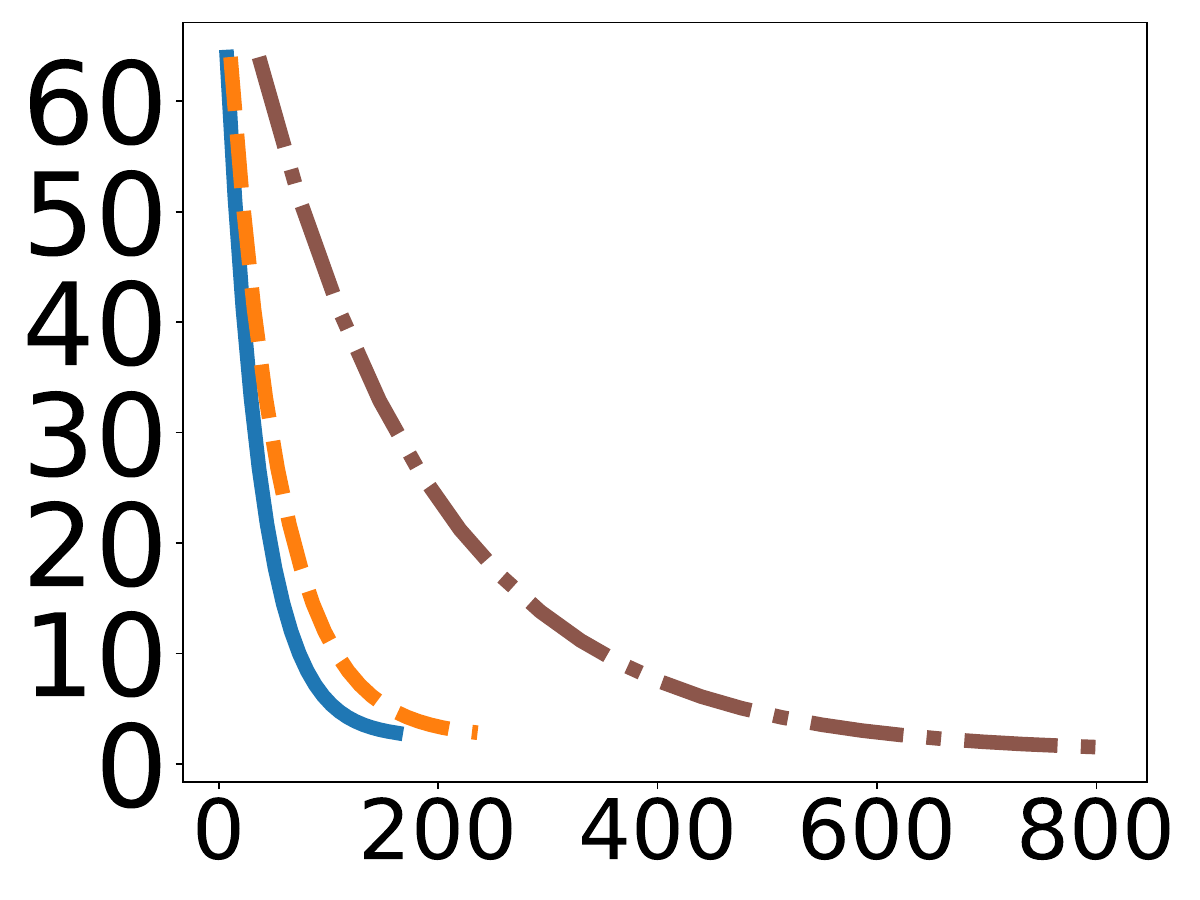}
            \vspace{-5mm}
            \caption{\textbf{LJ}}

        \end{subfigure}
    \vspace{-1mm}
    \caption{ Relative error (\%) vs. query time (second)}
    \label{fig:error_qtime}
    \vspace{-2mm}
\end{figure*}

\begin{figure*}
    \centering
    \begin{subfigure}{0.6\linewidth}
        \includegraphics[width=\textwidth]{figures/bar_legend.pdf}
    \end{subfigure}\\
        \begin{subfigure}{0.24\linewidth}
            \includegraphics[width=\textwidth]{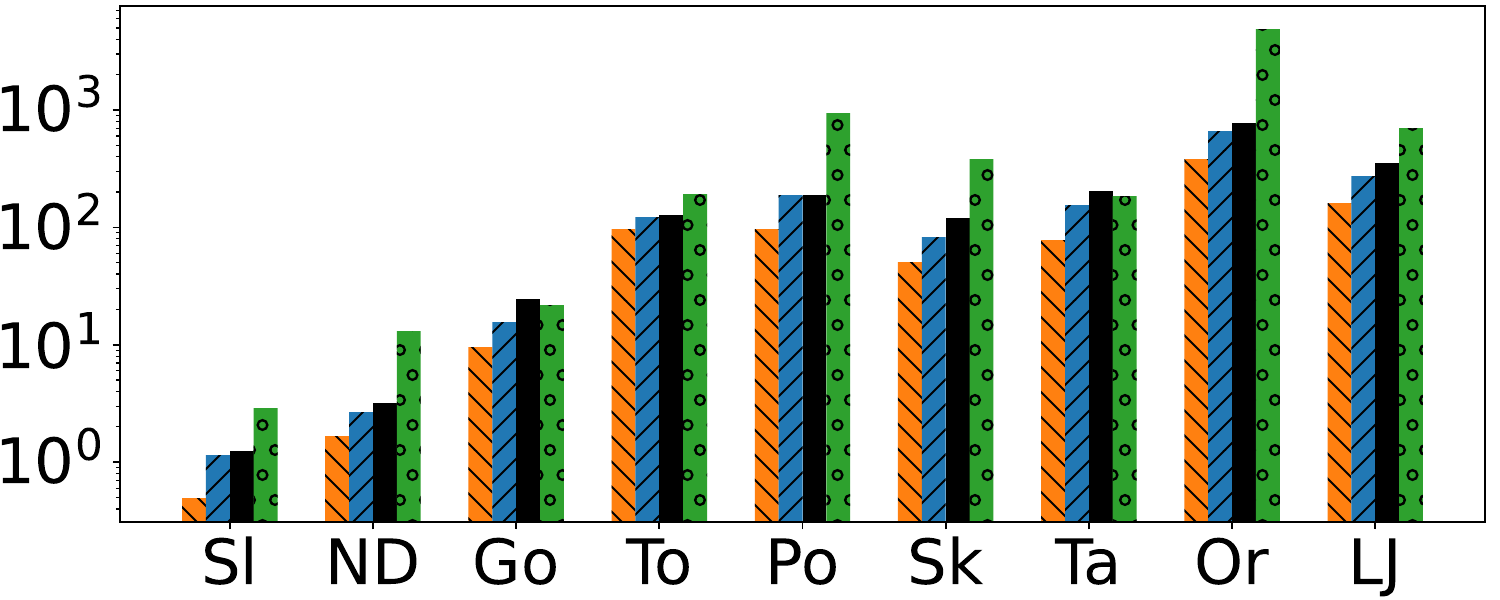}
            \vspace{-5mm}
            \caption{L-hop transition probability}
        \end{subfigure}
            \begin{subfigure}{0.24\linewidth}
                \includegraphics[width=\textwidth]{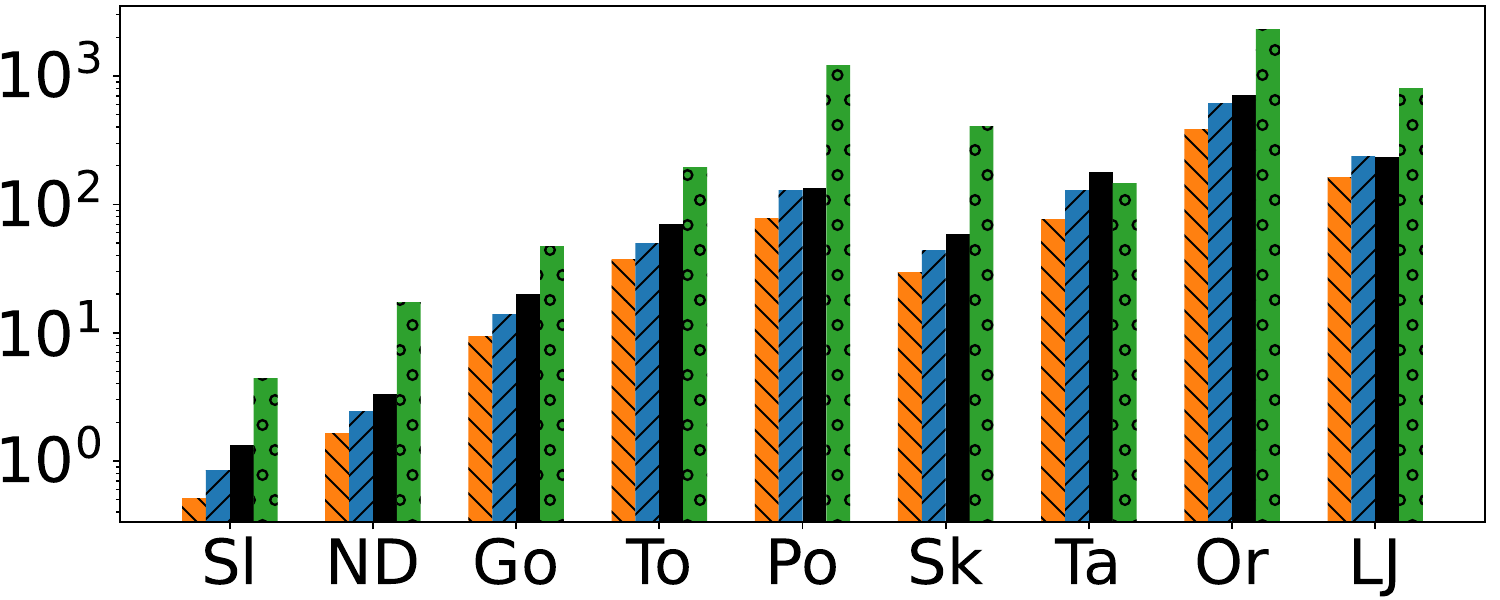}
                \vspace{-5mm}
                \caption{PageRank}
        \end{subfigure}
            \begin{subfigure}{0.24\linewidth}
                \includegraphics[width=\textwidth]{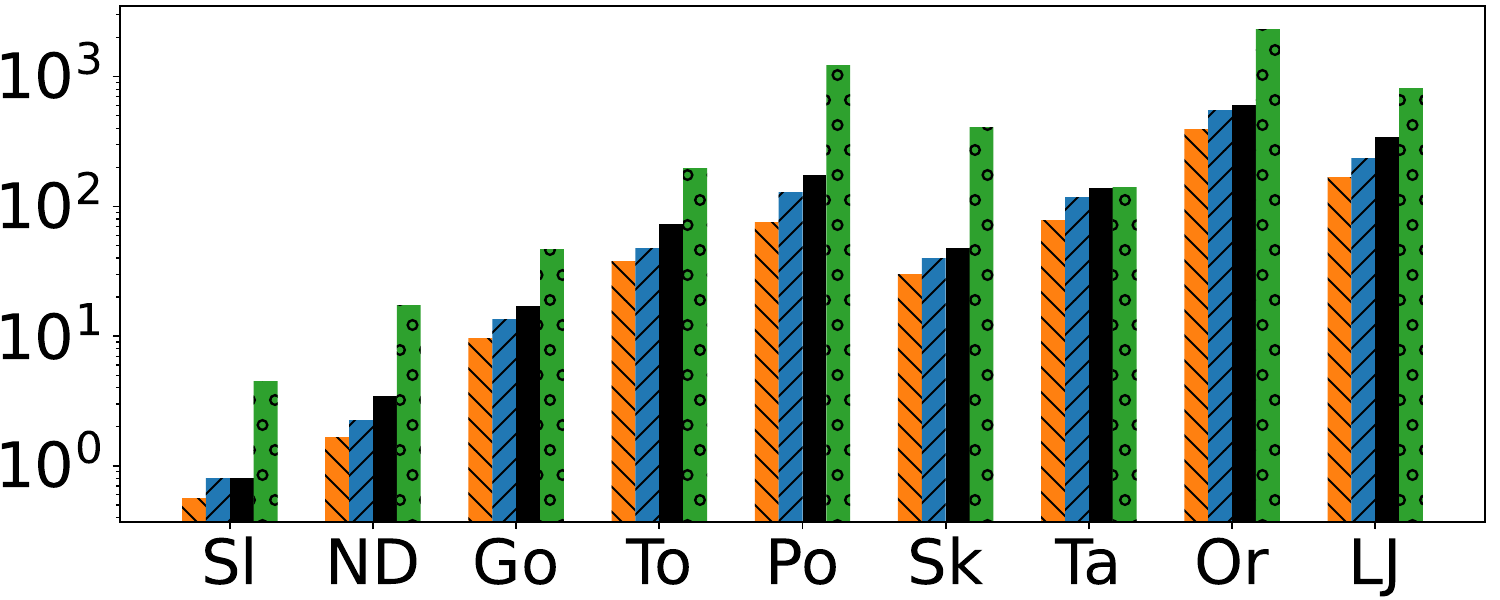}
                \vspace{-5mm}
                \caption{Personalized PageRank}
        \end{subfigure}
            \begin{subfigure}{0.24\linewidth}
                \includegraphics[width=\textwidth]{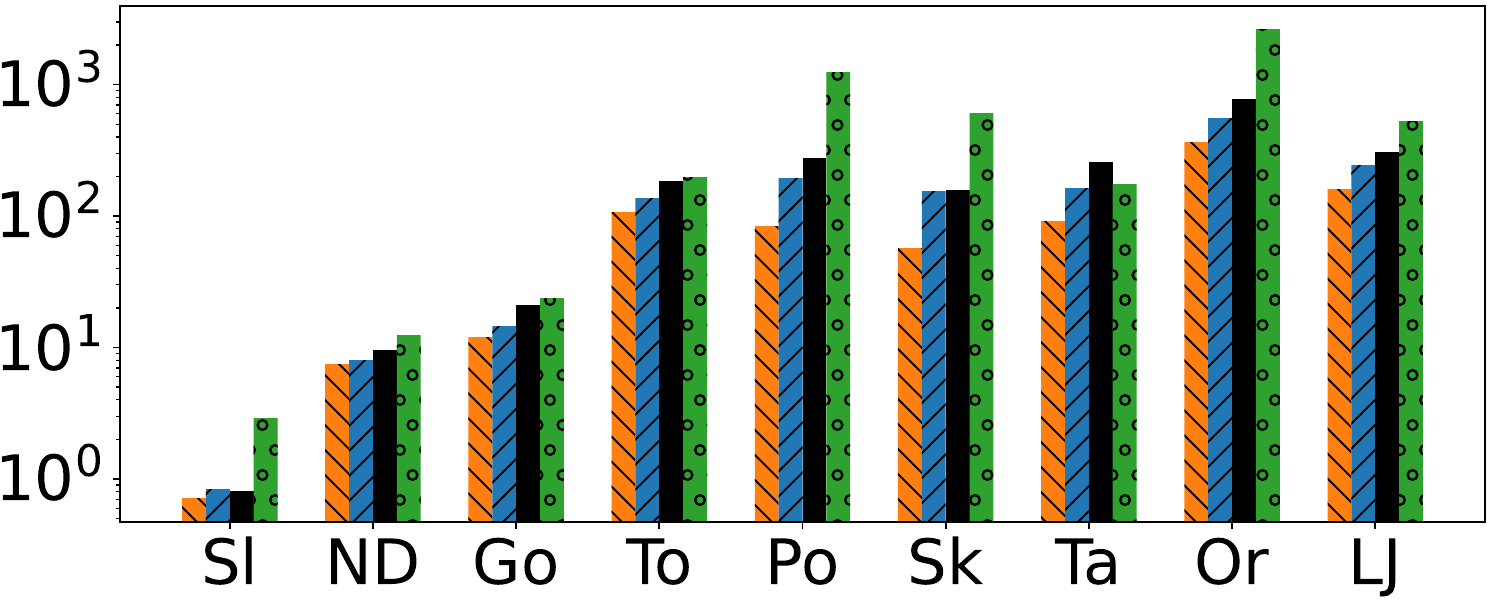}
                \vspace{-5mm}
                \caption{single-target PPR}
        \end{subfigure}
            \begin{subfigure}{0.24\linewidth}
                \includegraphics[width=\textwidth]{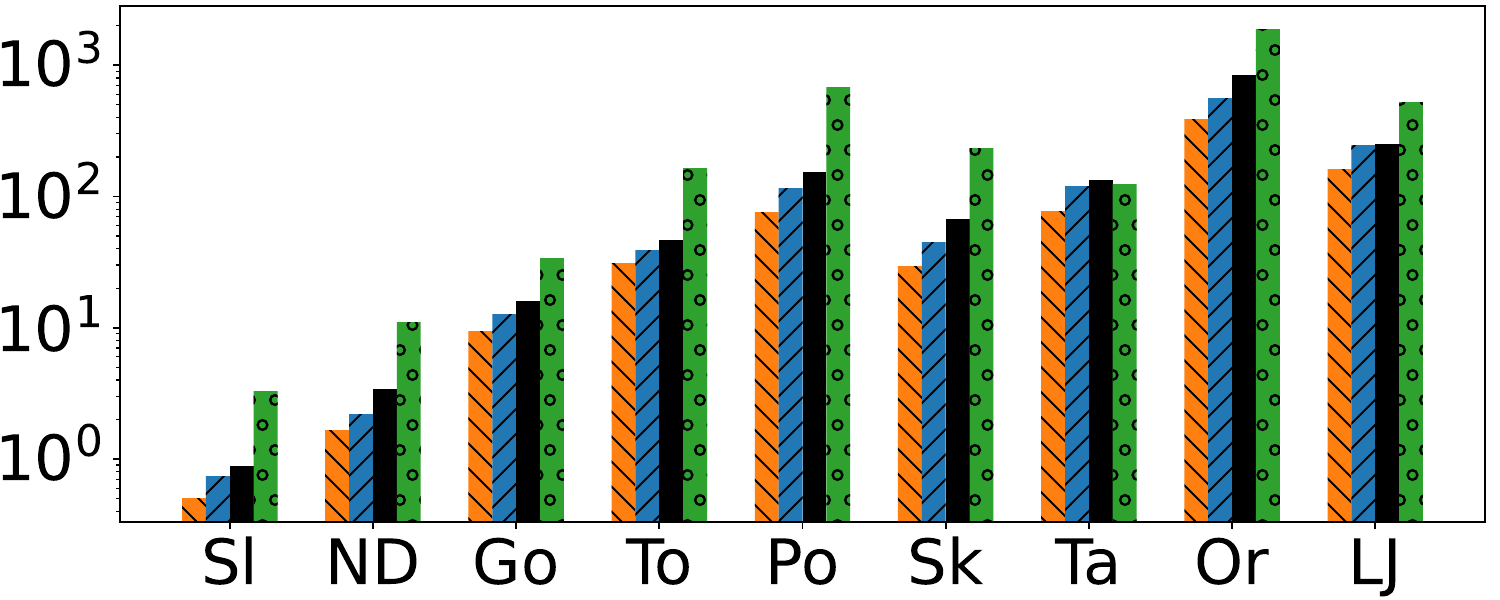}
                \vspace{-5mm}
                \caption{heat kernel PageRank}
        \end{subfigure}
        % \begin{subfigure}{0.24\linewidth}
        %     \includegraphics[width=\textwidth]{figures/qe/KATZ.pdf}
        %     \caption{Katz index}
        % \end{subfigure}\\
        \begin{subfigure}{0.24\linewidth}
            \includegraphics[width=\textwidth]{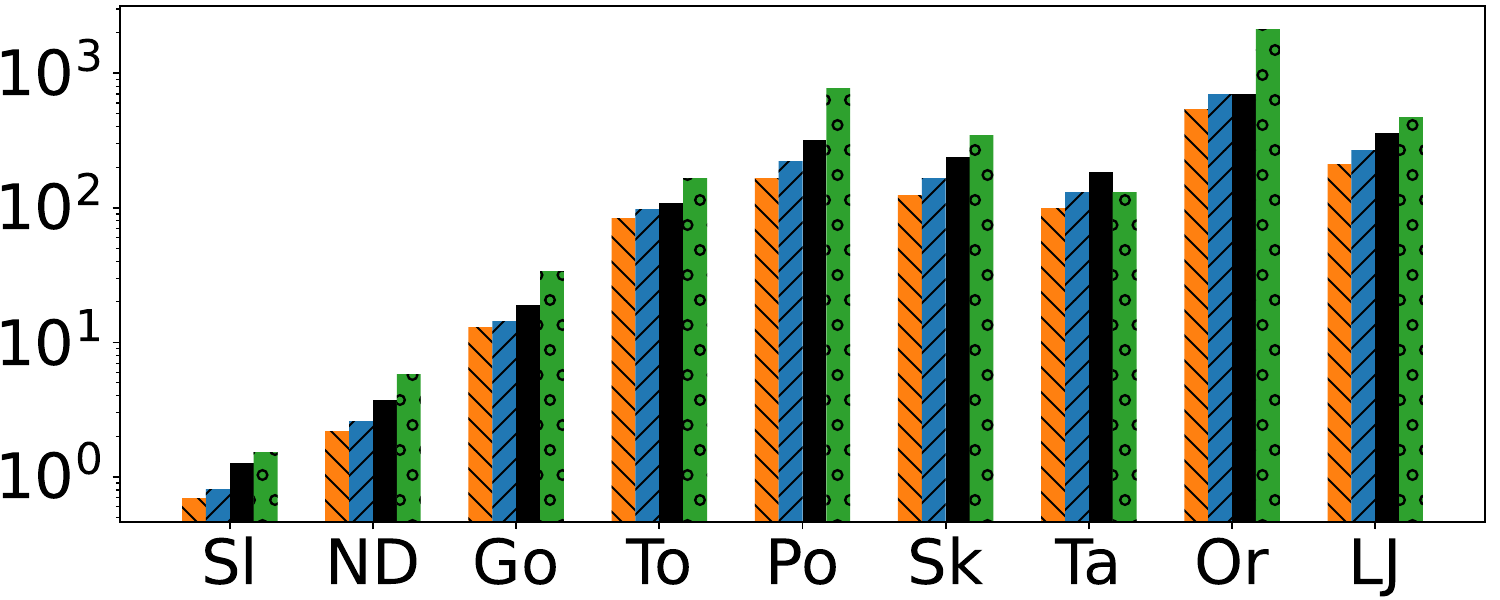}
            \vspace{-5mm}
            \caption{SGCN}
        \end{subfigure}
        \begin{subfigure}{0.24\linewidth}
            \includegraphics[width=\textwidth]{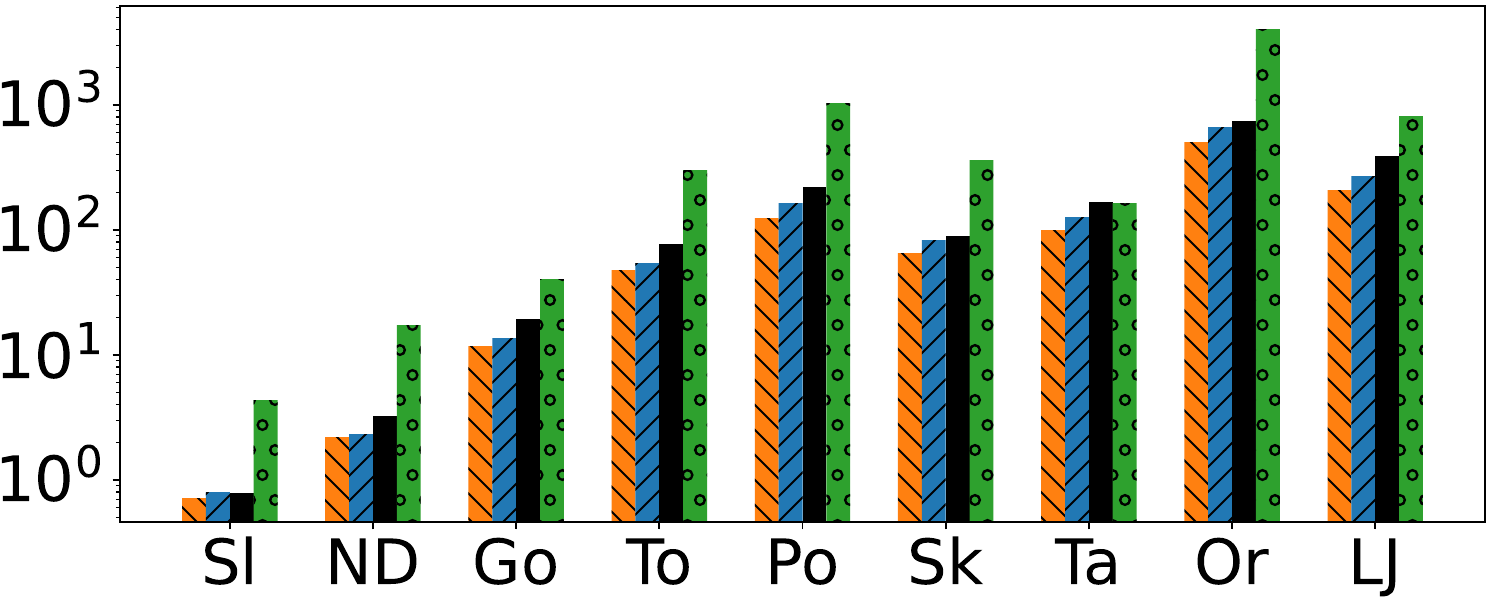}
            \vspace{-5mm}
            \caption{APPNP}
        \end{subfigure}
        \begin{subfigure}{0.24\linewidth}
            \includegraphics[width=\textwidth]{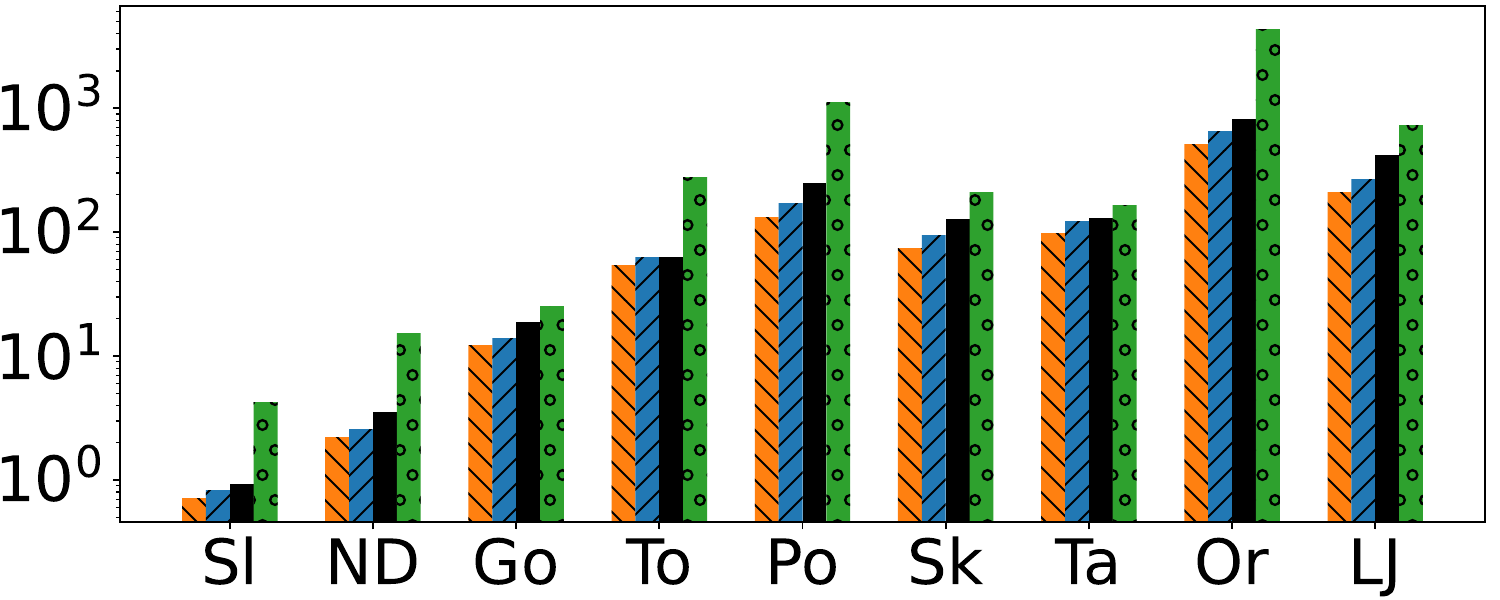}
            \vspace{-5mm}
            \caption{GDC}
        \end{subfigure}
    \vspace{-1mm}
    \caption{Average query time (second) on different applications}
    \label{fig:query_all}
    \vspace{-2mm}
\end{figure*}

\begin{figure}[h]
\centering
\begin{subfigure}{0.7\linewidth}
    \hspace{+2mm}\includegraphics[width=1\textwidth]{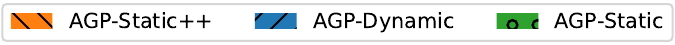}
\end{subfigure}\\
\begin{subfigure}{0.8\linewidth}
    \includegraphics[width=1\textwidth]{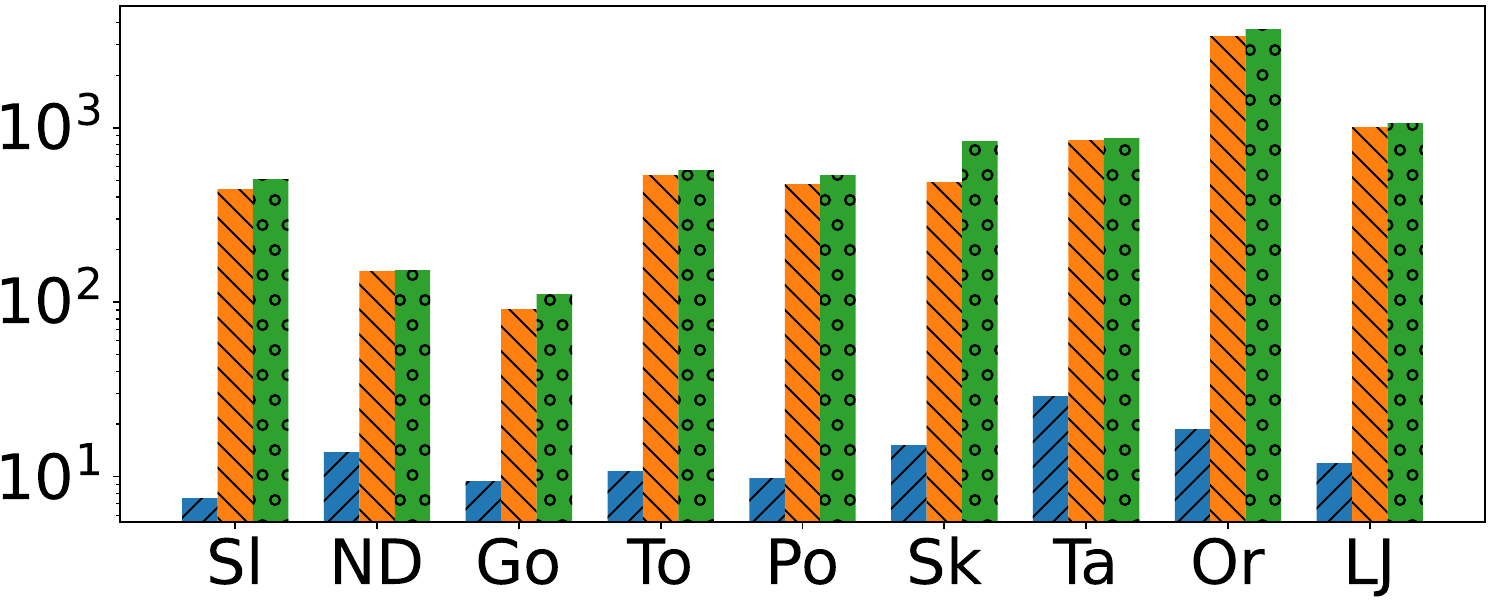}  
\end{subfigure}
\vspace{-4mm}
\caption{Average update time ($\times10^{-6}$ second)}
    \label{fig:update}
    \vspace{-4mm}
\end{figure}

\begin{table}[t]
\centering 
\caption{Comparison of Query Time, Average Update Time, and Average Number of Updates per Query Threshold}
\scalebox{0.65}{%
\begin{tabular}{l|c|c|c|c|c} 
\toprule 
 & \multicolumn{2}{c|}{\textbf{Query time}} & \multicolumn{2}{c|}{\textbf{Average update time ($\times 10^{-3}$)}}  & \multirow{2}{*}{\textbf{updates/query}}\\ \cline{2-5} 
& \emph{AGP-Static++} & \emph{AGP-Dynamic} & \emph{AGP-Static++} & \emph{AGP-Dynamic} &  \\ \hline
\texttt{Sl} ($\bar{d} = 12.13$) & 0.596 & 0.827 & $0.51$ & $0.01$ & 461 \\ \hline
\texttt{ND} ($\bar{d} = 6.69$) & 2.493 & 3.009 & $0.35$ & $0.03$ & 1617 \\ \hline
\texttt{Go} ($\bar{d} = 9.86$) & 10.633 & 13.854 & $1.02$ & $0.09$ & 3459 \\ \hline
\texttt{To} ($\bar{d} = 28.38$) & 56.823 & 70.018 & $30.93$ & $0.58$ & 435 \\ \hline
\texttt{Po} ($\bar{d} = 27.36$) & 100.400 & 157.439 & $25.40$ & $0.47$ & 2287 \\ \hline
\texttt{Sk} ($\bar{d} = 13.06$) & 53.745 & 82.031 & $19.79$ & $0.36$ & 1455 \\ \hline
\texttt{Ta} ($\bar{d} = 3.9$) & 85.588 & 130.488 & $8.58$ & $0.28$ & 5412 \\ \hline
\texttt{Or} ($\bar{d} = 76.22$) & 427.653 & 614.773 & $913.36$ & $4.67$ & 206 \\ \hline
\texttt{LJ} ($\bar{d} = 17.69$) & 176.794 & 248.747 & $96.62$ & $1.08$ & 753 \\ \bottomrule 
\end{tabular} 
}
\label{tab:qvu}
\end{table}

\begin{figure*}[th]
    \centering
    \begin{subfigure}{0.5\linewidth}
        \includegraphics[width=\textwidth]{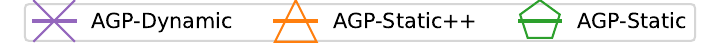}
    \end{subfigure}\\ \vspace{1mm}
        \begin{subfigure}{0.17\linewidth}
           \includegraphics[width=\textwidth]{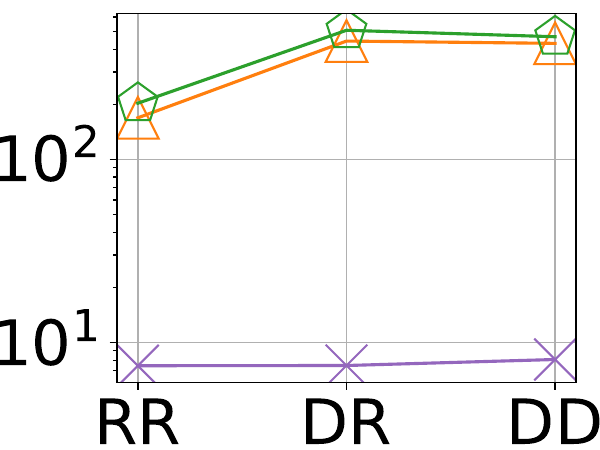}
            \caption{\textbf{Sl}}
        \end{subfigure}
            \begin{subfigure}{0.17\linewidth}
                \includegraphics[width=\textwidth]{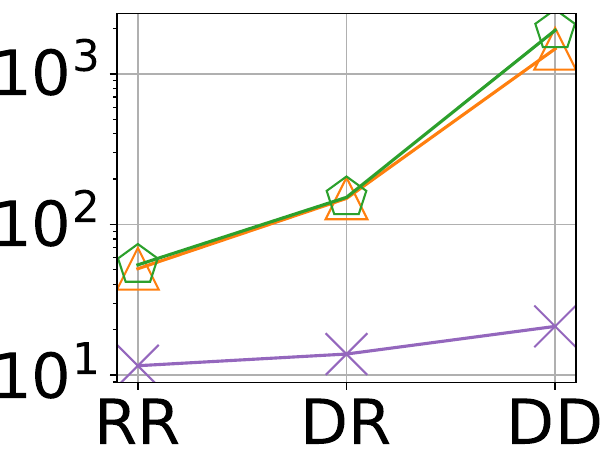}
                \caption{\textbf{ND}}
        \end{subfigure}
            \begin{subfigure}{0.17\linewidth}
                \hspace{-1mm}\includegraphics[width=\textwidth]{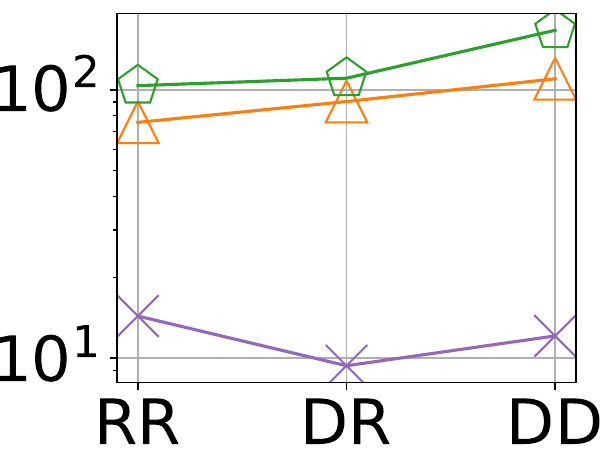}
                \caption{\textbf{Go}}
        \end{subfigure}
            \begin{subfigure}{0.17\linewidth}
                \hspace{-1mm}\includegraphics[width=\textwidth]{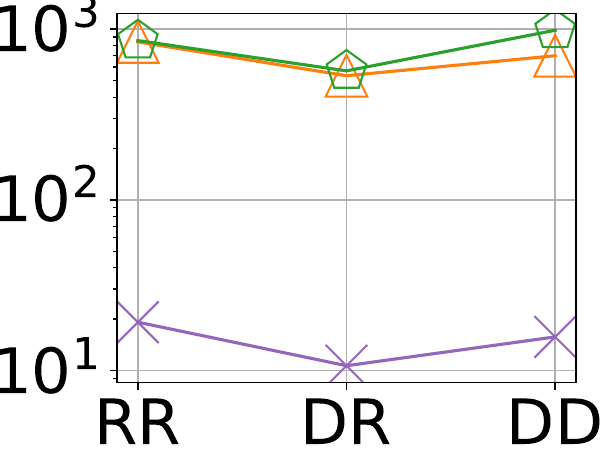}
                \caption{\textbf{To}}
        \end{subfigure}
            \begin{subfigure}{0.17\linewidth}
                \hspace{-2mm}\includegraphics[width=\textwidth]{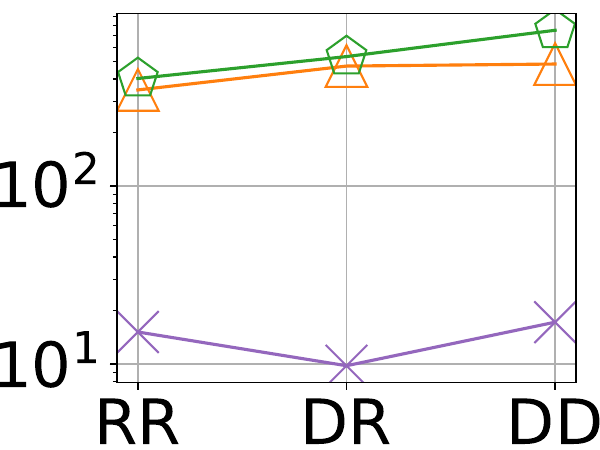}
                \caption{\textbf{Po}}
        \end{subfigure}
        \begin{subfigure}{0.17\linewidth}
            \hspace{-3mm}\includegraphics[width=\textwidth]{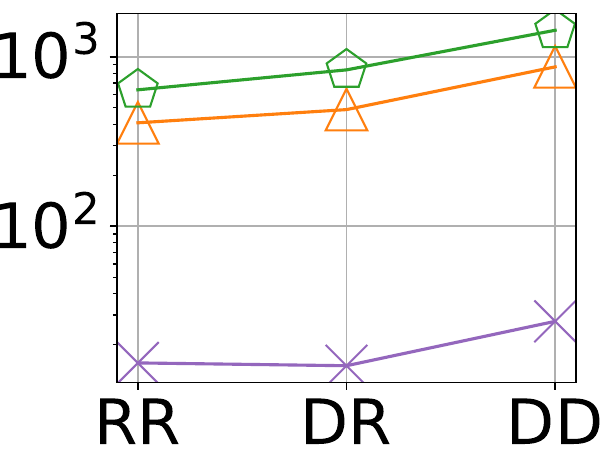}
            \caption{\textbf{Sk}}
        \end{subfigure}
        \begin{subfigure}{0.17\linewidth}
            \hspace{-3mm}\includegraphics[width=\textwidth]{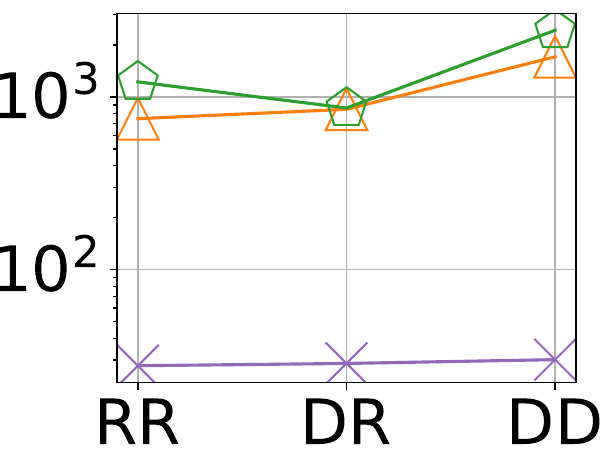}
            \caption{\textbf{Ta}}
        \end{subfigure}
        \begin{subfigure}{0.17\linewidth}
            \hspace{-3mm}\includegraphics[width=\textwidth]{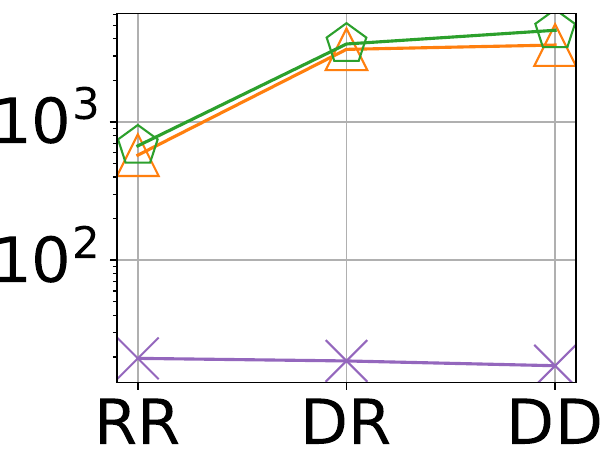}
            \caption{\textbf{Or}}
        \end{subfigure}
        \begin{subfigure}{0.17\linewidth}
            \hspace{-3mm}\includegraphics[width=\textwidth]{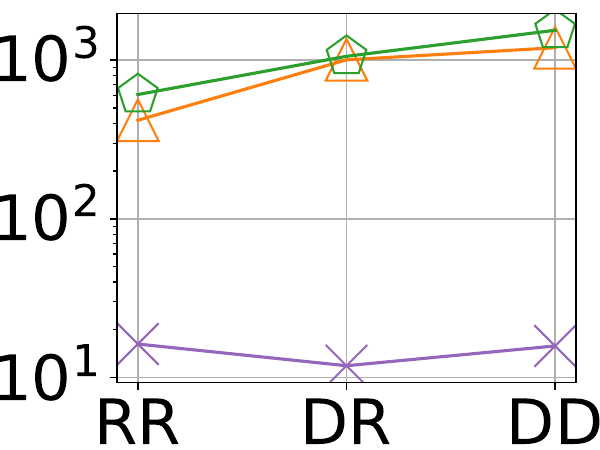}
            \caption{\textbf{LJ}}

        \end{subfigure}
    \vspace{-3mm}
    \caption{Average update time ($\times10^{-6}$ sec) vs. update pattern}
    \label{fig:ud_varying_full}
% \vspace{-5mm}
\end{figure*}

\begin{figure*}[th]
    % \centering
    % \begin{subfigure}{1\linewidth}
    %     \includegraphics[width=\textwidth]{figures/lines_legend.pdf}
    % \end{subfigure}
        \begin{subfigure}{0.17\linewidth}
            \includegraphics[width=\textwidth]{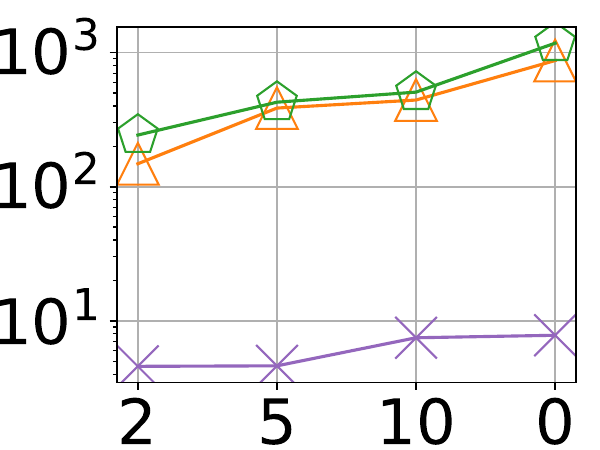}
            \caption{\textbf{Sl}}
        \end{subfigure}
            \begin{subfigure}{0.17\linewidth}
                \includegraphics[width=\textwidth]{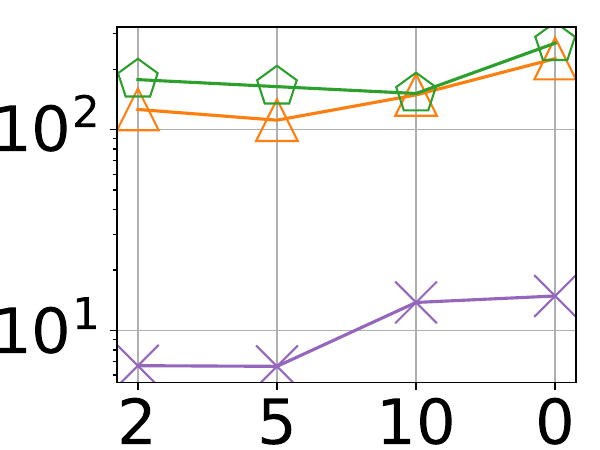}
                \caption{\textbf{ND}}
        \end{subfigure}
            \begin{subfigure}{0.17\linewidth}
                \hspace{-1mm}\includegraphics[width=\textwidth]{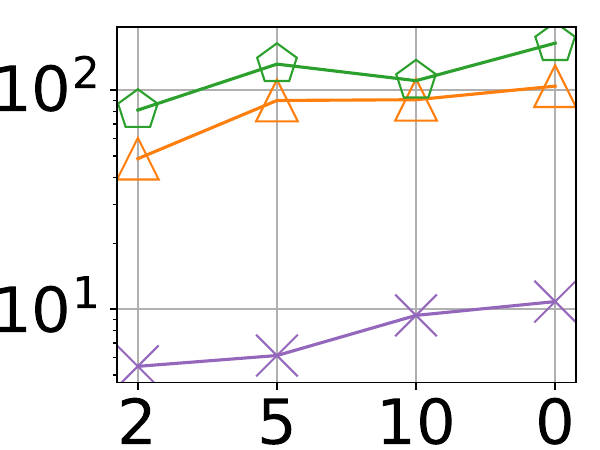}
                \caption{\textbf{Go}}
        \end{subfigure}
            \begin{subfigure}{0.17\linewidth}
                \hspace{-1mm}\includegraphics[width=\textwidth]{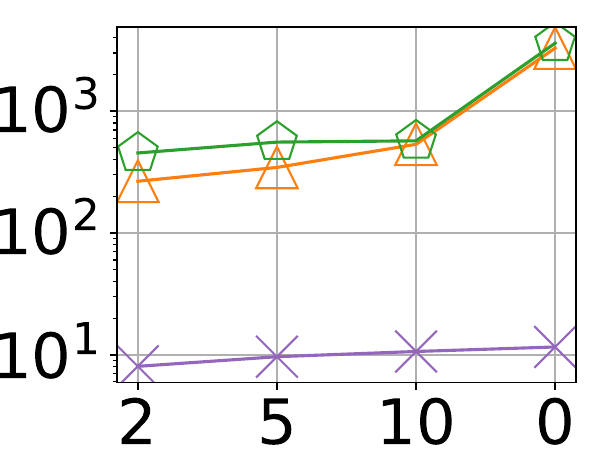}
                \caption{\textbf{To}}
        \end{subfigure}
            \begin{subfigure}{0.17\linewidth}
                \hspace{-2mm}\includegraphics[width=\textwidth]{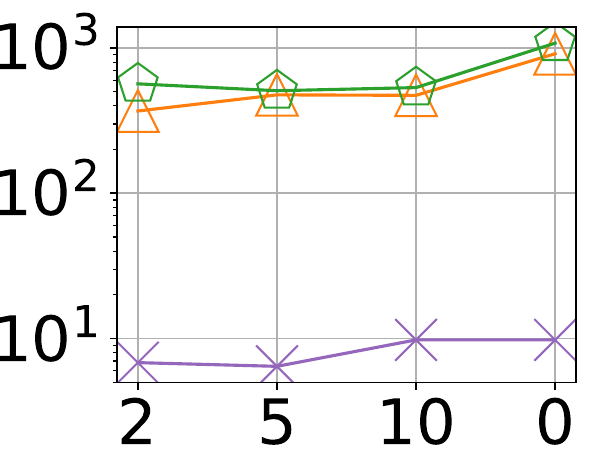}
                \caption{\textbf{Po}}
        \end{subfigure}
        \begin{subfigure}{0.17\linewidth}
            \hspace{-2mm}\includegraphics[width=\textwidth]{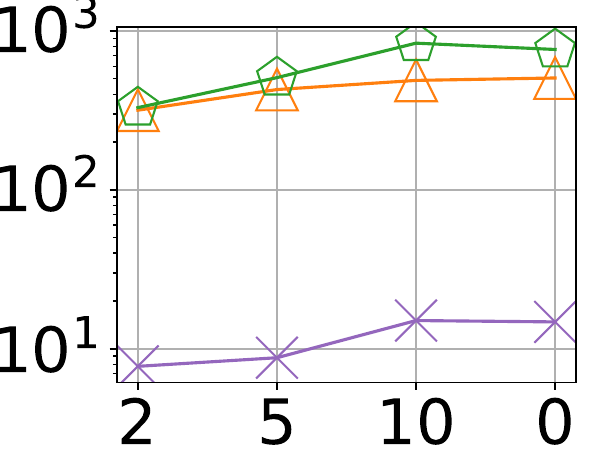}
            \caption{\textbf{Sk}}
        \end{subfigure}
        \begin{subfigure}{0.17\linewidth}
            \hspace{-2mm}\includegraphics[width=\textwidth]{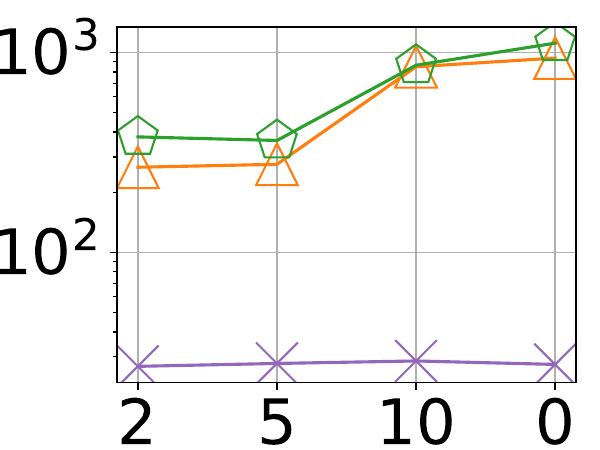}
            \caption{\textbf{Ta}}
        \end{subfigure}
        \begin{subfigure}{0.17\linewidth}
            \hspace{-2mm}\includegraphics[width=\textwidth]{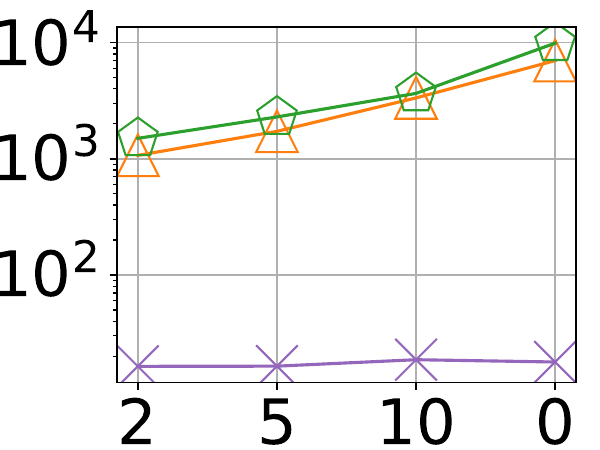}
            \caption{\textbf{Or}}
        \end{subfigure}
        \begin{subfigure}{0.17\linewidth}
            \hspace{-2mm}\includegraphics[width=\textwidth]{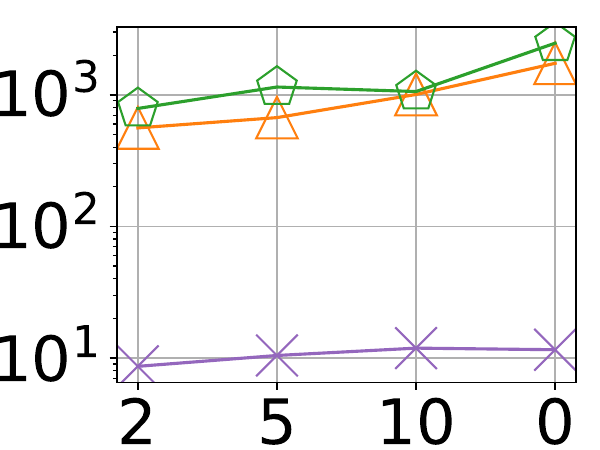}
            \caption{\textbf{LJ}}

        \end{subfigure}
    \vspace{-3mm}
    \caption{Average update time ($\times10^{-6}$ second) vs. $\eta$}
    \label{fig:eta_varying_full}
\end{figure*}

\vspace{1mm}
\noindent\textbf{Parameter Setting.} We evaluate all applications in Table~\ref{tab:example} using their respective parameter settings. Following~\cite{wang2021approximate}, we set $\alpha = 0.2$ for random walk, $t = 4$ for Poisson distribution. 
%and $\beta = \frac{0.85}{\lambda}$ for  Katz index, where $\lambda$ denotes the largest eigenvalue of the input graph's adjacency matrix. 
For applications where $\mathbf{x}$ is a one-hot vector, we construct it by selecting $j$ uniformly at random from $[1,n]$ and setting $\mathbf{x}_j=1$. For GNN applications, we generate $n$ values following a uniform distribution with their sum equal to 1 and assign them to $\mathbf{x}$. We set $\delta=\frac{1}{n}$, $c = 0.1$ and $L=O(\log n)$ for all queries.

\vspace{1mm}
\noindent\textbf{Results.} For each dataset, we run five queries using five randomly generated instances of $\mathbf{x}$, except for PageRank, where $\mathbf{x}=[1/n, \ldots, 1/n]^T$ remains fixed. The average runtime of all applications on each dataset is reported in Figure~\ref{fig:query}. As discussed in Section~\ref{sec:query_ana}, our tighter bounds lead to improvements in query time across all datasets. Notably, on \texttt{Po}, our \algosta\ speeds up \algoagp\ by 9.74$\times$. 
For \algodyn, it picks a larger $p^*$ in subset sampling, and hence  
generates a larger candidate sample set for rejection sampling. 
Nonetheless, it still outperforms \algoagp\ by up to 6.21$\times$ (also on \texttt{Po}). 
Meanwhile, despite its theoretical superiority, \algodpss~performs worse than \algosta; this is mainly due to the larger hidden constants introduced by the more complex structure of DPSS. The results for each individual application demonstrate similar results, which are shown in Figure~\ref{fig:query_all}. Across all applications, our method \algosta\ achieves the best performance, thanks to its tighter bound and the new subset sampling algorithm. \algodyn\ exhibits slightly higher query time due to its use of a larger $p^*$ in each bucket. \algodpss\ performs slightly worse, primarily because of the large hidden constant in its query complexity.

\noindent \textbf{Query time vs. error.} 
Recall that the algorithms run $L$ steps recursively to approximate $\boldsymbol{\pi}$.
Next, we measure the relative errors and the running time consumed over these $L$ iterations. 
The ground truth is gained by running the deterministic algorithm for $L'$ steps until $\boldsymbol{\hat{\pi}}_{L'}(u) - \boldsymbol{\hat{\pi}}_{L'-1}(u) < 10^{-9}$ for all $u \in V$.
The average relative errors (ARE's) for all vertices vs. the running time on PageRank are reported in Figure~\ref{fig:error_qtime}. 
We observe that all algorithms have achieved the target relative error bound, which is 10\% as $c = 0.1$ within $L$ steps. 
Specifically, on \texttt{ND}, \texttt{Go}, \texttt{Po}, and \texttt{LJ}, the ARE's are below 3\%. 
Since $\varepsilon$ is set to be different values, recall that $\varepsilon = O(\frac{\delta}{L})$ in
our AGP-Static++ and AGP-Dynamic, while $\varepsilon = O(\frac{\delta}{L^2})$ in AGP-Static.
We observe that our methods incur slightly higher errors than AGP-Static after $L$ steps. 
This is because with a smaller $\varepsilon$, AGP-Static performs more deterministic propagation (with higher running time) and hence potentially introduces less error. 
Nonetheless, when considering the query time, our algorithms reach low relative errors significantly faster, except on  \texttt{Ta}.
In particular, as shown in Figure~\ref{fig:error_qtime} (i), on the largest graph, \texttt{LJ}, AGP-Static++ achieved an ARE of 3.085\% in 187 seconds, whereas AGP-Static required 565 seconds to achieve 3.537\% ARE.
However, on \texttt{Ta}, 
its small average degree 
%it has a small average degree 
($\bar{d} = 3.90$ as shown in Table~\ref{tab:dataset})
limits the efficiency gain from the randomized propagations for queries, 
%a randomized method offers limited advantage for the query, 
resulting in comparable performance across all three algorithms.

\subsection{Updating Efficiency}
\label{sec:update_effi}

\noindent\textbf{Competitors.} We compare our \algodyn\ and \algosta\  with \algoagp. \algodpss\ uses the same update method as \algosta\ but with a larger hidden constant, and is omitted here.

\vspace{1mm}
\noindent\textbf{Update Generation.} 
We randomly generate a sequence of edge insertions and deletions for each dataset.
We examine the impact of the insertion-to-deletion ratio $\eta$ by assigning the probabilities of an insertion and a deletion as $\frac{\eta}{1+\eta}$ and $\frac{1}{1+\eta}$, respectively.
To delete an edge, we uniformly sample an existing edge and delete it.
To insert an edge, we adopt three strategies.
\textbf{Random-Random (RR)}: A non-existent edge is inserted by randomly selecting two unconnected vertices.
\textbf{Degree-Random (DR)}: A vertex $u$ is selected with probability $\frac{d_u}{2m}$, where $d_u$ is the degree of $u$ and $m$ is the current number of edges. Then, a vertex $v$ not yet connected to $u$ is chosen uniformly at random.
\textbf{Degree-Degree (DD)}: Vertex $u$ is selected as in DR. Then, vertex $v$ is selected from the vertices not yet connected to $u$, with probability proportional to $\frac{d_v}{2m}$.

By default, we set the insertion-to-deletion ratio to 10:1 ($\eta = 10$) and adopt the \textbf{DR} strategy, reflecting the common real-world scenario where insertions are more frequent than deletions, and users tend to follow high-degree nodes (celebrity) in social networks. We create $2m$ updates for each dataset and each $\eta$ and update pattern, where $m$ is the number of edges of the input graph.

\vspace{1mm}
\noindent\textbf{Results.}  As shown in Figure~\ref{fig:update}, our \algodyn\ outperforms \algosta\  and \algoagp\ across all datasets in update efficiency. 
Compared to \algosta\ and \algoagp, the speedup ranges from 8× (on \texttt{Go}) to 177× (on \texttt{Or}). 
This trend roughly matches the average degree of each dataset as shown in Table~\ref{tab:dataset}, which is consistent with the theoretical improvement from the $O(d_u)$ update time of existing solutions to our amortized $O(1)$ update time.
The results for \algosta\ and \algoagp\ indicate that updates are more costly on graphs with a larger average degree, which is expected since a change in $d_u$ can incur an $O(d_u)$ cost. 
Additionally, \algoagp\ maintains a sorted linked list which pays an extra $O(\log n)$ update time, hence is slower. 
%than our \algosta.

\vspace{1mm}

\noindent\underline{Impact of Update Pattern.} As further shown in Figure~\ref{fig:eta_varying_full}, overall, the update times for \algosta\ and \algoagp\ increase as updates become more skewed (from RR to DR to DD). This is expected, as updates on higher-degree vertices incur a larger cost (i.e., $O(d_u)$). However, this trend does not always hold for \algodyn. When the degree of a vertex increases, the threshold $2 \cdot \tilde{d}_u$ (or $\frac{1}{2} \cdot \tilde{d}_u$) also increases, making it harder to trigger the $O(d_u)$ cost in \algodyn. 
Notably, \algodyn\ outperforms the baseline methods on all datasets for all update patterns by up to 207$\times$ (DD on \texttt{Or}).

\vspace{1mm}
\noindent\underline{Impact of Insertion-to-Deletion Ratio.} Similarly, as shown in Figure~\ref{fig:eta_varying_full}, when the ratio of insertion increases, the update times for \algosta\ and \algoagp\ increase ($\eta = 0$ means no deletions), while our \algodyn\ has a more stable trend and outperforms \algoagp\ by up to 553$\times$ on \texttt{Or} when $\eta = 0 $. Additionally, on larger datasets, our \algoagp\ performs more stably. This is because there are more updates ($2m$) on larger datasets, which justify our amortized cost.

\vspace{1mm}
{\noindent\underline{Impact of Update-to-Query Frequency Ratio.}
%From the previous comparison, 
From the above experiments,
we can see that \algodyn~is less efficient than \algosta in query time while it outperforms \algosta~in update cost.
To further study the trade off,
we compare the update and query times of \algodyn~and \algosta~in Table~\ref{tab:qvu}. 
The results show that \algodyn~is at most 36\% slower than \algosta~(i.e., $1.36 \times$ the time). 
We further examine the update-to-query frequency ratio under which \algodyn~outperforms \algosta. 
This ratio depends on the average degree $\bar{d}$ of the dataset. 
For instance, in \texttt{Ta} with $\bar{d}=3.9$, a relatively high update-to-query ratio is required 
for \algodyn~to gain an advantage. 
In contrast, in \texttt{Or} with $\bar{d}=76.22$, only 206 updates per query are sufficient 
for \algodyn~to become more efficient. 
On large graphs such as \texttt{LJ}, which contains over 1 billion edges, 
\algodyn~outperforms \algosta~with as few as 754 updates per query. 
Since updates are usually much more frequent than queries in practice, \algodyn~remains highly competitive.
}

\begin{figure}
\begin{subfigure}{0.333\linewidth}
    \includegraphics[width=1\textwidth]{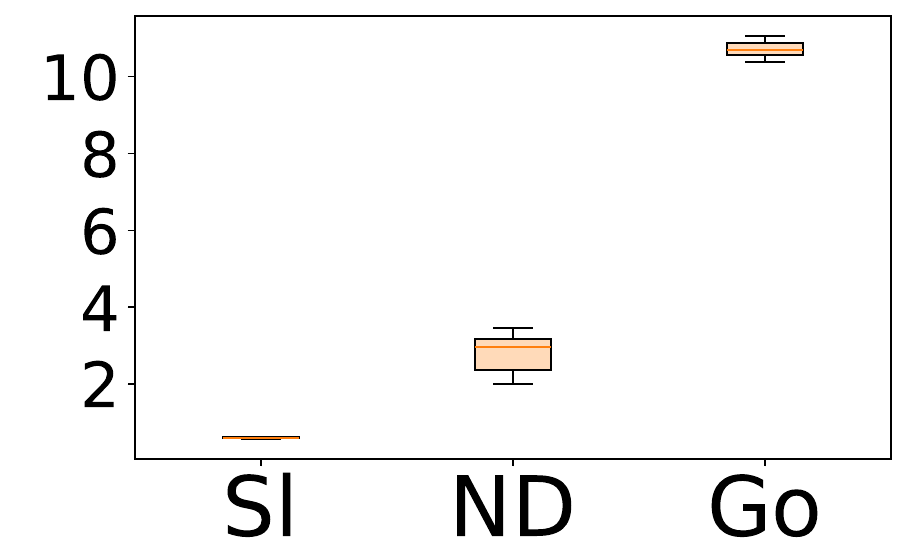}
\end{subfigure}%
\begin{subfigure}{0.333\linewidth}
    \includegraphics[width=1\textwidth]{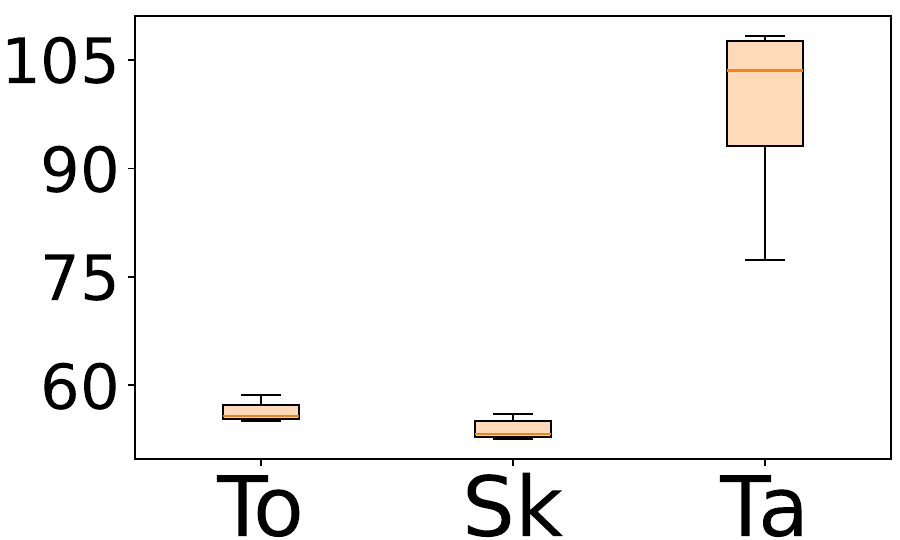}
\end{subfigure}%
\begin{subfigure}{0.333\linewidth}
    \includegraphics[width=1\textwidth]{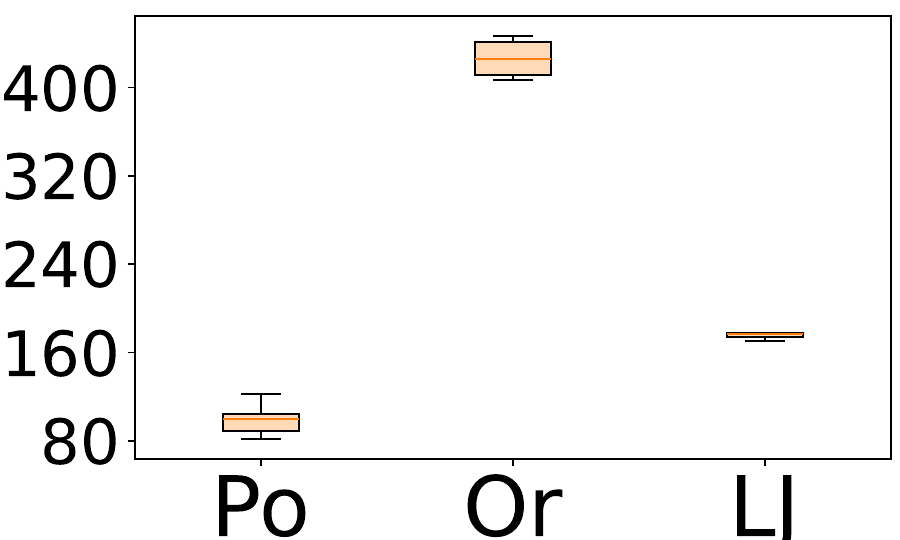}
\end{subfigure}
\vspace{-4mm}
    \caption{Query time (sec) on random $a$ and $b$ of \algosta}
    \label{fig:abd}
\end{figure}
\begin{figure}
\begin{subfigure}{0.333\linewidth}
    \includegraphics[width=1\textwidth]{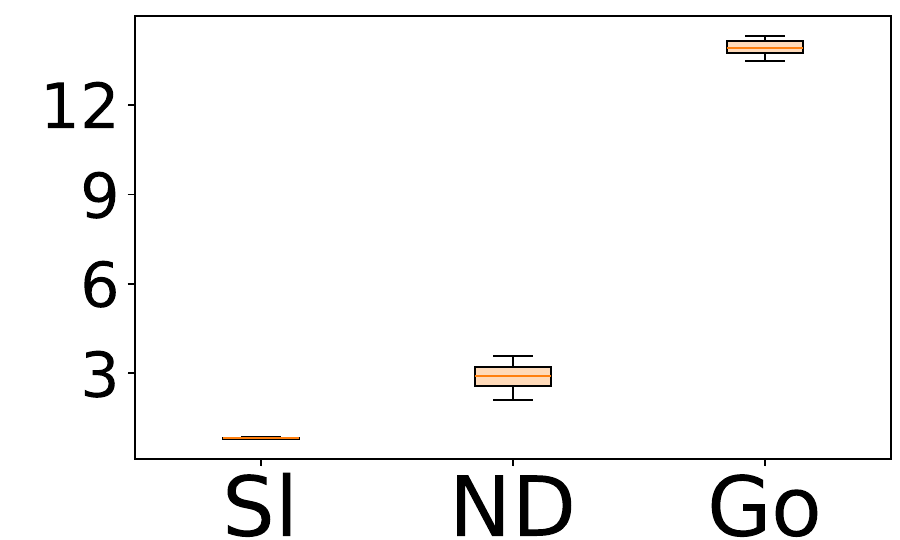}
\end{subfigure}%
\begin{subfigure}{0.333\linewidth}
    \includegraphics[width=1\textwidth]{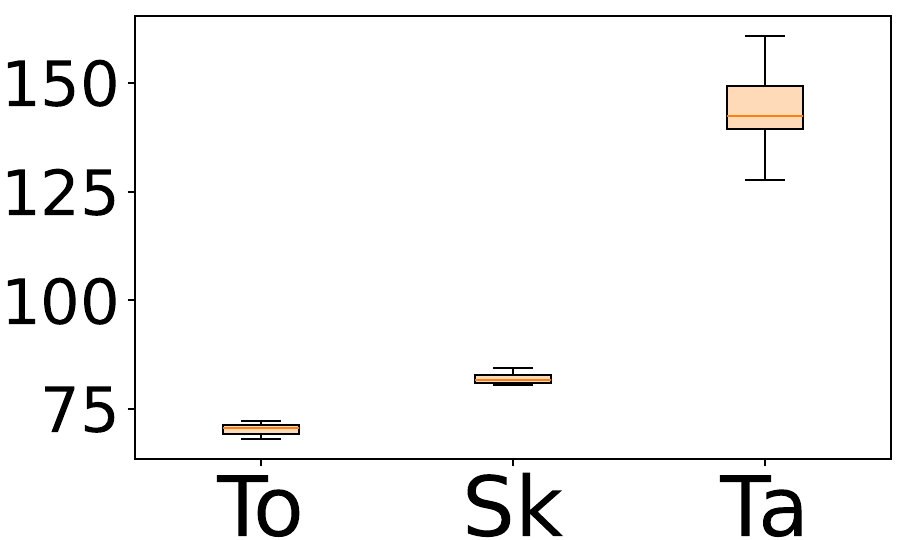}
\end{subfigure}%
\begin{subfigure}{0.333\linewidth}
    \includegraphics[width=1\textwidth]{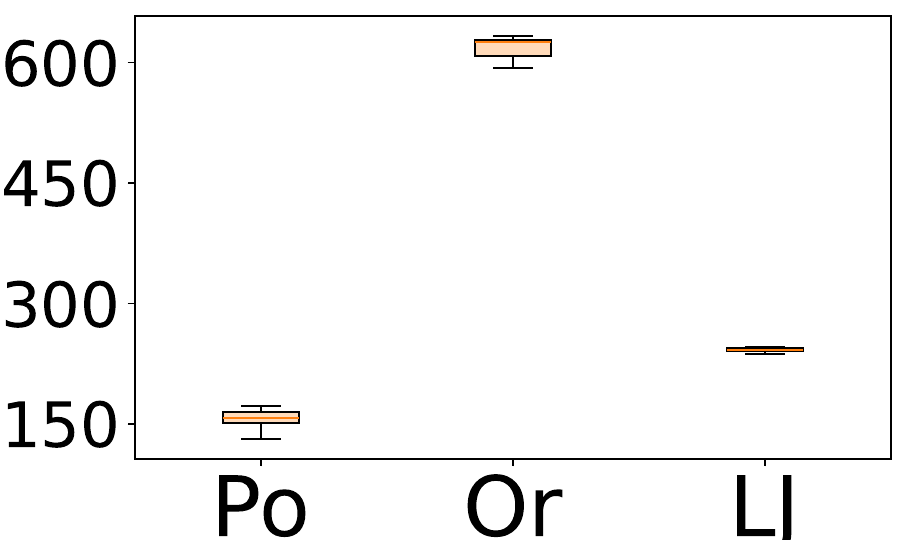}
\end{subfigure}
\vspace{-4mm}
    \caption{Query time (sec) on random $a$ and $b$ of \algodyn}
    \label{fig:aba+}
\end{figure}
\begin{figure}
\begin{subfigure}{0.333\linewidth}
    \includegraphics[width=1\textwidth]{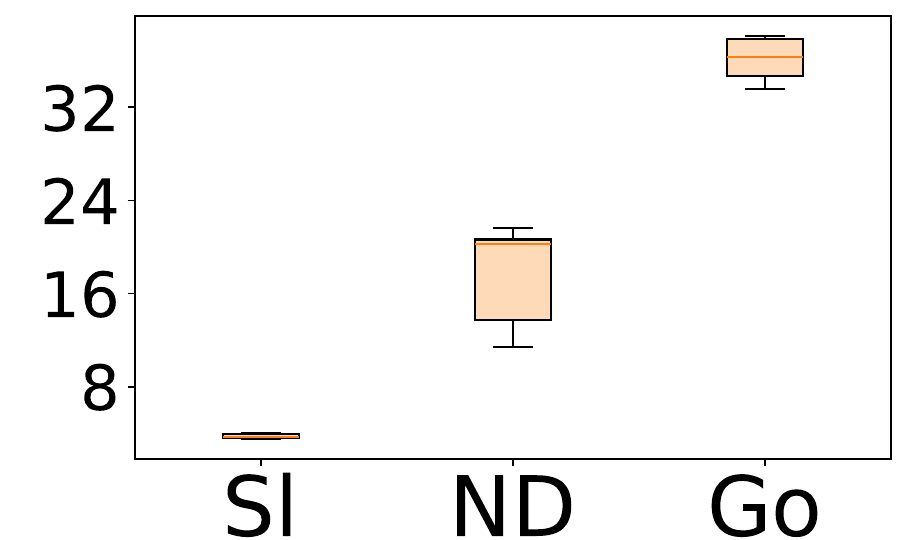}
\end{subfigure}%
\begin{subfigure}{0.333\linewidth}
    \includegraphics[width=1\textwidth]{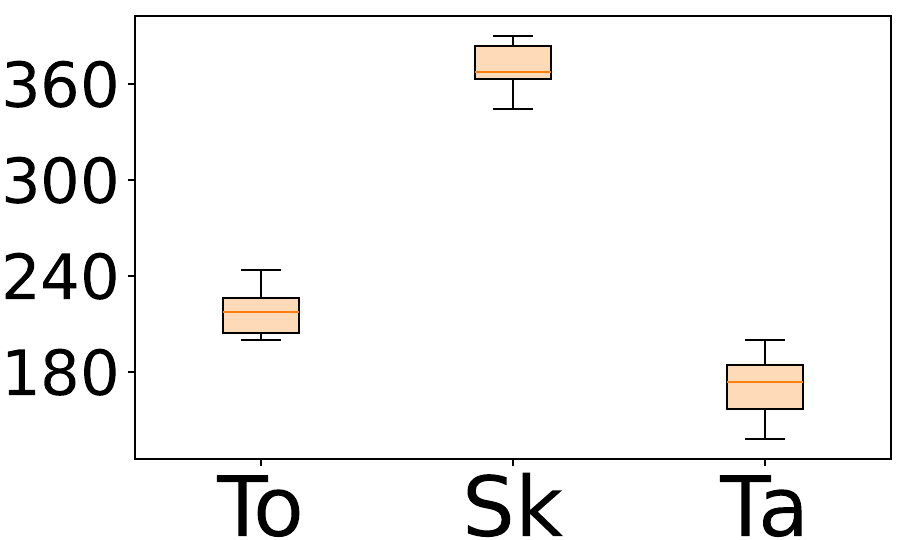}
\end{subfigure}%
\begin{subfigure}{0.333\linewidth}
    \includegraphics[width=1\textwidth]{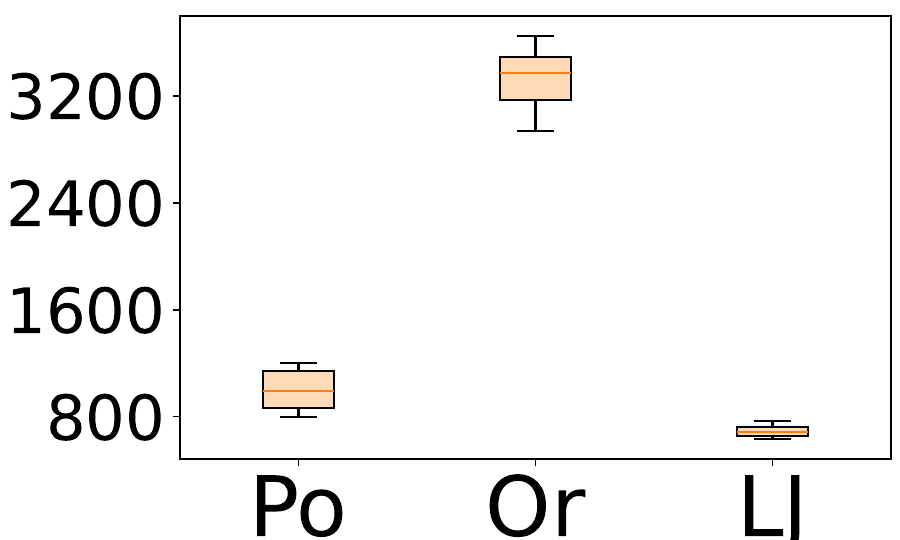}
\end{subfigure}
\vspace{-4mm}
    \caption{Query time (sec) on random $a$ and $b$ of \algoagp}
    \label{fig:aba}
\end{figure}

\vspace{4mm}
\subsection{Arbitrary $a$ and $b$} 
%alpha, x 
%We evaluate our algorithms using arbitrary values of $a$ and $b$.
For each dataset, we execute 100 queries where $a$ and $b$ are drawn from the uniform distribution over $[0,1]$, subject to the constraint $a + b \geq 1$. We set $w_i = \alpha(1-\alpha)^i$ with $\alpha = 0.2$, and $\mathbf{x}$ to a one hot vector generated uniformly at random.
The results for \algodyn\ are shown in Figure~\ref{fig:abd}. As expected, the query times vary across different $(a, b)$ combinations, since according to Theorem~\ref{theorem:query}, $a$ and $b$ impact the output size and, consequently, the query complexity.
Nevertheless, as shown in Figure~\ref{fig:abd}, \algodyn\ consistently handles all queries within the same order of magnitude across datasets, with standard deviations ranging from 1.6\% to 33\% of the average query time (the highest variation is observed on \texttt{ND}).
The performance of \algosta\ and \algoagp\ exhibits similar trends but with slower response times, which are shown in Figure~\ref{fig:aba+} and Figure~\ref{fig:aba}, respectively. \algodpss\ cannot support $a$ and $b$ given on the fly without rebuilding the data structures, and hence it is omitted.

\begin{figure}
\begin{subfigure}{0.35\linewidth}
    \includegraphics[width=1\textwidth]{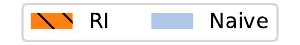}
\end{subfigure}\\
\begin{subfigure}{0.46\linewidth}
    \includegraphics[width=1\textwidth]{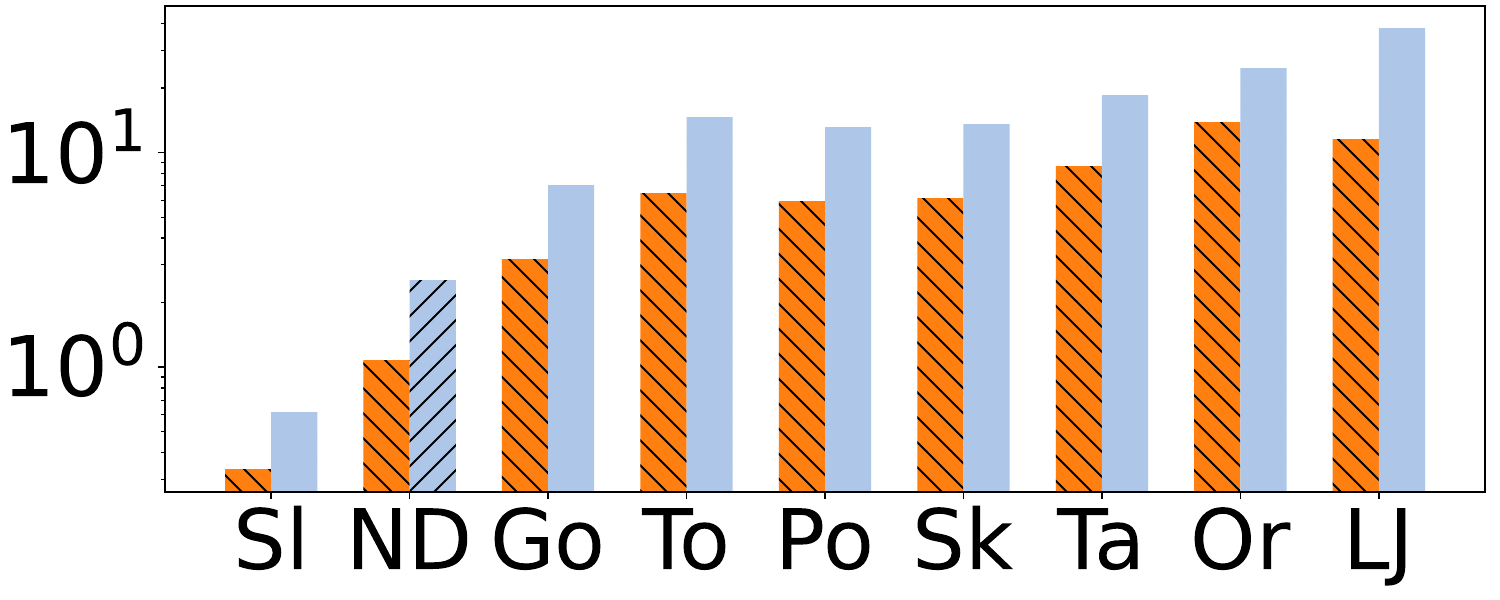}
    \vspace{-5mm}
    \caption{Random $\mathbf{x}$}
    \label{fig:init}
\end{subfigure}
\begin{subfigure}{0.46\linewidth}
    \includegraphics[width=1\textwidth]{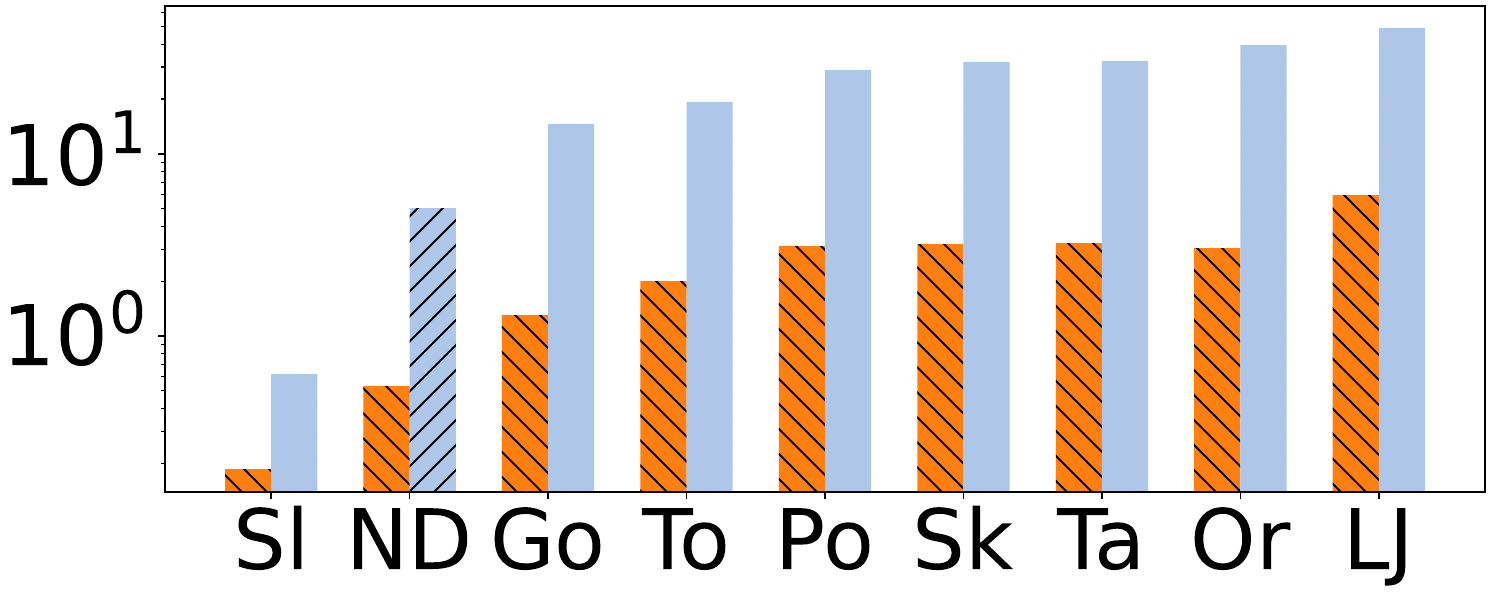}
    \vspace{-5mm}
    \caption{PageRank}
    \label{fig:pr}
\end{subfigure}
\vspace{-2mm}
    \caption{Average initialization time ($\times10^{-3}$ second) }
\vspace{-2mm}
\end{figure}

\begin{figure*}[ht]
    \centering
    \begin{subfigure}{0.15\linewidth}
        \includegraphics[width=\textwidth]{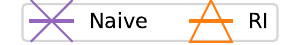}
    \end{subfigure}\\
        \begin{subfigure}{0.18\linewidth}
            \includegraphics[width=\textwidth]{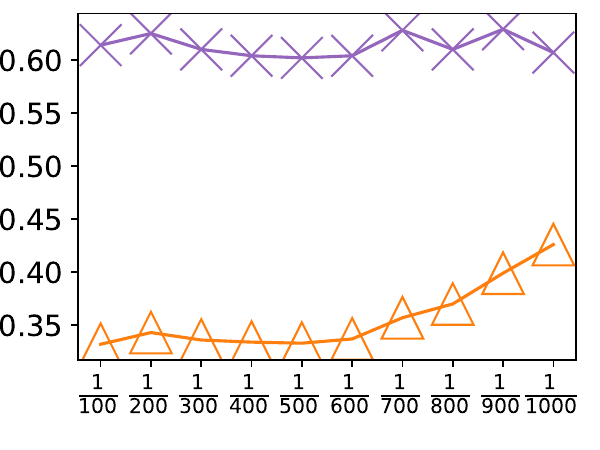}
            \caption{\textbf{Sl}}
        \end{subfigure}
            \begin{subfigure}{0.18\linewidth}
                \includegraphics[width=\textwidth]{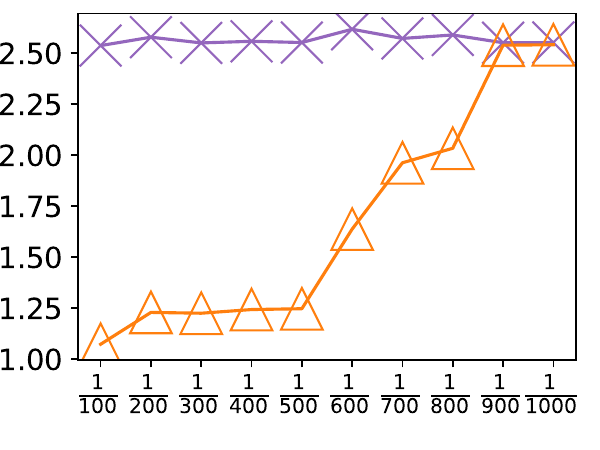}
                \caption{\textbf{ND}}
        \end{subfigure}
            \begin{subfigure}{0.18\linewidth}
                \includegraphics[width=\textwidth]{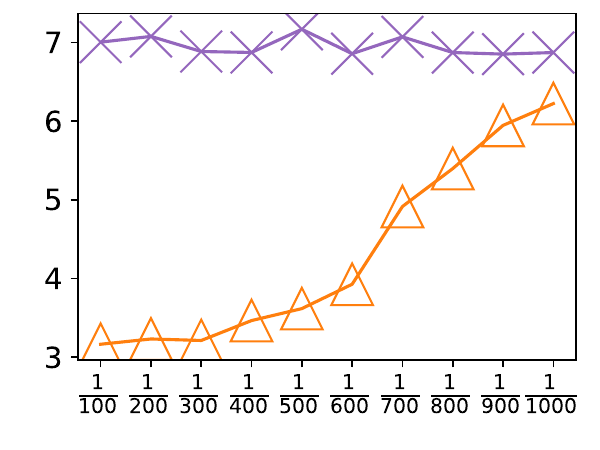}
                \caption{\textbf{Go}}
        \end{subfigure}
            \begin{subfigure}{0.18\linewidth}
                \includegraphics[width=\textwidth]{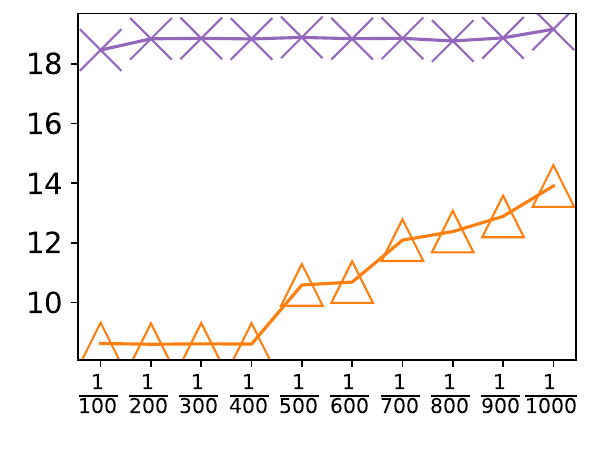}
                \caption{\textbf{To}}
        \end{subfigure}
            \begin{subfigure}{0.18\linewidth}
                \includegraphics[width=\textwidth]{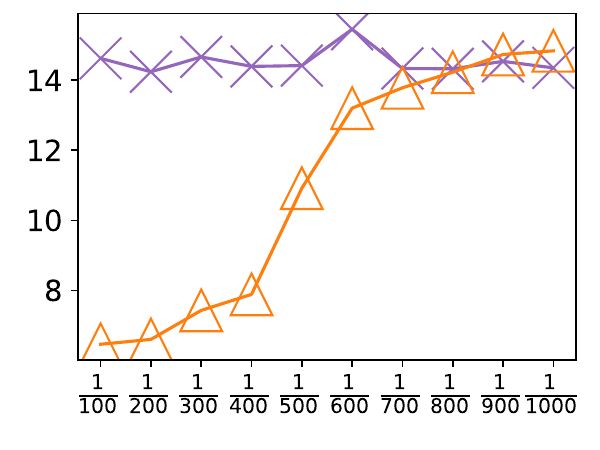}
                \caption{\textbf{Po}}
        \end{subfigure}
     \centering
    
            \begin{subfigure}{0.18\linewidth}
            \includegraphics[width=\textwidth]{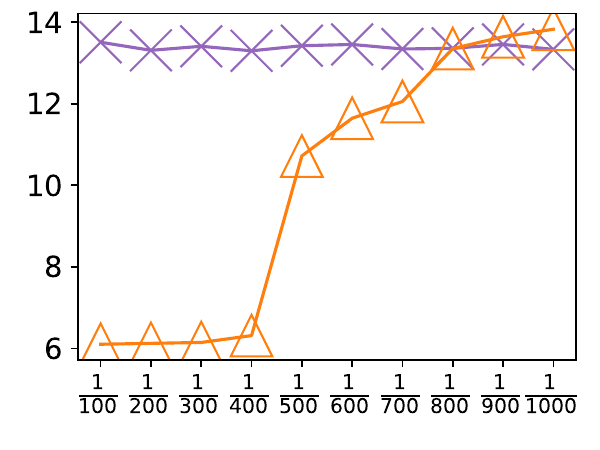}
            \caption{\textbf{Sk}}
        \end{subfigure}
        \begin{subfigure}{0.18\linewidth}
            \includegraphics[width=\textwidth]{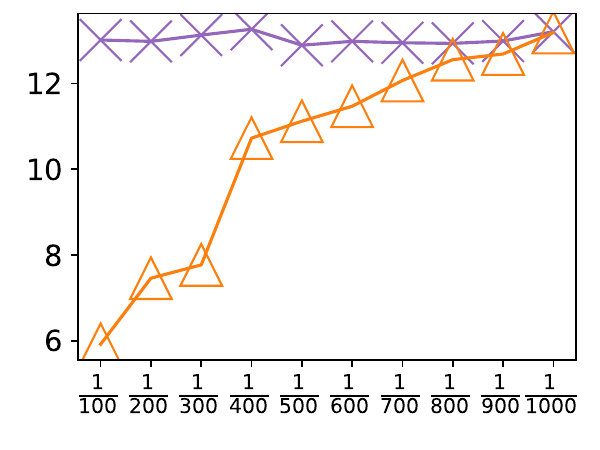}
            \caption{\textbf{Ta}}
        \end{subfigure}
        \begin{subfigure}{0.18\linewidth}
            \includegraphics[width=\textwidth]{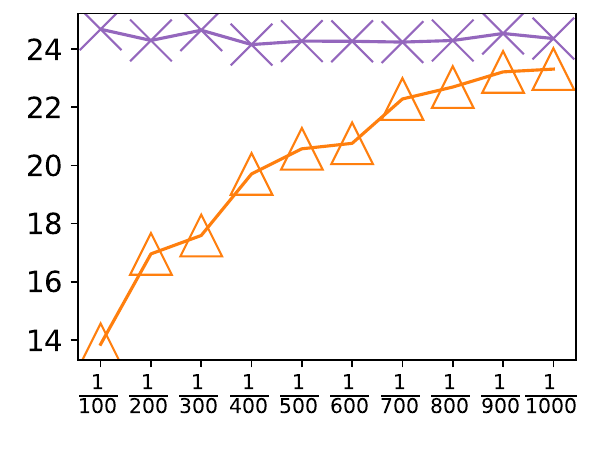}
            \caption{\textbf{Or}}
        \end{subfigure}
        \begin{subfigure}{0.18\linewidth}
            \includegraphics[width=\textwidth]{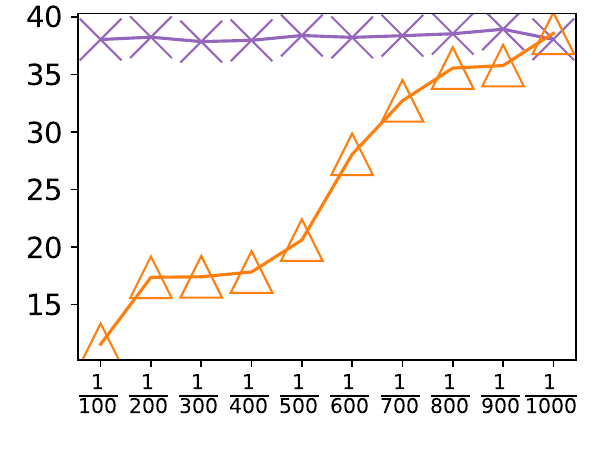}
            \caption{\textbf{LJ}}
        \end{subfigure}
    \caption{Average initialization time ($\times10^{-3}$ second) vs. $\delta$}
    \label{fig:delta_varying_full}
% \vspace{-4mm}
\end{figure*}

\subsection{Initialization Efficiency}
\label{ap:init}
\noindent\textbf{Competitors.} We compare our randomized algorithm (denoted as \emph{RI}) with the naive deterministic initialization approach.
%as mentioned in Section~\ref{sec:rand}. 

\vspace{1mm}
\noindent\textbf{Generate $\mathbf{x}$.} We generate $\mathbf{x}$ in the form of $P$ and $S$ in three steps. First, generate the group number $|P|$ in $[1, \frac{1}{\varepsilon}]$ uniformly at random. Second, for each group, generate group size $|P_i|$ following uniform distribution with $\sum_{i=1}^{|P|}|P_i| = n$. Last, for each group, generate a value $s_i$ following uniform distribution satisfying $\sum_{i=1}^{|P|}s_i = 1$, and then set the value of each group $S_i = \frac{s_i}{|P_i|}$.

For each dataset, we generate 100 random vectors $\mathbf{x}$ and set $\delta = \frac{1}{100}$.
As shown in Figure~\ref{fig:init}, \emph{RI} consistently outperforms the \emph{Naive} approach across all datasets, achieving up to a 3.29$\times$ speedup on \texttt{LJ}.
We also evaluate our algorithm in the context of the PageRank application, where $\mathbf{x} = \left[ \frac{1}{n}, \frac{1}{n}, \dots, \frac{1}{n} \right]^T$.
In this setting, our randomized initialization method exhibits even greater advantages, yielding at least a 2$\times$ speedup on \texttt{Sl} and up to an 11$\times$ improvement on \texttt{Or}, without compromising the approximation quality.
These results highlight the practical benefits of our randomized initialization strategy. 
%{Due to space limit, we defer the results of initialization time against $\eta$ to our technical report~\cite{sourceCode}.}

As shown in Figure~\ref{fig:delta_varying_full}, when $\delta$ decreases, the performance of \emph{RI} degrades as expected, since a tighter error bound necessitates more direct assignments rather than sampling.
Nevertheless, even with $\delta = \frac{1}{1000}$, \emph{RI} still achieves up to a 1.42$\times$ speedup.

\vspace{1mm}
\section{Conclusion}
We revisited the Approximate Graph Propagation problem, which generalizes different graph propagation applications through a unified formulation. For the static case of the problem without graph updates, we derived a tighter query time complexity bound, improving over the state of the art by a factor up to $O(\log^2 n)$ and resulting in a 10$\times$ increase in query efficiency empirically. 
%We then proposed a randomized initialization algorithm that avoids the $O(n)$-time for algorithm initialization, yielding a 2$\times$ speedup over the naive approach. 
For the dynamic case with graph updates, our proposed \algodyn\ algorithm supports flexible query parameters given on-the-fly and attains an amortized $O(1)$ update time, accelerating update efficiency by up to 177$\times$ empirically. %In further work, we will extend our algorithms to weighted graphs. 

\newpage

\bibliographystyle{ACM-Reference-Format}
\bibliography{ref}

@misc{sourceCode,
  author ={Zhuowei Zhao},
  year ={2025},
  title ={Source code:https://github.com/alvinzhaowei/AGP-dynamic},
  url ={https://github.com/alvinzhaowei/AGP-dynamic}
}

@article{gan2024optimal,
  title={Optimal Dynamic Parameterized Subset Sampling},
  author={Gan, Junhao and Umboh, Seeun William and Wang, Hanzhi and Wirth, Anthony and Zhang, Zhuo},
  journal={Proceedings of the ACM on Management of Data},
  volume={2},
  number={5},
  pages={1--26},
  year={2024},
  publisher={ACM New York, NY, USA}
}

@inproceedings{gasteiger2019diffusion,
  title={Diffusion improves graph learning},
  author={Gasteiger, Johannes and Wei{\ss}enberger, Stefan and G{\"u}nnemann, Stephan},
  booktitle = {NeurIPS},
  pages = {13366--13378},
  year={2019}
}

@article{chung2007heat,
  title={The heat kernel as the {PageRank} of a graph},
  author={Chung, Fan},
  journal={Proceedings of the National Academy of Sciences},
  volume={104},
  number={50},
  pages={19735--19740},
  year={2007},
  publisher={National Acad Sciences}
}

@inproceedings{jung2017bepi,
  title={BePI: Fast and memory-efficient method for billion-scale random walk with restart},
  author={Jung, Jinhong and Park, Namyong and Lee, Sael and Kang, U},
  booktitle={SIGMOD},
  pages={789--804},
  year={2017}
}

@inproceedings{Page1999ThePC,
  title={The {PageRank} Citation Ranking: Bringing Order to the Web},
  author={Lawrence Page and Sergey Brin and Rajeev Motwani and Terry Winograd},
  booktitle={WWW},
  year={1999},
  url={https://api.semanticscholar.org/CorpusID:1508503}
}

@inproceedings{wu2019simplifying,
  title={Simplifying graph convolutional networks},
  author={Wu, Felix and Souza, Amauri and Zhang, Tianyi and Fifty, Christopher and Yu, Tao and Weinberger, Kilian},
  booktitle={ICML},
  pages={6861--6871},
  year={2019}
}

@inproceedings{bojchevski2020scaling,
  title={Scaling graph neural networks with approximate pagerank},
  author={Bojchevski, Aleksandar and Gasteiger, Johannes and Perozzi, Bryan and Kapoor, Amol and Blais, Martin and R{\'o}zemberczki, Benedek and Lukasik, Michal and G{\"u}nnemann, Stephan},
  booktitle={KDD},
  pages={2464--2473},
  year={2020}
}

@article{foster2001faster,
  title={A faster katz status score algorithm},
  author={Foster, Kurt C and Muth, Stephen Q and Potterat, John J and Rothenberg, Richard B},
  journal={Computational \& Mathematical Organization Theory},
  volume={7},
  pages={275--285},
  year={2001},
  publisher={Springer}
}

@inproceedings{wang2021approximate,
  title={Approximate graph propagation},
  author={Wang, Hanzhi and He, Mingguo and Wei, Zhewei and Wang, Sibo and Yuan, Ye and Du, Xiaoyong and Wen, Ji-Rong},
  booktitle={KDD},
  pages={1686--1696},
  year={2021}
}

@article{lofgren2013personalized,
  title={Personalized {PageRank} to a target node},
  author={Lofgren, Peter and Goel, Ashish},
  journal={arXiv preprint arXiv:1304.4658},
  year={2013}
}

@inproceedings{gasteiger2018predict,
  title={Predict then propagate: Graph neural networks meet personalized pagerank},
  author={Gasteiger, Johannes and Bojchevski, Aleksandar and G{\"u}nnemann, Stephan},
  booktitle={ICLR},
  year={2019}
}

@misc{snapnets,
    author       = {Jure Leskovec and Andrej Krevl},
    title        = {{SNAP Datasets}: {Stanford} large network dataset collection},
    howpublished = {\url{http://snap.stanford.edu/data}},
    month        = jun,
    year         = 2014
}

@article{achiam2023gpt,
  title={{GPT}-4 technical report},
  author={Achiam, Josh and Adler, Steven and Agarwal, Sandhini and Ahmad, Lama and Akkaya, Ilge and Aleman, Florencia Leoni and Almeida, Diogo and Altenschmidt, Janko and Altman, Sam and Anadkat, Shyamal and others},
  journal={arXiv preprint arXiv:2303.08774},
  year={2023}
}

@article{grattafiori2024llama,
  title={The {Llama} 3 herd of models},
  author={Grattafiori, Aaron and Dubey, Abhimanyu and Jauhri, Abhinav and Pandey, Abhinav and Kadian, Abhishek and Al-Dahle, Ahmad and Letman, Aiesha and Mathur, Akhil and Schelten, Alan and Vaughan, Alex and others},
  journal={arXiv preprint arXiv:2407.21783},
  year={2024}
}

@article{liu2024deepseek,
  title={{DeepSeek-v3} technical report},
  author={Liu, Aixin and Feng, Bei and Xue, Bing and Wang, Bingxuan and Wu, Bochao and Lu, Chengda and Zhao, Chenggang and Deng, Chengqi and Zhang, Chenyu and Ruan, Chong and others},
  journal={arXiv preprint arXiv:2412.19437},
  year={2024}
}

@inproceedings{li2023large,
  title={Large language models in finance: A survey},
  author={Li, Yinheng and Wang, Shaofei and Ding, Han and Chen, Hang},
  booktitle={ACM International Conference on AI in Finance},
  pages={374--382},
  year={2023}
}

@article{chang2024survey,
  title={A survey on evaluation of large language models},
  author={Chang, Yupeng and Wang, Xu and Wang, Jindong and Wu, Yuan and Yang, Linyi and Zhu, Kaijie and Chen, Hao and Yi, Xiaoyuan and Wang, Cunxiang and Wang, Yidong and others},
  journal={ACM Transactions on Intelligent Systems and Technology},
  volume={15},
  number={3},
  pages={1--45},
  year={2024},
  publisher={ACM New York, NY}
}

@article{kasneci2023chatgpt,
  title={{ChatGPT} for good? On opportunities and challenges of large language models for education},
  author={Kasneci, Enkelejda and Se{\ss}ler, Kathrin and K{\"u}chemann, Stefan and Bannert, Maria and Dementieva, Daryna and Fischer, Frank and Gasser, Urs and Groh, Georg and G{\"u}nnemann, Stephan and H{\"u}llermeier, Eyke and others},
  journal={Learning and Individual Differences},
  volume={103},
  pages={102274},
  year={2023},
  publisher={Elsevier}
}

@article{thirunavukarasu2023large,
  title={Large language models in medicine},
  author={Thirunavukarasu, Arun James and Ting, Darren Shu Jeng and Elangovan, Kabilan and Gutierrez, Laura and Tan, Ting Fang and Ting, Daniel Shu Wei},
  journal={Nature Medicine},
  volume={29},
  number={8},
  pages={1930--1940},
  year={2023},
  publisher={Nature Publishing Group US New York}
}

@inproceedings{fan2024survey,
  title={A survey on {RAG} meeting {LLMs}: Towards retrieval-augmented large language models},
  author={Fan, Wenqi and Ding, Yujuan and Ning, Liangbo and Wang, Shijie and Li, Hengyun and Yin, Dawei and Chua, Tat-Seng and Li, Qing},
  booktitle={KDD},
  pages={6491--6501},
  year={2024}
}

@inproceedings{lewis2020retrieval,
  title={Retrieval-augmented generation for knowledge-intensive NLP tasks},
  author={Lewis, Patrick and Perez, Ethan and Piktus, Aleksandra and Petroni, Fabio and Karpukhin, Vladimir and Goyal, Naman and K{\"u}ttler, Heinrich and Lewis, Mike and Yih, Wen-tau and Rockt{\"a}schel, Tim and others},
  booktitle={NeurIPS},
  pages={9459--9474},
  year={2020}
}

@article{siriwardhana2023improving,
  title={Improving the domain adaptation of retrieval augmented generation ({RAG}) models for open domain question answering},
  author={Siriwardhana, Shamane and Weerasekera, Rivindu and Wen, Elliott and Kaluarachchi, Tharindu and Rana, Rajib and Nanayakkara, Suranga},
  journal={Transactions of the Association for Computational Linguistics},
  volume={11},
  pages={1--17},
  year={2023},
  publisher={MIT Press One Broadway, 12th Floor, Cambridge, Massachusetts 02142, USA~…}
}

@article{edge2024local,
  title={From local to global: A graph {RAG} approach to query-focused summarization},
  author={Edge, Darren and Trinh, Ha and Cheng, Newman and Bradley, Joshua and Chao, Alex and Mody, Apurva and Truitt, Steven and Metropolitansky, Dasha and Ness, Robert Osazuwa and Larson, Jonathan},
  journal={arXiv preprint arXiv:2404.16130},
  year={2024}
}

@inproceedings{gutierrez2024hipporag,
  title={{HippoRAG}: Neurobiologically inspired long-term memory for large language models},
  author={Guti{\'e}rrez, Bernal Jim{\'e}nez and Shu, Yiheng and Gu, Yu and Yasunaga, Michihiro and Su, Yu},
  booktitle={NeurIPS},
  year={2024}
}

@article{huang2025ket,
  title={{KET-RAG}: A Cost-Efficient Multi-Granular Indexing Framework for Graph-RAG},
  author={Huang, Yiqian and Zhang, Shiqi and Xiao, Xiaokui},
  journal={arXiv preprint arXiv:2502.09304},
  year={2025}
}

@inproceedings{wang2024knowledge,
  title={Knowledge graph prompting for multi-document question answering},
  author={Wang, Yu and Lipka, Nedim and Rossi, Ryan A and Siu, Alexa and Zhang, Ruiyi and Derr, Tyler},
  booktitle={AAAI},
  pages={19206--19214},
  year={2024}
}

@article{wu2024medical,
  title={Medical graph {RAG}: Towards safe medical large language model via graph retrieval-augmented generation},
  author={Wu, Junde and Zhu, Jiayuan and Qi, Yunli and Chen, Jingkun and Xu, Min and Menolascina, Filippo and Grau, Vicente},
  journal={arXiv preprint arXiv:2408.04187},
  year={2024}
}

@inproceedings{bringmann2013exact,
  title={Exact and efficient generation of geometric random variates and random graphs},
  author={Bringmann, Karl and Friedrich, Tobias},
  booktitle={International Colloquium on Automata, Languages, and Programming},
  pages={267--278},
  year={2013}
}

@book{ross2020first,
  title={A first course in probability},
  author={Ross, Sheldon M},
  year={2020},
  publisher={Pearson}
}

@book{trefethen2022numerical,
  title={Numerical inear algebra},
  author={Trefethen, Lloyd N. and Bau, David},
  year={2022},
  publisher={SIAM}
}

@article{farach2015exact,
  title={Exact sublinear binomial sampling},
  author={Farach-Colton, Mart{\'\i}n and Tsai, Meng-Tsung},
  journal={Algorithmica},
  volume={73},
  pages={637--651},
  year={2015},
  publisher={Springer}
}

@inproceedings{gyongyi2004combating,
  title={Combating web spam with TrustRank},
  author={Gy{\"o}ngyi, Zolt{\'a}n and Garcia-Molina, Hector and Pedersen, Jan},
  booktitle={VLDB},
  pages={576--587},
  year={2004},
  organization={VLDB Endowment}
}

@article{garcia2022binomial,
  title={The binomial distribution: Historical origin and evolution of its problem situations},
  author={Garc{\'\i}a-Garc{\'\i}a, Jaime Israel and Fern{\'a}ndez Coronado, Nicol{\'a}s Alonso and Arredondo, Elizabeth H and Imilp{\'a}n Rivera, Isaac Alejandro},
  journal={Mathematics},
  volume={10},
  number={15},
  pages={2680},
  year={2022},
  publisher={MDPI}
}

@inproceedings{kuhl2017history,
  title={History of random variate generation},
  author={Kuhl, Michael E},
  booktitle={WSC},
  pages={231--242},
  year={2017},
  organization={IEEE}
}

@article{kachitvichyanukul1988binomial,
  title={Binomial random variate generation},
  author={Kachitvichyanukul, Voratas and Schmeiser, Bruce W},
  journal={Communications of the ACM},
  volume={31},
  number={2},
  pages={216--222},
  year={1988},
  publisher={ACM New York, NY, USA}
}

\end{document}